\documentclass[12pt,pdftex,noinfoline]{imsart}

\RequirePackage[OT1]{fontenc}
\usepackage{amsthm,amsmath,amsfonts,natbib,mathtools,amssymb}
\RequirePackage{hypernat}
\usepackage[ruled,section]{algorithm}
\usepackage{bbm}
\usepackage{algorithmic}
\usepackage{graphicx}
\usepackage{verbatim}
\usepackage{times}
\usepackage[T1]{fontenc}
\usepackage[text={6.0in,8.6in},centering]{geometry}
\usepackage{subfigure}
\usepackage{graphicx}

\usepackage[usenames,dvipsnames]{xcolor}
\usepackage[pdftex,pdfborderstyle={/S/U/W 1},colorlinks]{hyperref}
\usepackage[all]{hypcap}
\hypersetup{citecolor=MidnightBlue}
\hypersetup{linkcolor=Black}
\hypersetup{urlcolor=MidnightBlue}

\usepackage{enumerate}
\usepackage{fullpage}
\usepackage{tikz}
\usetikzlibrary{decorations.pathreplacing}
\usetikzlibrary{arrows,positioning}

\let\hat\widehat
\let\tilde\widetilde
\let\bar\overline

\def\Var{\text{Var}} % the symbol Var for covariance used the sans serif letter

\def\given{\,|\,}
\def\W{\mathcal{W}}
\def\eqdef{\equiv}
\def\indicator{\mathbbm{1}}
\def\ones{\mathbbm{1}}
\def\sfrac#1#2{{#1}/{#2}}
\def\tr{\mathop{\rm tr}}

\def\text#1{\mbox{\rm #1}}

\newcommand{\argmin}{\mathop{\rm argmin}}

\newcommand{\supp}{{\rm supp}}

\theoremstyle{plain}
\newtheorem{theorem}{Theorem}[section]
\newtheorem{proposition}[theorem]{Proposition}
\newtheorem{lemma}[theorem]{Lemma}

\theoremstyle{remark}

\theoremstyle{definition}

\def\P{{\mathbb P}}
\def\E{{\mathbb E}}
\def\supp{\mathop{\text{supp}\kern.2ex}}
\def\argmin{\mathop{\text{arg\,min}\kern.2ex}}

\def\sfrac#1#2{{#1}/{#2}}

\newcommand{\R}{\mathbb{R}}

\def\shape#1{
  \lower5pt\hbox{
  \hskip-7pt
  \tikzset{circ/.style={circle, draw, fill=black, scale=.2}}
  \begin{tikzpicture}[semithick,scale=.3]
  \node (l1) at (0,.866) [circ]{};
  \node (l2) at (1,.866) [circ]{};
  \node (l3) at (0.5,0) [circ]{};
  #1
  \end{tikzpicture}
  \hskip-8pt}
}

\def\edgeshape{
  \raise2pt\hbox{
  \hskip-8pt
  \tikzset{circ/.style={circle, draw, fill=black, scale=.2}}
  \begin{tikzpicture}[semithick,scale=.3]
  \node (l1) at (0,.866) [circ]{};
  \node (l2) at (1,.866) [circ]{};
  \draw[-] (l1) to node [auto] {} (l2);
  \end{tikzpicture}
  \hskip-4pt}
}

\def\veeshape{
  \lower5pt\hbox{
  \hskip-7pt
  \tikzset{circ/.style={circle, draw, fill=black, scale=.2}}
  \begin{tikzpicture}[semithick,scale=.3]
  \node (l1) at (0,.866) [circ]{};
  \node (l2) at (1,.866) [circ]{};
  \node (l3) at (0.5,0) [circ]{};
  \draw[-,color=white] (l1) to node [auto] {} (l2);
  \draw[-] (l1) to node [auto] {} (l3);
  \draw[-] (l2) to node [auto] {} (l3);
  \end{tikzpicture}
  \hskip-8pt}
}

\def\triangleshape{
  \lower5pt\hbox{
  \hskip-7pt
  \tikzset{circ/.style={circle, draw, fill=black, scale=.2}}
  \begin{tikzpicture}[semithick,scale=.3]
  \node (l1) at (0,.866) [circ]{};
  \node (l2) at (1,.866) [circ]{};
  \node (l3) at (0.5,0) [circ]{};
  \draw[-] (l1) to node [auto] {} (l3);
  \draw[-] (l2) to node [auto] {} (l3);
  \draw[-] (l1) to node [auto] {} (l2);
  \end{tikzpicture}
  \hskip-8pt}
}

\numberwithin{equation}{section}
\graphicspath{{./figs/}}

\begin{document}

\begin{frontmatter}
\title{\Large Testing for Global Network Structure \vskip-5pt Using Small Subgraph Statistics}
%\runtitle{Testing Network Structure Using Small Subgraph Statistics}
\runauthor{Gao and Lafferty}

\begin{aug}
\vskip10pt
\author{\fnms{Chao}
  \snm{Gao${}^{1}$}\ead[label=e1]{chaogao@galton.uchicago.edu}}
 \and 
\author{\fnms{John}
  \snm{Lafferty${}^{2}$}\ead[label=e4]{john.lafferty@yale.edu}}
\vskip10pt
\end{aug}

\begin{abstract}
We study the problem of testing for community structure in networks
using relations between the observed frequencies of small subgraphs.
We propose a simple test for the existence of communities based only
on the frequencies of three-node subgraphs.  The test statistic is
shown to be asymptotically normal under a null assumption of no
community structure, and to have power approaching one under a
composite alternative hypothesis of a degree-corrected stochastic
block model.  We also derive a version of the test that applies to
multivariate Gaussian data. Our approach achieves near-optimal
detection rates for the presence of community structure, in regimes
where the signal-to-noise is too weak to explicitly estimate the
communities themselves, using existing computationally efficient
algorithms. We demonstrate how the method can be effective for
detecting structure in social networks, citation networks for
scientific articles, and correlations of stock returns between
companies on the S\&P 500.
\end{abstract}
 
\vskip20pt 
\end{frontmatter}

\maketitle

\vskip10pt

\footnotetext[1]{Department of Statistics, University of Chicago;
  email: chaogao@galton.uchicago.edu}
\footnotetext[2]{Department of Statistics and Data Science, Yale
  University; email: john.lafferty@yale.edu}

\def\ER{Erd\H{o}s-R\'{e}nyi}
\def\EZ{Erd\H{o}s-Zuckerberg}

\section{Introduction}

The statistical properties of graphs and networks have been
intensively studied in recent years, resulting in a rich and detailed
body of knowledge on stochastic graph models. Examples include graphons
and the stochastic block model \citep{holland83,lovasz12}, preferential
attachment models \citep{price76,barabasi99}, and
other generative network models \citep{bollobas01}. This work often seeks to model the
network structures observed in ``naturally occurring'' settings, such
as social media.  The focus has been on developing simple models that
can be rigorously studied, while still capturing some of the phenomena
observed in actual data. Related work has developed procedures to find
structure in networks, for example using spectral algorithms for finding
communities \citep{rohe2011spectral,jin2015,arias-castro2014}.
Another line of research has studied 
estimation and detection of signals on graphs where
the structure of the signal is exploited to develop efficient
procedures \citep{dfs16,arias-castro2011}.

In this work our focus is on understanding how global structural
properties of networks might be inferred from purely local
properties. In the absence of a probability model to generate the
graph, this is a classical mathematical topic. For example, convex
polyhedra and planar graphs satisfy the invariant $\chi = 2$, where
the Euler-Poincar\'e characteristic $\chi = V - E + F$ is defined in
terms of the number of vertices, edges and faces of the polyhedron or
associated planar graph.  More general relations between local
structure and global invariants lie at the heart of combinatorics and
algebraic topology; topological data analysis is the study of such
relations under a data sampling model.  In this paper we study how the
presence of communities in a network is related to relations between
the densities of small subgraphs, such as edges, vees, and
triangles.

Our investigation was in part inspired by the work of
\cite{ugander2013subgraph}, who present striking data on the empirical
distributions of 3-node and 4-node subgraphs of Facebook friend
networks, comparing them to the distributions that would be obtained
under an \ER{} model.  In particular, it is noted that
the small subgraph frequencies of the Facebook subnetworks can be
close to the corresponding probabilities under an \ER{}
model, even though the subnetworks are expected to exhibit community
structure. However, the subgraph frequencies are not arbitrary. In
fact, the global graph structure places purely combinatorial
restrictions on the subgraph probabilities, sometimes called
homomorphism constraints \citep{razborov2008}.  The interplay between
the structural properties and homomorphism constraints is discussed by
\cite{ugander2013subgraph}, who pose the broad research question
``What properties of social graphs are `social' properties and what
properties are `graph' properties?''  Their work develops two
complementary methods to shed light on this question. First, they
propose a generative model that extends the \ER{} model
and better matches the empirical data. Second, they develop methods to
bound the homomorphism constraints that determine the feasible space
of subgraph probabilities.

\begin{figure*}[!htbp]
\begin{center}
\begin{tabular}{cc}
\includegraphics[width=.45\textwidth]{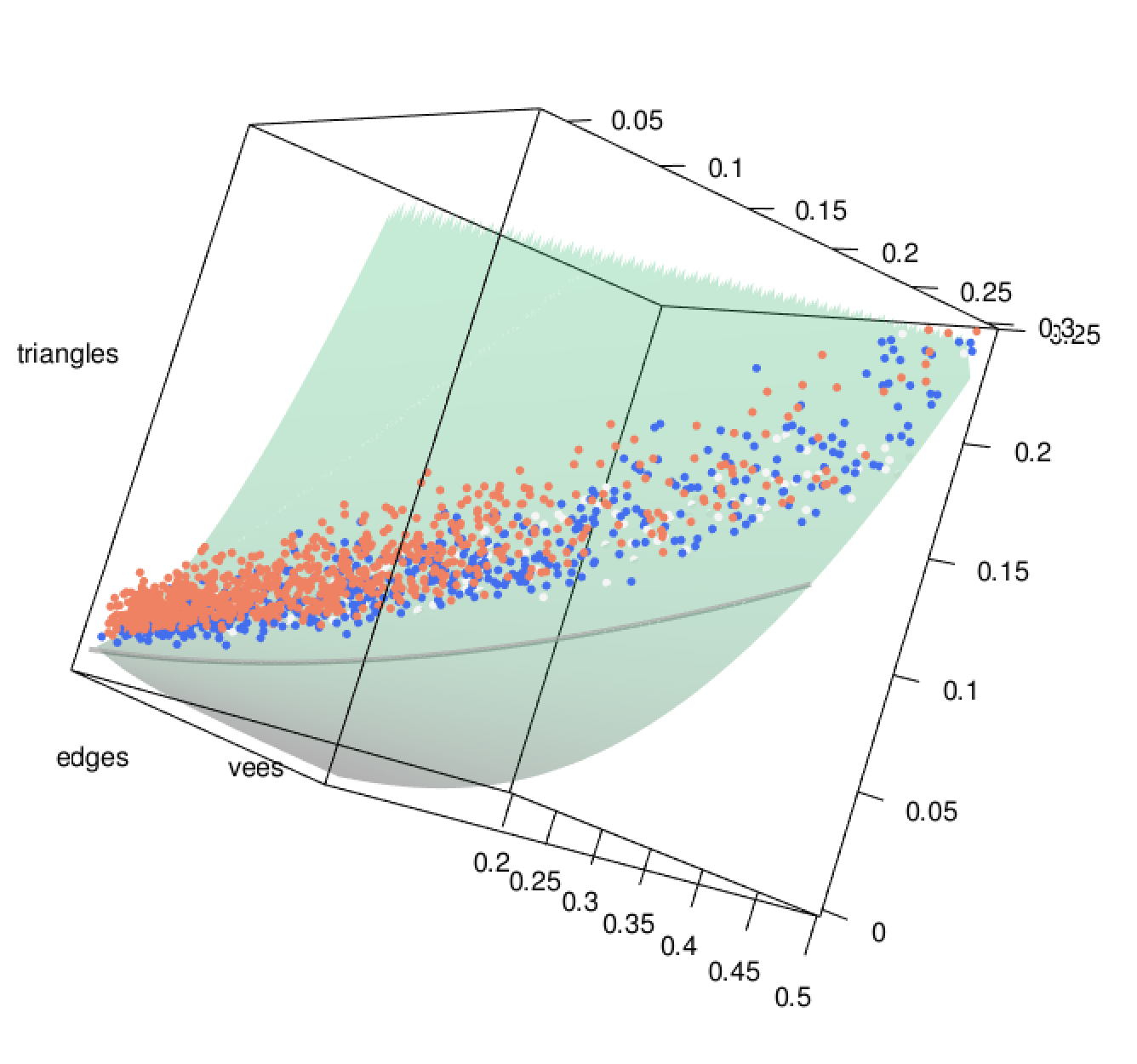} &
\includegraphics[width=.45\textwidth]{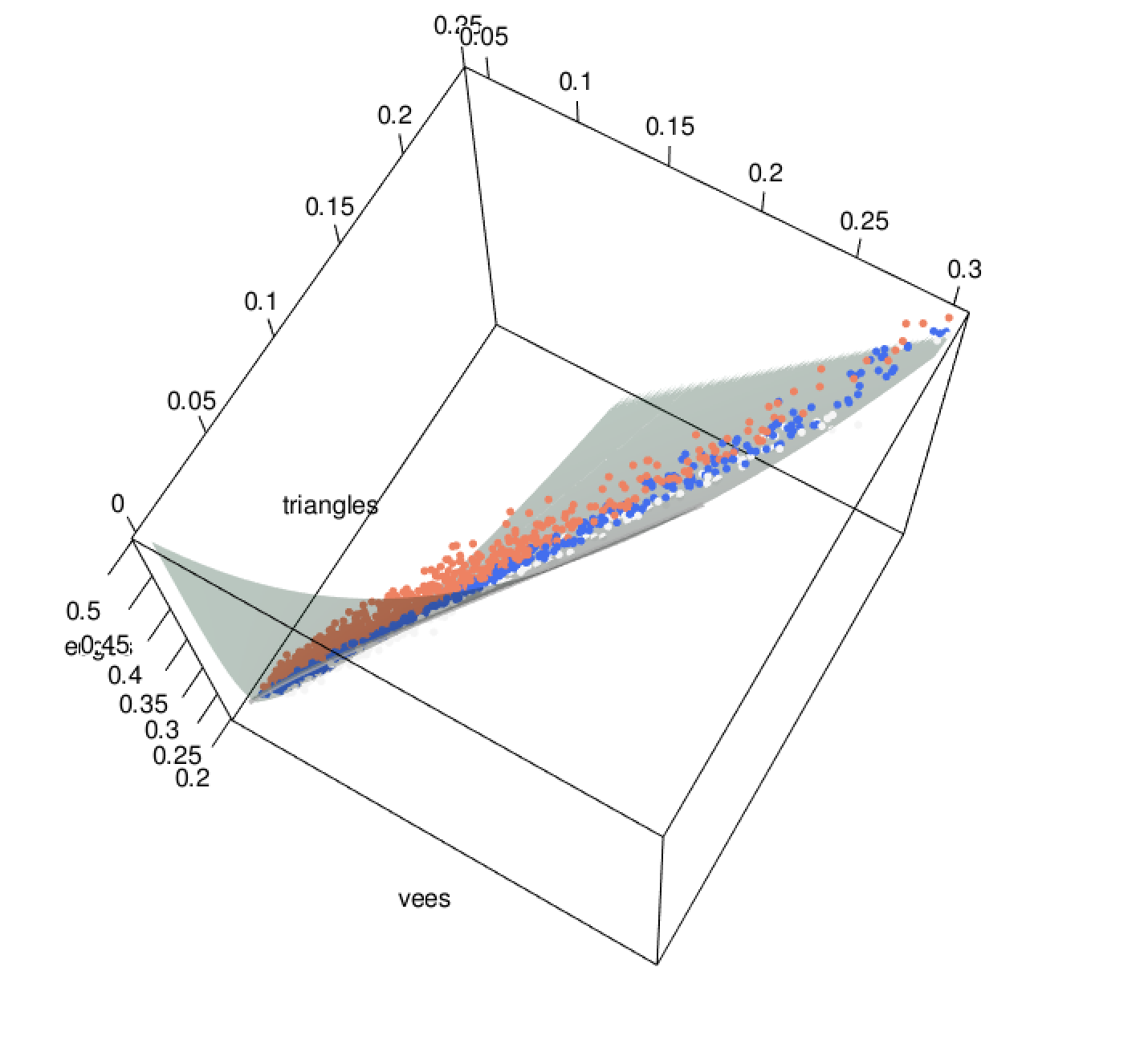} \\
\includegraphics[width=.35\textwidth]{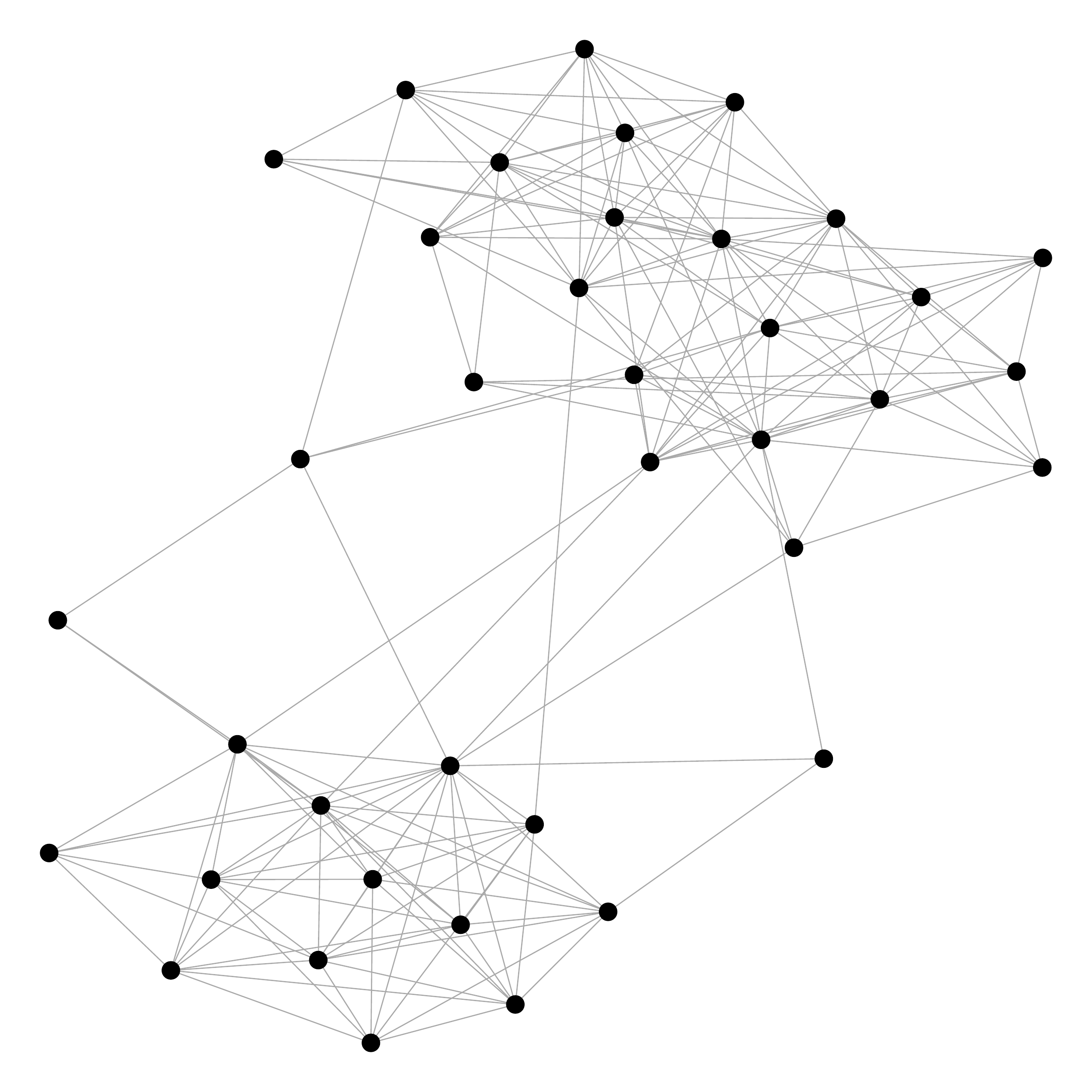} &
\includegraphics[width=.35\textwidth]{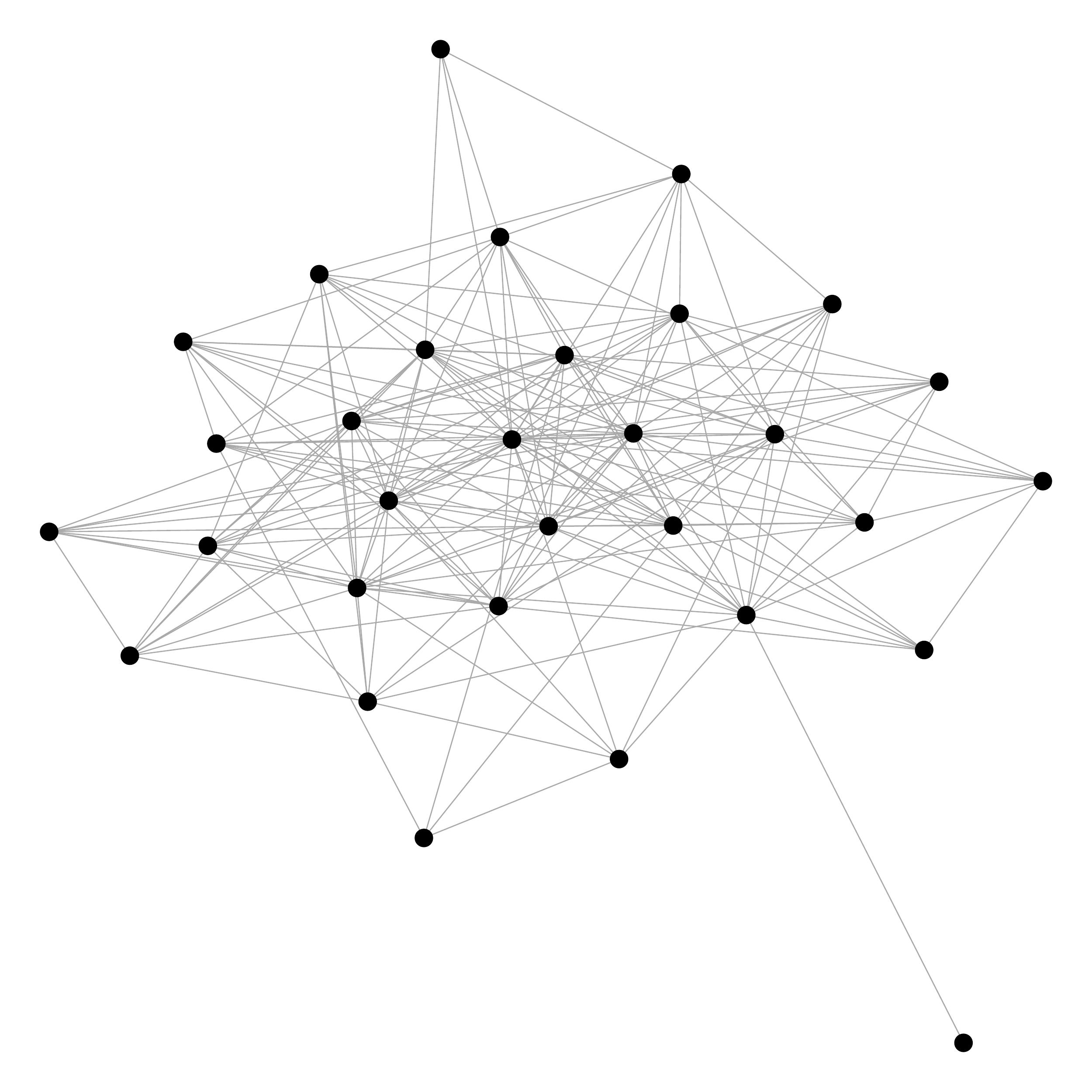}
\end{tabular}
\end{center}
\caption{Top row: Two different views of 3-dimensional plots of the
  triangle, vee, and edge densities of a collection of Facebook
  subnetworks. Each point represents a single subnetwork, the
  induced graph of the friends of a single Facebook user (ego), in the
  fb100 data. Each subnetwork in this collection has between 20 and 40
  nodes. The coordinates $(\hat E, \hat V, \hat T)$ of a given point
  are the relative frequencies of edges, vees, and triangles in the
  network.  The points are colored according to the p-value of the EZ
  score test: red, smaller than $10^{-5}$; blue, in the range
  $(10^{-5}, 10^{-2})$; white, larger than $10^{-2}$. The shaded
  surface indicates the subfamily of degree-corrected stochastic block
  models for which $\chi_{ev} = T - (V/E)^3 = 0$, and no community
  structure is present. The curve on the surface, visible in the left
  plot, corresponds to the subfamily of \ER{} graphs where $(E,V,T) =
  (p, p^2, p^3)$. Qualitatively, we see that even the points that lie
  relatively close to the \EZ{} surface $\chi_{ez}=0$, are far from the
  \ER{} curve, which may be attributed to degree heterogeneity. Bottom
  row: Typical Facebook subnetworks with small and large p-values. The
  left graph, with an EZ score of 8.25 and p-value of $10^{-16}$, has
  clear community structure. The right graph, with an EZ score of
  0.498 and p-value of 0.62, exhibits no community structure.}
\label{fig:surface}
\end{figure*}

In the present paper we take a statistical approach to
distinguishing graphs with community structure from unstructured random graphs using only small
subgraph frequencies, framing the problem in terms of
statistical testing.  The starting point for our analysis is 
the degree-corrected
stochastic block model, a simple generative model for random networks
that captures two salient properties that are observed empirically in
social networks and other data---community structure and degree
heterogeneity. While the precise specification of the model is deferred to
the following section, it is characterized by a few key parameters, including 
the number of communities $k\geq 1$, the within-community connectivity
probability $a$, and the between-community connectivity probability $b$.
For random networks, we define
$$ \chi_{ez}= T - \left(\frac{V}{E}\right)^3,$$
where now $E$, $V$, and $T$ are the 
expected densities of edges (\edgeshape), vees (\veeshape),
and triangles (\triangleshape) in the graph.
A simple calculation, which we present in the following section, shows that 
under the degree-corrected stochastic block model,
\begin{equation}
\chi_{ez} = (k-1)\left(\frac{a-b}{k}\right)^3.
\label{ezq}
\end{equation}
We thus see that $\chi_{ez}=0$ if and only if the network has no
communities, under the assumed model, meaning that $a=b$ or $k=1$. In particular, an \ER{} random
graph with edge probability $p$ satisfies
$$
\chi_{ez} = p^3 - \left(\frac{p^2}{p}\right)^3 = 0.
$$ Since $\chi_{ez}=0$ distinguishes unstructured random networks from
those with community structure, within the large class of
degree-corrected stochastic block models, we refer to $\chi_{ez}$ as
the \textit{\EZ{} characteristic}, or the \textit{EZ characteristic},
for short.

Starting from the simple relation in \eqref{ezq}, this paper explores
the mathematical, statistical, and empirical properties of the \EZ{}
characteristic as a test for global community structure. We find that
this simple functional has remarkable properties, both theoretically
and empirically. We develop a testing framework for the null
hypothesis corresponding to $\chi_{ez}=0$, and analyze its power and
scaling behavior.  Empirically, we find that the test is effective on
the types of Facebook subnetworks studied by
\cite{ugander2013subgraph}. In particular, graphs with small p-values
under this test exhibit clear community structure, while unstructured
subnetworks typically lie close to the cubic surface defined by the
invariant $\chi_{ez}=0$, with correspondingly large p-values; see Figure~1.
However, even the unstructured subnetworks are relatively far
from the curve in this surface traced out by the \ER\ subfamily, and a
closely related test based on an \ER{} null hypothesis is ineffective.

In related work, \cite{mossel2012stochastic} prove 
a Poisson limit law for counts of cycles in stochastic block models,
and \cite{maugis2017} consider the use of small
subgraph counts to test whether a collection of networks is drawn from
a known graphon null model. \cite{bubeck2014} study tests based on
signed triangles for distinguishing Erd\H{o}s-R\'{e}nyi graphs from
random geometric graphs in the dense
regime. \cite{banerjee2016contiguity} and \cite{banerjee2017} study
tests based on signed circles and establish a relation to the
likelihood ratio statistic for stochastic block models.
\cite{ambroise:2012} and \cite{allman:2011}
study moment estimators for the parameters
of stochastic block models; these estimators
were used recently by \cite{kloumann:2017} in
the context of seed set expansion and node
ranking in personalized search.

In the following section we provide further detail on the relation
\eqref{ezq}, which is the key equation in our testing approach, and
then develop a central limit theorem for this characteristic 
based on empirical estimates of the densities of edges,
vees and triangles. The power and scaling behavior of the resulting
test is analyzed in Section~\ref{sec:ezpower}. 
In Section~\ref{sec:sbm} we comment on a related test when the null
model is chosen to be an Erd\H{o}s-R\'{e}nyi model rather than a
configuration model, the essential difference being
degree heterogeneity. The discussion in this section sheds light
on the empirical findings shown in Figure~\ref{fig:surface},
where the Facebook subnetworks are relatively far from the
\ER{} subfamily, but close to the surface defined by $\chi_{ez}=0$.
  Section~\ref{sec:experiments} gives illustrations of
our testing framework on Facebook social networks, citations from
statistics journal articles, and stock returns of companies on the
S\&P 500.  In each of these settings, we find that the \EZ{} test
gives interesting and interpretable results, and is effective at
identifying community structure using only the local information 
available in two and three node subgraph statistics. Finally, in the supplementary material, we give some extensions of the \EZ{} characteristic and present the proofs of the results of the paper.
Section~\ref{sec:nbhds} introduces a rigorous framework of tests for neighborhood graphs.
Section~\ref{sec:gauss} considers correlation structures for
multivariate Gaussian data, and derives an analogous test for community
structure under this model.

\section{The Erd\H{o}s-Zuckerberg Test}
\label{sec:eztest}

\def\mip#1#2{\left\langle #1,\, #2\right\rangle}

One of the most popular network models
of community structure is the degree-corrected stochastic block model
(DCBM) \citep{dasgupta2004spectral,karrer2011stochastic}. Under
this model, a random adjacency matrix $A$ is generated according
to $A_{ij} \given \theta_{ij}\sim\text{Bernoulli}(\theta_{ij})$ independently
for each edge $(i,j)$, where the mean parameters $\theta_{ij}$ have 
a blockwise low rank structure that models degree heterogeneity and
community structure. For community structure, 
latent variables $Z_i \sim \text{Uniform}([k])$ are generated independently for
each node $i$, where the integer $k\geq 1$ is the number of
communities. For degree heterogeneity, variables $W_i\sim \W$ are
generated independently for each node from a distribution $\W$; the value $W_i$ can be thought of as a measure of the
``sociability'' of node $i$. Conditional on $Z$ and $W$, the mean parameters
$\theta_{ij}$ are then given by
\begin{equation}
\theta_{ij} \given W,Z =\begin{cases}
W_i W_j a, & Z_i=Z_j, \\
W_i W_j b, & Z_i\neq Z_j,
\end{cases}\label{eq:def-dcbm}
\end{equation}
where $a$ is the within-community connectivity probability, and $b$ is
the between-community connectivity probability. 

Thus, the distribution of $\{A_{ij}\}_{1\leq i<j\leq n}$ is fully
determined by the parameters $a$, $b$, $k$ and the distribution
$\mathcal{W}$. 
%As is pointed out in
%\cite{karrer2011stochastic,jin2015fast,gao2016community} among others,
%
The parameterization \eqref{eq:def-dcbm} is not identifiable, since
the model is invariant to multiplying $a$ and $b$ by some arbitrary number $t>0$,
and dividing each $W_i$ by $\sqrt{t}$.
Thus, without loss of generality, we introduce the constraint
\begin{equation}
\E(W^2) =1,\label{eq:2nd-w}
\end{equation} 
for the distribution $\mathcal{W}$, so that the parameters $a$ and $b$ are uniquely determined.

The problem of community detection in the setting of the DCBM has been well studied in the
literature
\citep{lei2015consistency,gulikers2015impossibility,zhao2012consistency,jin2015,chen2015convexified}.
\cite{gao2016community} derive the minimax rate of the problem with respect to the Hamming
loss. All of this work assumes there exists a clustering structure in the model and the
number of clusters $k$ is given. In the current paper, we shift the focus to testing for
community structure, without estimating $k$ or the clusters themselves.  Under the DCBM,
the lack of such structure is equivalent to $k=1$ or $a=b$.

The following result is central to our testing procedure
and analysis.
\begin{proposition}
Define the population
edge, vee, and triangle probabilitities by
\begin{eqnarray}
\label{eq:E} E &=& \P(A_{12}=1),\\
\label{eq:V} V &=& \P(A_{12}A_{13}=1),\\
\label{eq:T} T &=& \P(A_{12}A_{13}A_{23}=1).
\end{eqnarray}
Then under the degree-corrected stochastic block model \eqref{eq:def-dcbm}
and assuming the identifiability condition
\eqref{eq:2nd-w},  we have
\begin{eqnarray}
\label{eq:E+} E &=&  (\E W)^2\left(\frac{1}{k}a+\frac{k-1}{k}b\right),\\
\label{eq:V+} V &=& (\E W)^2\left(\frac{1}{k}a+\frac{k-1}{k}b\right)^2, \\
\label{eq:T+} T &=& \frac{1}{k^2}a^3+\frac{3(k-1)}{k^2}ab^2+\frac{(k-1)(k-2)}{k^2}b^3.
\end{eqnarray}
From these relations, it follows that
\begin{equation}
\chi_{ez} \eqdef T-\left(\frac{V}{E}\right)^3=\frac{(k-1)(a-b)^3}{k^3}.\label{eq:EZ}
\end{equation}
\label{prop:EZ}
\end{proposition}

\begin{proof}
The identities follow from direct calculation. For example, because
$\{W_i\}$ and $\{Z_i\}$ are sampled independently, we have that
\begin{align*}
E = \E(A_{12}) 
&= \E\bigl(\E(A_{12}\given Z_1, Z_2, W_1, W_2)\bigr) \\
&= \E\left(W_1 W_2 \bigl(a \indicator(Z_1=Z_2) + b \indicator(Z_1 \neq Z_2\bigr)\right)\\
&= \E\left(W_1 W_2 \Bigl(a\frac{1}{k} + b \frac{k-1}{k}\Bigr)\right)\\
&= (\E W)^2 \Bigl(a\frac{1}{k} + b \frac{k-1}{k}\Bigr).
\end{align*}
Equations \eqref{eq:V+} and \eqref{eq:T+} are derived similarly, using the constraint
$\E(W^2)=1$. The relation \eqref{prop:EZ} is then seen to hold after some algebra.
\end{proof}

The relation \eqref{eq:EZ} implies that $k=1$ or $a=b$ if and only if
$\chi_{ez}=0$, which characterizes whether or not the network has community structure.
When $\chi_{ez} = 0$, the model is reduced to
$\theta_{ij} = a\/W_iW_j$, which is also recognized as the configuration model
\citep{van2016random},
and closely related to the Chung-Lu model of random graphs with
expected degrees \citep{chunglu2002}.
If $\chi_{ez} > 0$, then the network has an assortative clustering structure;
such a network will induce more triangles compared with the configuration model.
Conversely, the network will have disassortative clustering structure 
if $\chi_{ez} < 0$, in which case there will be fewer triangles. 
We see both types of structure in our empirical studies, described below.

\subsection{The EZ test}

We now develop a statistical test for the null hypothesis $H_0: \chi_{ez}=0$.
The empirical versions of relations \eqref{eq:E}--\eqref{eq:T} are
\begin{eqnarray*}
\hat{E} &=& \frac{1}{{n\choose 2}}\sum_{1\leq i<j\leq n}A_{ij}, \\
\hat{V} &=& \frac{1}{{n\choose 3}}\sum_{1\leq i<j<l\leq n}\frac{A_{ij}A_{il}+A_{ij}A_{jl}+A_{il}A_{jl}}{3}, \\
\hat{T} &=& \frac{1}{{n\choose 3}}\sum_{1\leq i<j<l\leq n}A_{ij}A_{il}A_{jl}.
\end{eqnarray*}
Therefore, we can reject the null hypothesis once the magnitude of the
plug-in test statistic 
\begin{equation}\hat\chi_{ez} \eqdef
  \hat{T}-\left(\sfrac{\hat{V}}{\hat{E}}\right)^3
\end{equation}
passes a threshold. The following result gives the asymptotic distribution of
$\hat\chi_{ez}$, and allows us to set the threshold and significance
level of the test.

\begin{theorem}\label{thm:main}
Assume $\E W^4=O(1)$ and $n^{-1}\ll a\asymp b\ll n^{-2/3}$. Suppose
\begin{equation}
\delta=\lim_n\frac{(k-1)(a-b)^3}{\sqrt{6}}\left(\frac{n}{k(a+(k-1)b)}\right)^{3/2}\in [0,\infty).\label{eq:shift}
\end{equation}
Then the following three convergence results hold:
\begin{align}
\frac{\sqrt{{n\choose
      3}}\;\hat \chi_{ez}}{\sqrt{\hat{T}}}
& \leadsto N(\delta,1)\\
\frac{\sqrt{{n\choose
      3}}\; \hat\chi_{ez} }{\sqrt{\left(\sfrac{\hat{V}}{\hat{E}}\right)^3}}
&\leadsto N(\delta,1)\\
{\textstyle 2{\sqrt{{n\choose
      3}}}\left(\sqrt{\hat{T}}-\left(\sfrac{\hat{V}}{\hat{E}}\right)^{3/2}\right)}&
\leadsto N(\delta,1).\label{eq:vst}
\end{align}
\end{theorem}

Theorem \ref{thm:main} shows that the asymptotic distribution of the
testing statistic is Gaussian. We can either normalize
$\sqrt{{n\choose 3}}\, \hat\chi_{ez}$ by
$\sqrt{\hat{T}}$ or by
$\sqrt{\left(\sfrac{\hat{V}}{\hat{E}}\right)^3}$. However, 
it may be possible that $\hat{T}=0$ or $\hat{V}=0$. Thus,
we prefer the normalization by
$\frac{1}{2}\Bigl(\sqrt{\hat{T}}+ \left(\sfrac{\hat{V}}{\hat{E}}\right)^{3/2}\Bigr)$,
which results in \eqref{eq:vst}. This square-root
normalization can be seen as a form of variance-stabilizing transformation
\citep{anscombe1948transformation}.

The assumption $n^{-1}\ll a\asymp b\ll n^{-2/3}$ controls the sparsity
of the graph. It covers the most interesting nontrivial range studied
in the  community detection literature, which is from $n^{-1}$ to $n^{-1}\log n$. Below the order of $n^{-1}$, the graph is so sparse that consistent community detection is not possible \citep{mossel2012stochastic,mossel2013proof}. Above the order of $n^{-1}\log n$, the graph carries sufficient information and strong consistency of community detection can be proved \citep{bickel09,abbe2014exact}.

\subsection{Power of the EZ test}
\label{sec:ezpower}

The mean of the asymptotic distribution is given in
\eqref{eq:shift}. When $\chi_{ez}=0$, we get $\delta=0$, and the asymptotic distribution is $N(0,1)$. Therefore, the p-value of the test can be calculated from the standard Gaussian quantile function.
When $k\rightarrow\infty$, the order of \eqref{eq:shift} is
$$\delta\asymp \left(\frac{n(a-b)^2}{k^{4/3}(a+b)}\right)^{3/2}.$$
This leads to the following result on the power of the test.

\begin{theorem}\label{thm:power}
Assume $\E W^4=O(1)$ and $n^{-1}\ll a\asymp b\ll n^{-2/3}$. Suppose
\begin{equation}
\frac{n(a-b)^2}{k^{4/3}(a+b)}\longrightarrow\infty.\label{eq:condition}
\end{equation}
Then, for any constant $t\asymp 1$, we have
$$\P \left(\bigl|{\textstyle 2{\sqrt{{n\choose 3}}}\left(\sqrt{\hat{T}}-\left(\sfrac{\hat{V}}{\hat{E}}\right)^{3/2}\right)}\bigr|>t\right)\rightarrow 1.$$
\end{theorem}

This result characterizes the power of the proposed test under the
condition \eqref{eq:condition}. Conditions of a similar form are 
common in the community detection literature. For example, in the
setting of the DCBM, \cite{gao2016community} require
$\frac{n(a-b)^2}{k^{5}(a+b)}\rightarrow\infty$ for minimax optimal
community detection. The scaling in \eqref{eq:condition} is the same
except for a much weaker dependence on $k$, indicating that 
the problem of testing for network structure may be statistically easier than
network clustering.

When $k=O(1)$, the condition \eqref{eq:condition} reduces to
$\frac{n(a-b)^2}{a+b}\rightarrow\infty$. The optimality of this
condition has been studied in the setting of the stochastic block model, which is a special
setting of the DCBM. For example, when $k=2$, 
\cite{mossel2012stochastic} show that distinguishing between an
Erd\H{o}s-R\'{e}nyi model and a stochastic block model is impossible when
$\frac{n(a-b)^2}{2(a+b)}<1$. On the other hand,
\cite{mossel2012stochastic,banerjee2016contiguity,banerjee2017} show
that when $\frac{n(a-b)^2}{2(a+b)}>1$, there exists a consistent test
to distinguish Erd\H{o}s-R\'{e}nyi model and a stochastic block model.
For a growing  number of communities $k$,
the impossibility result was extended by
\cite{banks2016information}, showing that 
an Erd\H{o}s-R\'{e}nyi model is indistinguishable from a stochastic
block model if $\frac{n(a-b)^2}{k\log k(a+b)}$ is bounded by some constant. Here, we
simplify the expression by assuming that $a\asymp b$.

In this paper, we study the more general setting of the
DCBM. Therefore, established lower bounds for the stochastic block
model also apply here.  The Erd\H{o}s-Zuckerberg test requires the
condition $\frac{n(a-b)^2}{k^{4/3}(a+b)}\rightarrow\infty$, which is
nearly optimal compared to these lower bound results.

\subsection{Computation of the test statistic for sparse networks}
\label{sec:comp}

Sparse matrix multiplication can be used to efficiently 
compute the test statistic $\hat\chi_{ez}$. If $A$ is the binary adjacency
matrix of the graph, then $(A^l)_{ij}$ is the number of paths of
length $l$ from $i$ to $j$. It follows that
\begin{align}
\label{eq:comp1}
\hat E &= \frac{1}{2{n\choose 2}} \mip{\ones}{A}\\
\label{eq:comp2}
\hat V &= \frac{1}{6{n\choose 3}} \left(\mip{\ones}{A^2} - \tr(A^2)\right) \\
\label{eq:comp3}
\hat T &= \frac{1}{6{n\choose 3}} \tr(A^3)
\end{align}
where $\tr(\cdot)$ denotes the matrix trace, $\mip{A}{B} = \tr(A^T B)$
is the matrix inner product, and the symbol $\ones$ denotes a matrix of all ones.
These relations were used to efficiently calculate the test statistic in the experiments
presented in Section~\ref{sec:experiments}.

\subsection{An EZ test for stochastic block models}
\label{sec:sbm}

When $\mathcal{W}$ becomes a delta measure on $1$, the DCBM reduces to
the SBM. A simplified Erd\H{o}s-Zuckerberg characteristic holds in
this setting, which only requires the estimation of the edge and the
triangle densities.
\begin{proposition}
When $\E W=\E W^2=1$, we have
$$T-E^3=\frac{(k-1)(a-b)^3}{k^3}.$$
\end{proposition}
This result is easily derived from Proposition \ref{prop:EZ} via the
relation $V=E^2$ when $\E W=1$ by \eqref{eq:E+}
and \eqref{eq:V+}. Analogous results to Theorem \ref{thm:main} and
Theorem \ref{thm:power} also hold for the plug-in statistic
$\hat{T}-\hat{E}^3$ under the SBM. In particular, $2\sqrt{{n\choose
3}}\left(\sqrt{\hat{T}}-\sqrt{\hat{E}^3}\right)\leadsto N(\delta,1)$,
where $\delta$ shares the same definition in \eqref{eq:shift}. Moreover,
the power of the corresponding test goes to one under the alternative
hypothesis of a stochastic block model with the same signal-to-noise ratio condition
\eqref{eq:condition}. See \cite{gao2017testing} for further detail.

While the form of this test is similar, there is a significant
difference between the Erd\H{o}s-Zuckerberg
characterizations for the SBM and the DCBM. Consider a DCBM with
$k=1$---in other words, a configuration model. 
By Proposition \ref{prop:EZ},
$T-\left(\sfrac{V}{E}\right)^3=0$. However, a simple calculation using
the expressions in \eqref{eq:E+}--\eqref{eq:T+} gives
$$T-E^3=a^3(1-(1-\Var(W))^3).$$ Thus, as long as $\Var(W)>0$, the
statistic satisfies $T-E^3>0$.  

This calculation shows that while the
configuration model is the benchmark of triangle frequency used in our
EZ test, this model will have 
more triangles compared with the benchmark of an Erd\H{o}s-R\'{e}nyi model. 
This phenomenon is apparent in the plots of Figure~1,
where the surface indicates the subfamily of degree-corrected stochastic block
models for which $\chi_{ev} = T - (V/E)^3 = 0$, and no community
structure is present. The curve on the surface, visible in the upper left
plot, corresponds to the subfamily of \ER{} graphs where $(E,V,T) =
(p, p^2, p^3)$. Each point represents a Facebook subnetwork; the points that lie
relatively close to the \EZ{} surface $\chi_{ez}=0$ are still far from the
\ER{} curve. This is attributable to degree heterogeneity in the
Facebook networks, which is captured by the configuration model.

\section{Examples}
\label{sec:experiments}

In this section we describe experiments with the proposed testing
framework on three types of data: social networks,
citations from journal articles, and stock returns of companies on the S\&P 500.
In each setting, we demonstrate the performance of the test
qualitatively, by showing examples of the networks that have large
and small p-values. For each of the three data sets, we find that
the \EZ{} test gives interesting and intepretable results, and is
effective at identifying community structure. 

\subsection{Facebook friend networks}
\label{sec:friends}

The current work was motivated by the empirical findings of
\cite{ugander2013subgraph}, which compared the distributions of 3-node
and 4-node subgraphs of Facebook friend networks to those obtained
under an \ER{} baseline model.  In this section we apply our testing
method to Facebook subnetworks similar to those used in this previous
work.

\begin{figure}[!ht]
\begin{center}
\begin{tabular}{cccc}
\includegraphics[width=.17\textwidth]{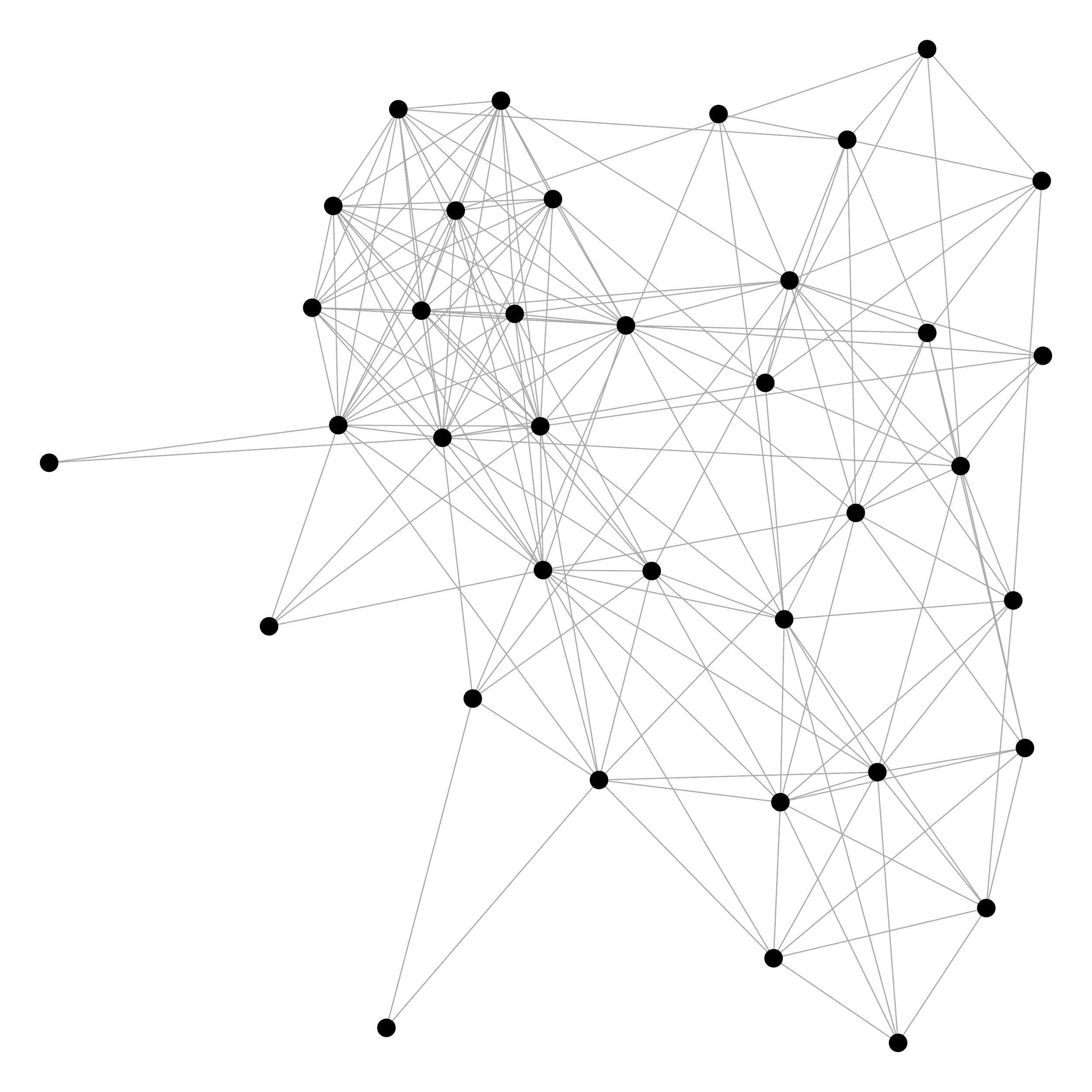} &
\includegraphics[width=.17\textwidth]{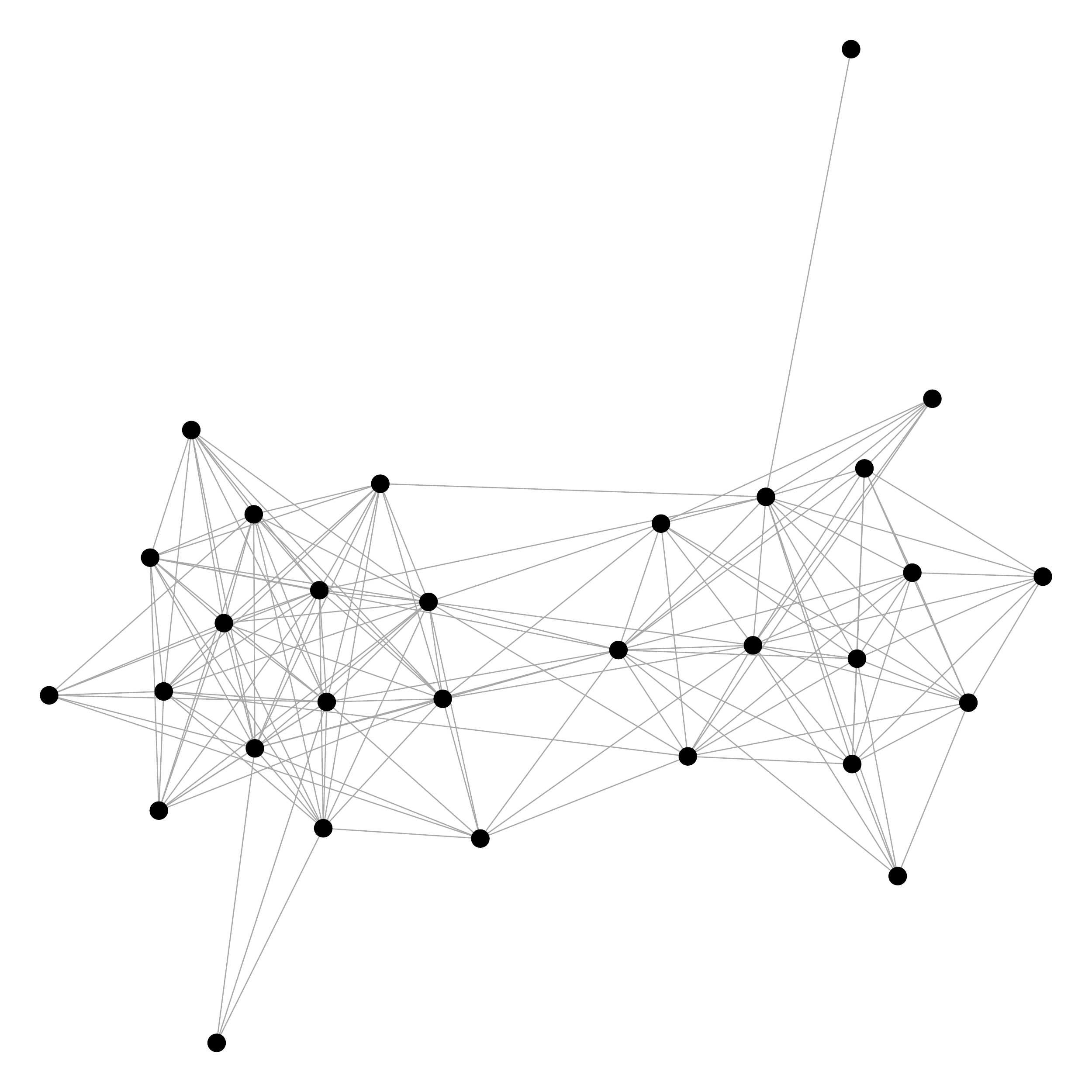} &
\includegraphics[width=.17\textwidth]{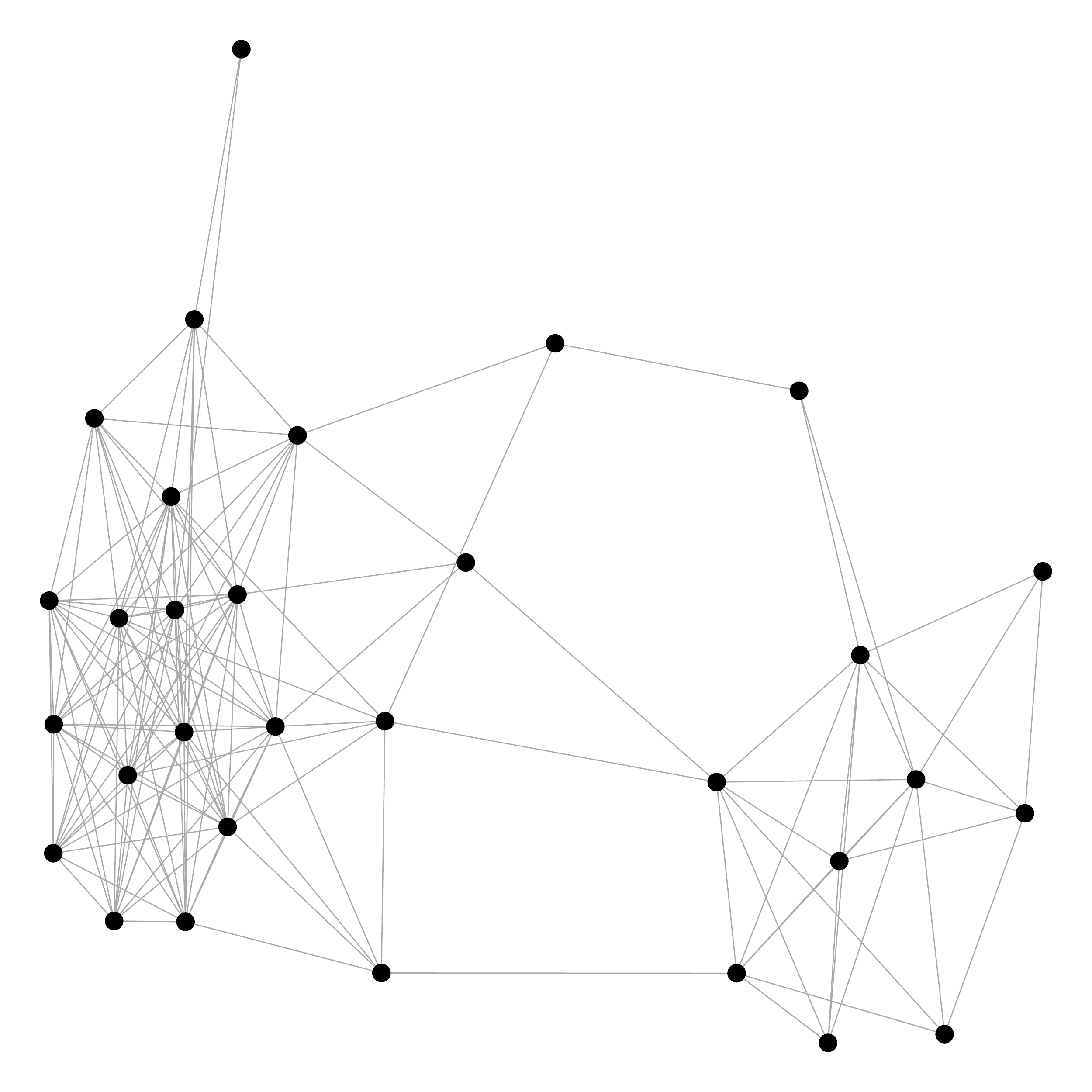} &
\includegraphics[width=.17\textwidth]{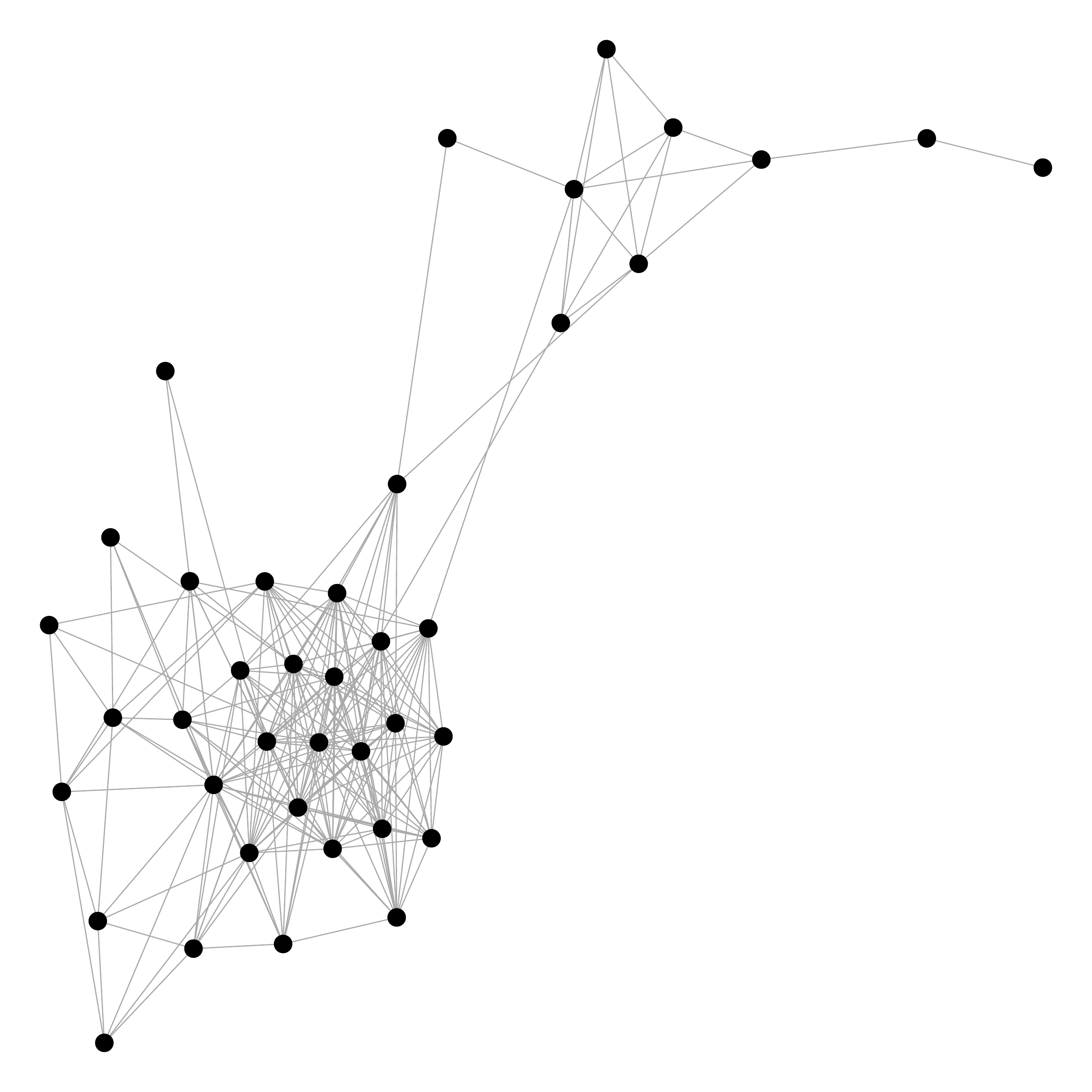} \\
\includegraphics[width=.17\textwidth]{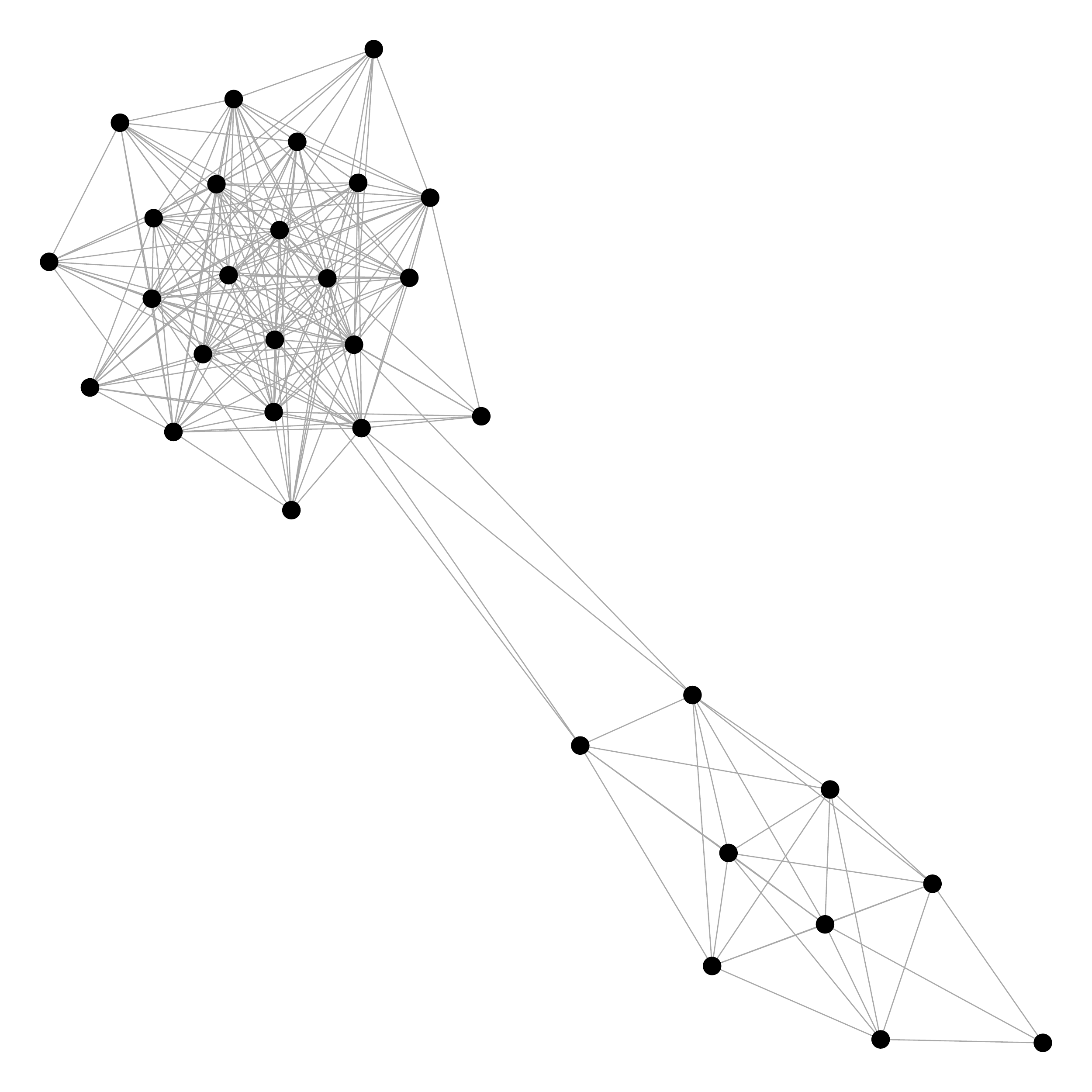} &
\includegraphics[width=.17\textwidth]{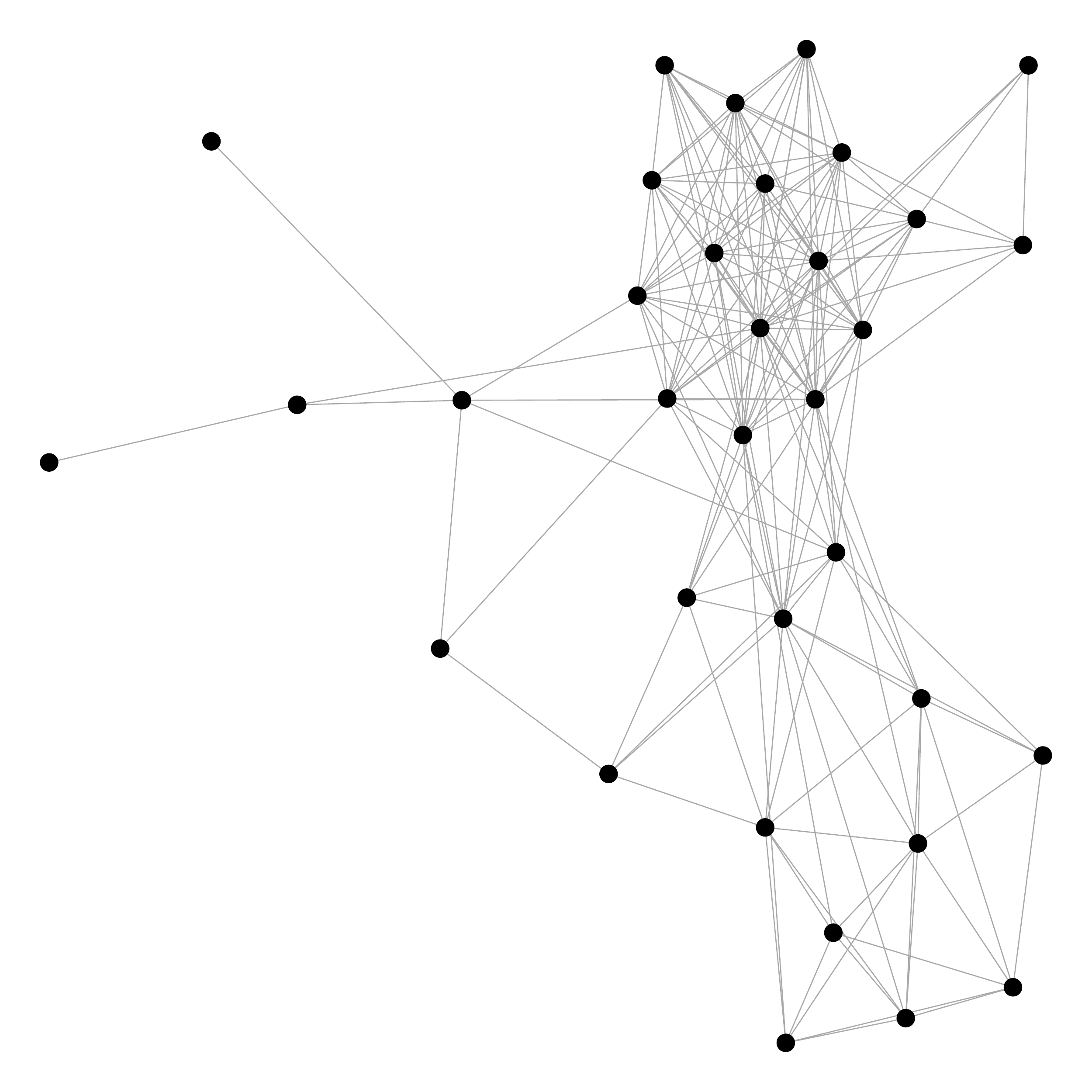} &
\includegraphics[width=.17\textwidth]{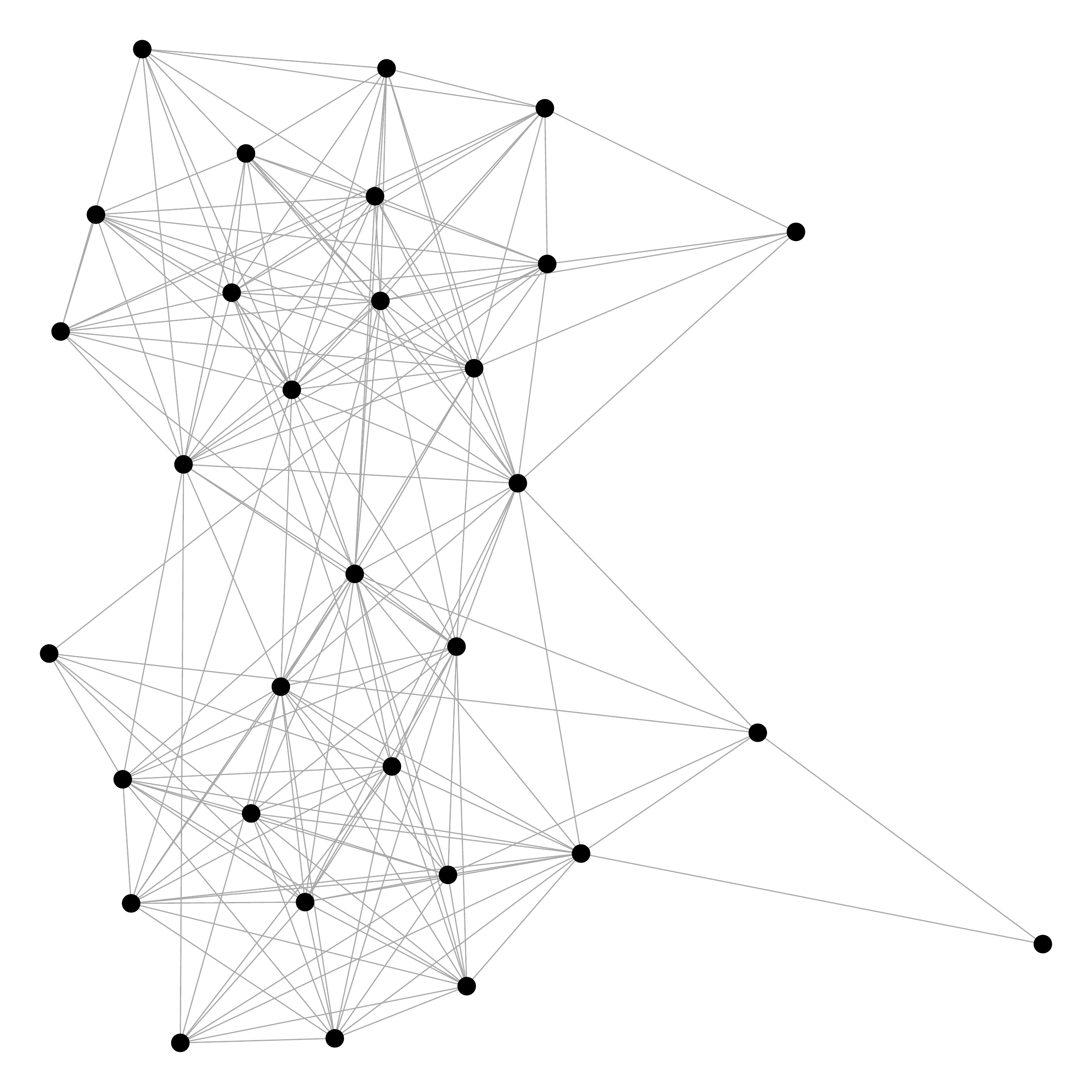} &
\includegraphics[width=.17\textwidth]{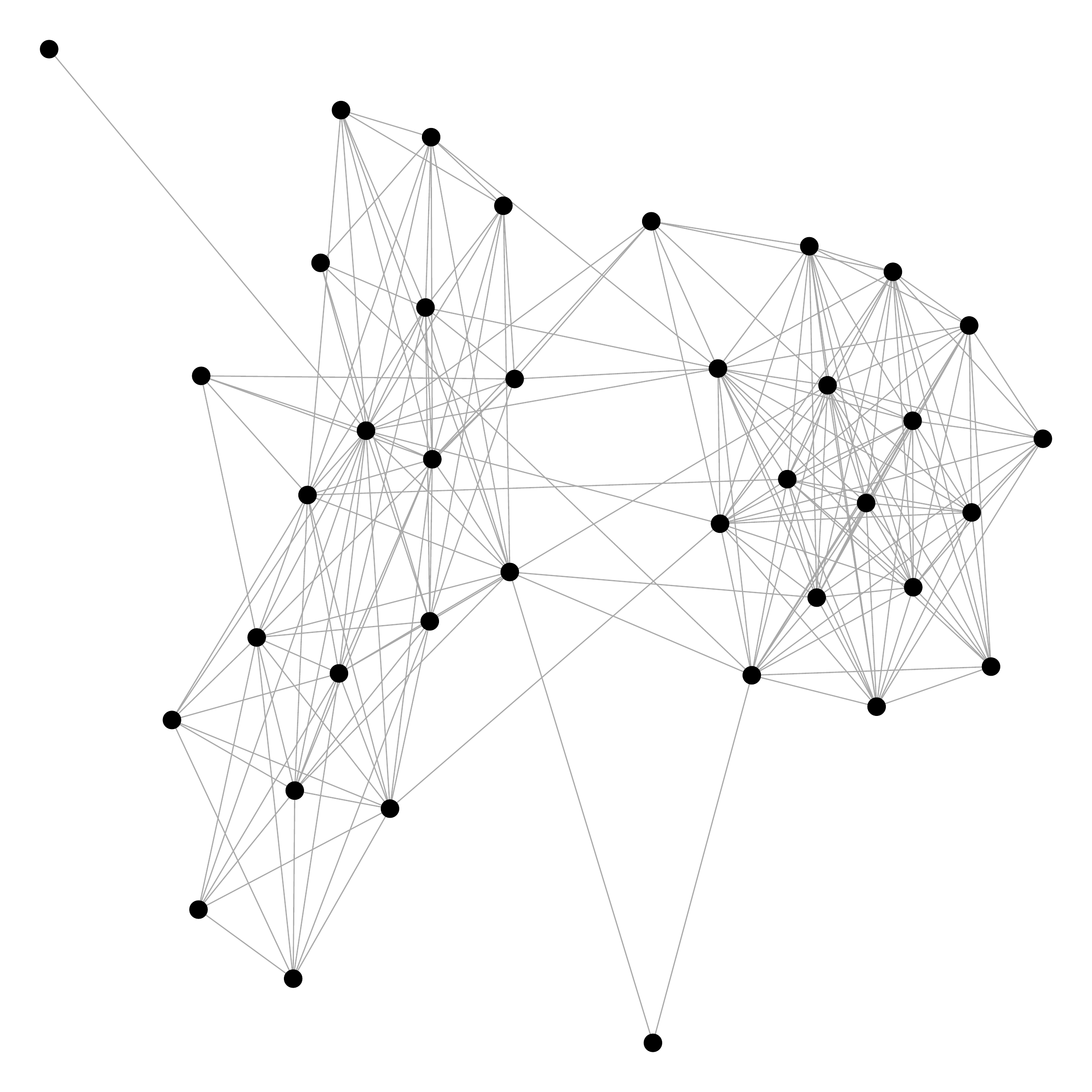} \\
\includegraphics[width=.17\textwidth]{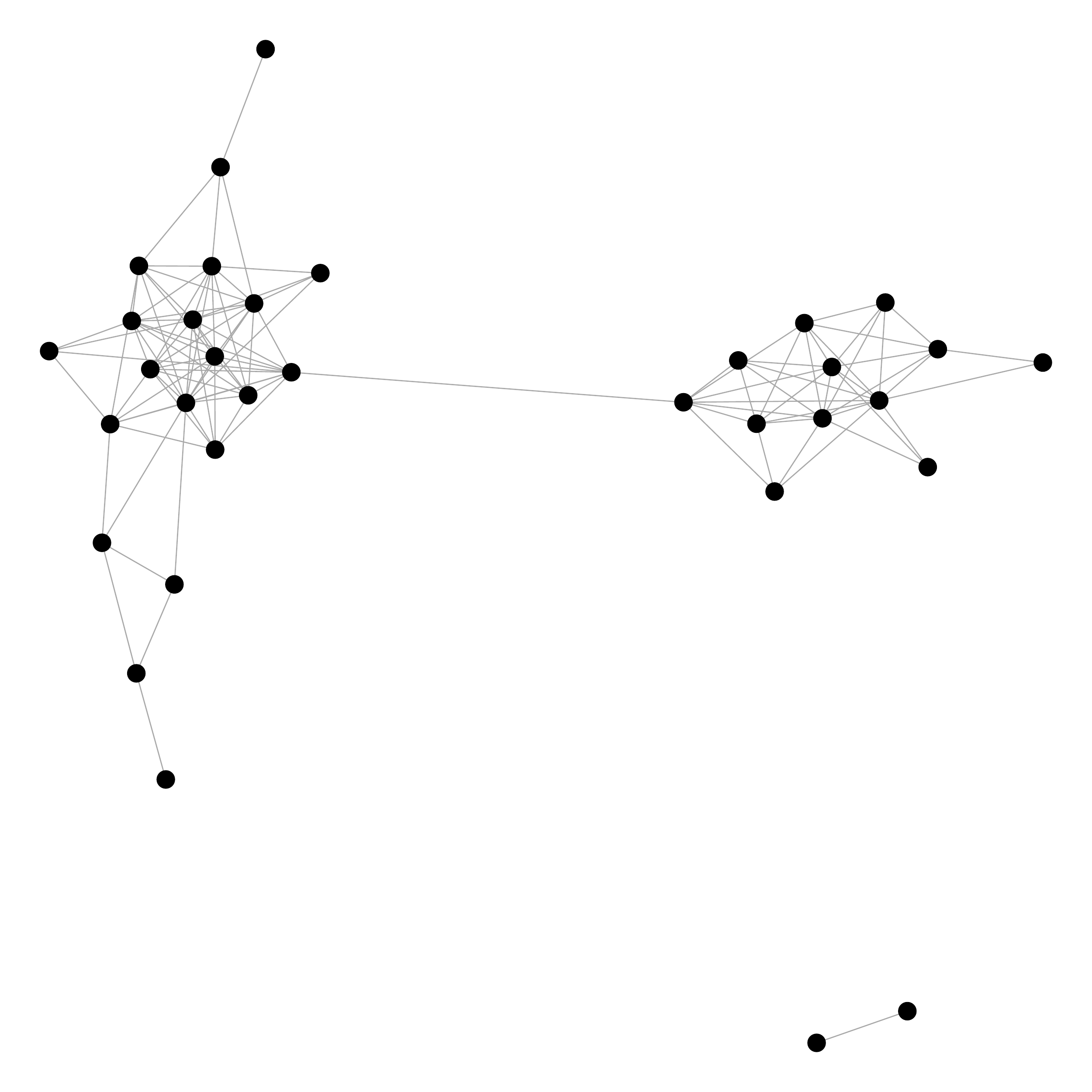} &
\includegraphics[width=.17\textwidth]{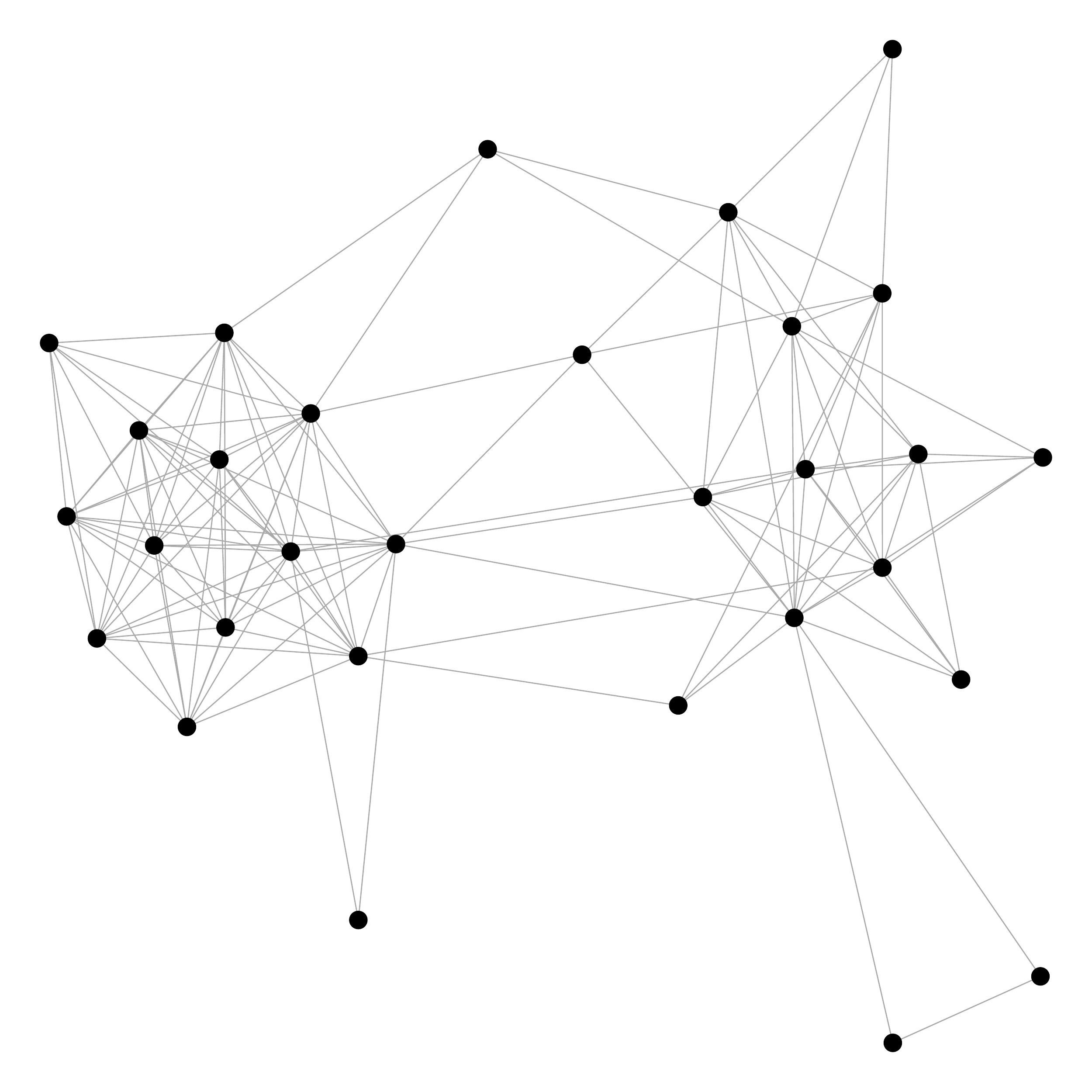} &
\includegraphics[width=.17\textwidth]{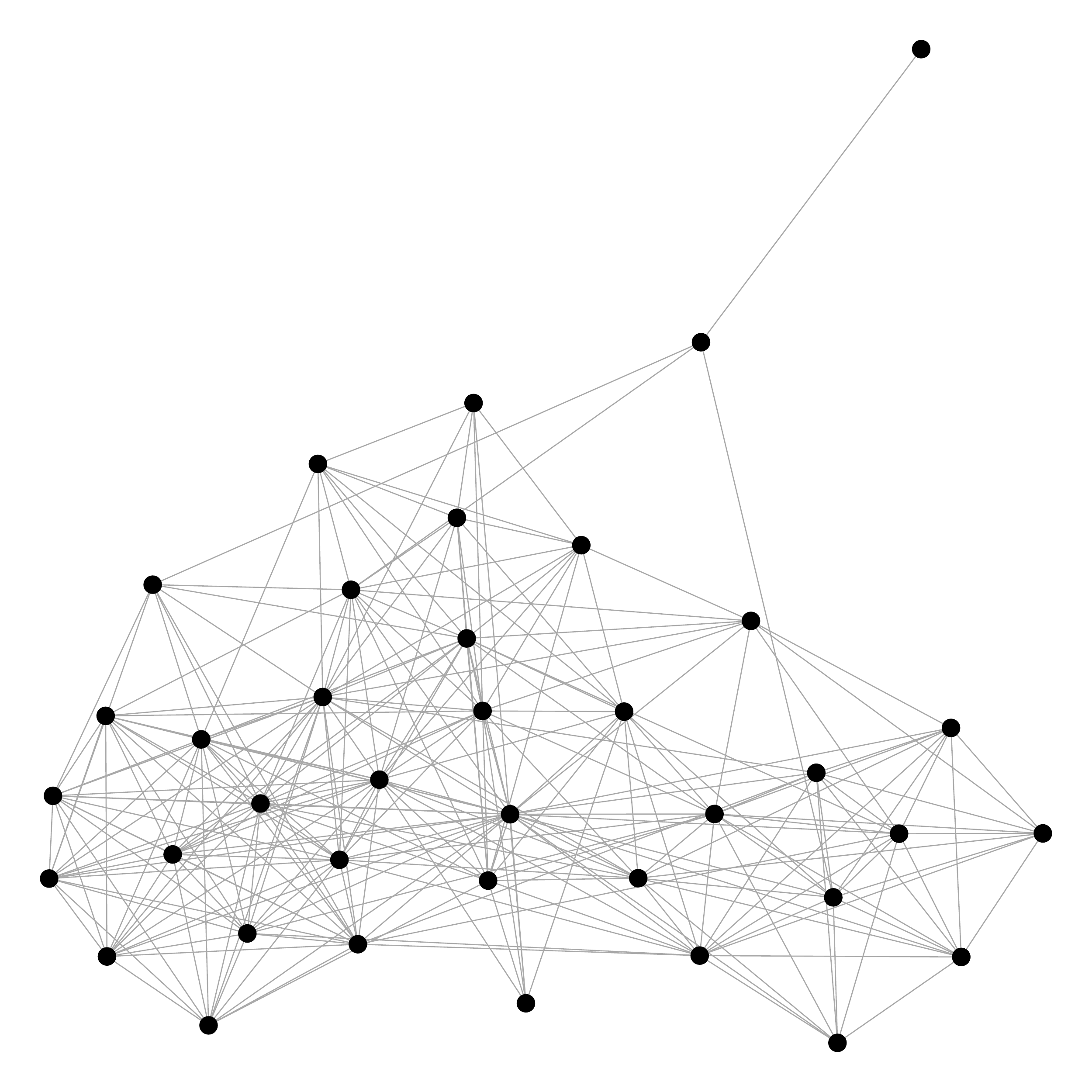} &
\includegraphics[width=.17\textwidth]{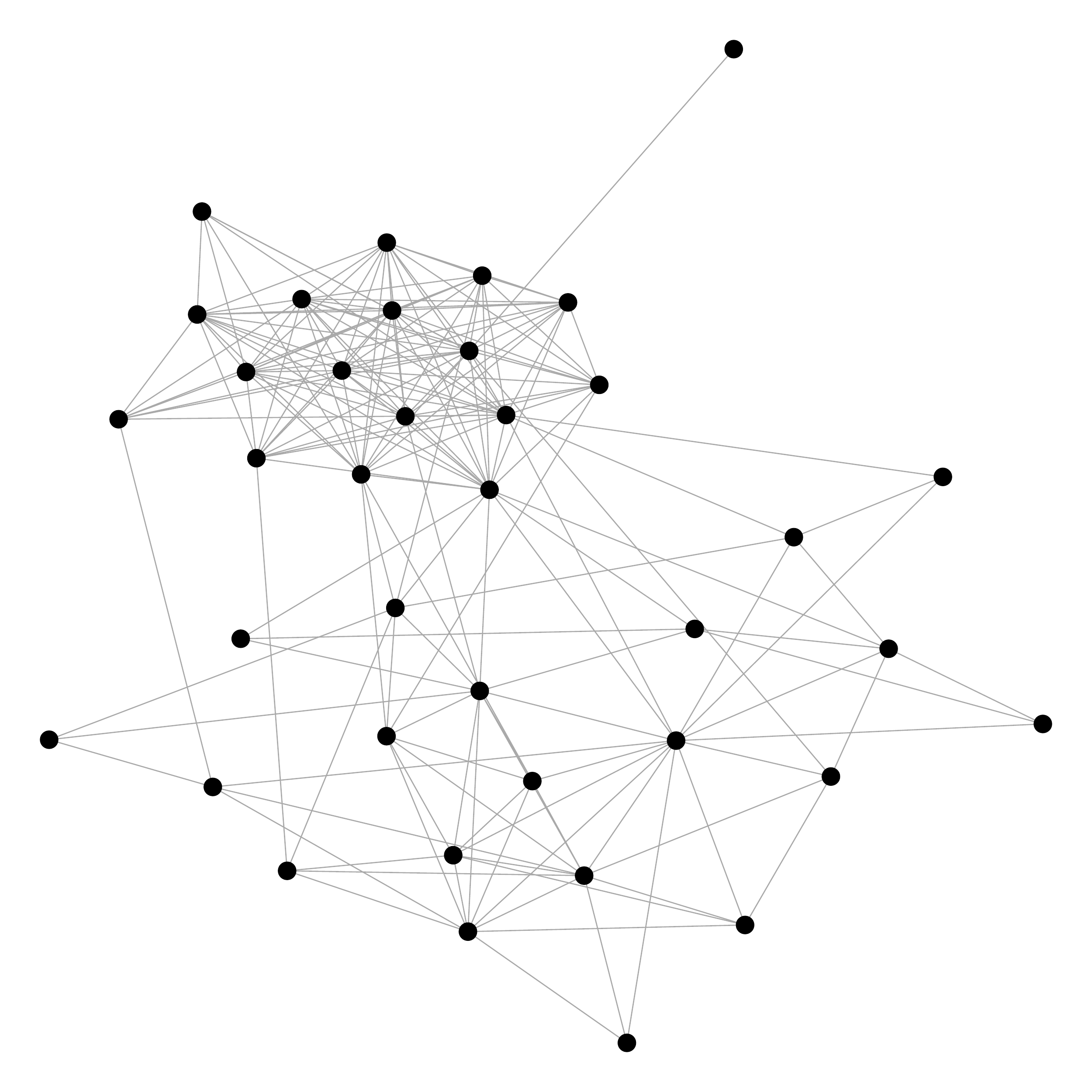} \\
\end{tabular}
\vskip20pt
\centerline{\hbox to 5in{\hrulefill}}
\vskip20pt
\begin{tabular}{cccc}
\includegraphics[width=.17\textwidth]{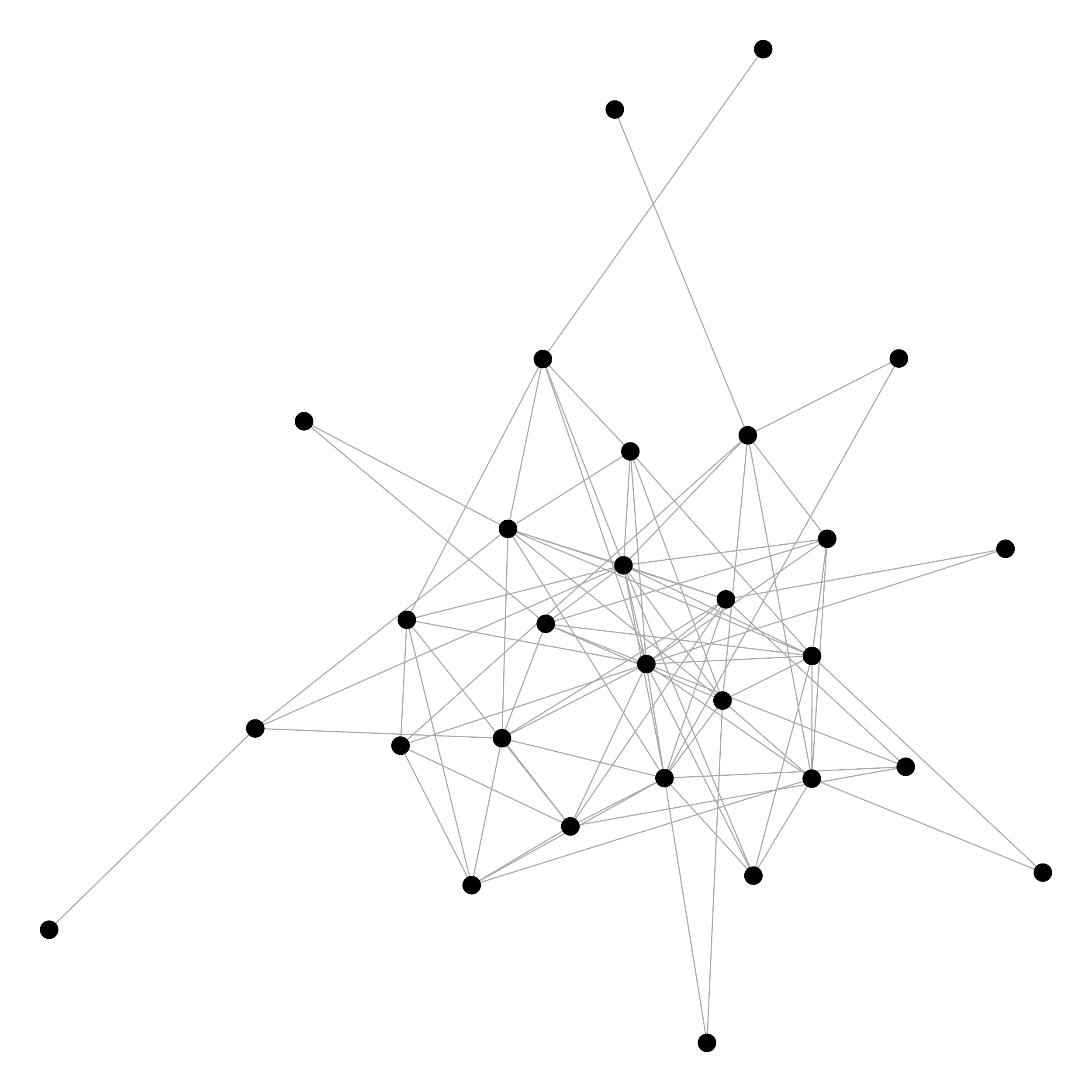} &
\includegraphics[width=.17\textwidth]{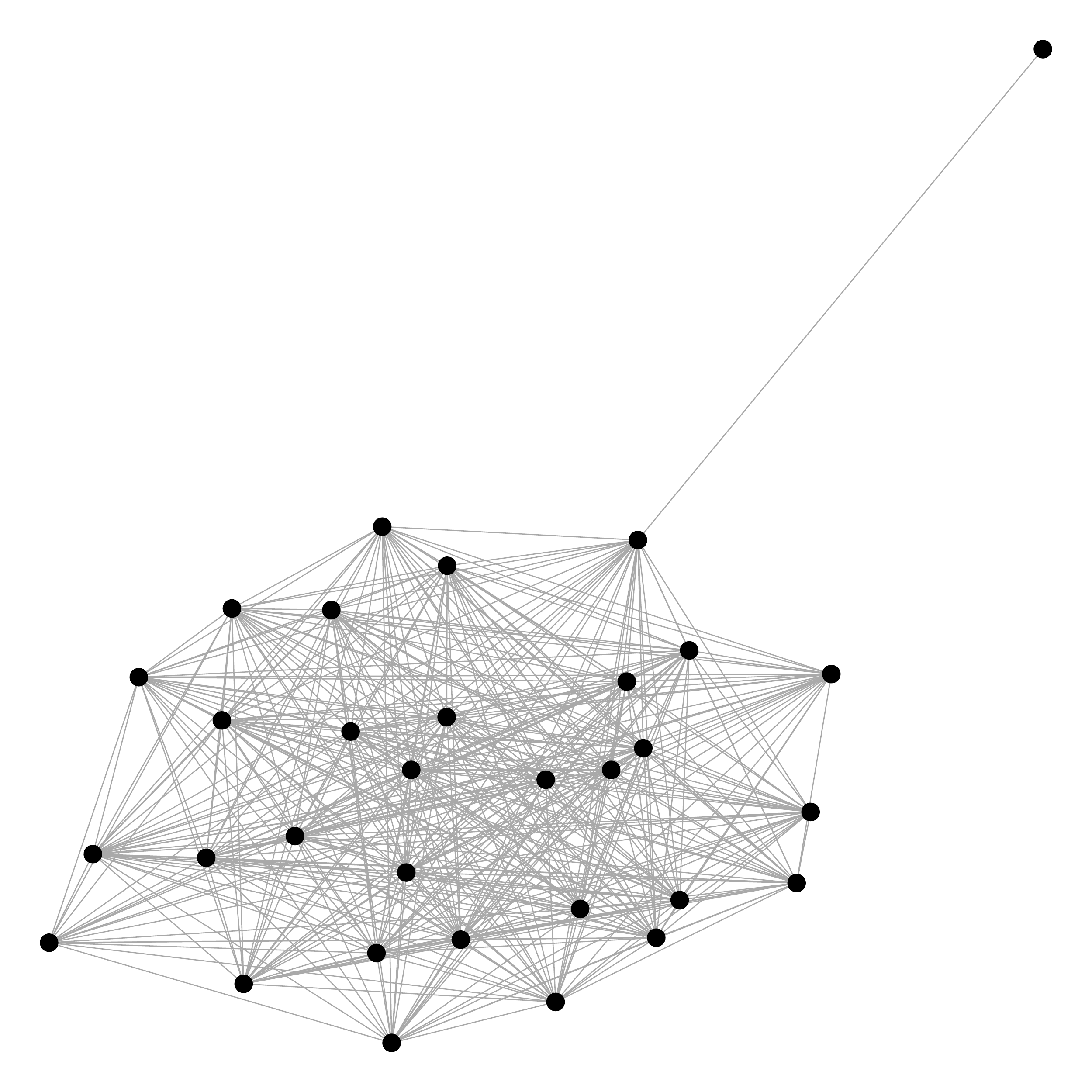} &
\includegraphics[width=.17\textwidth]{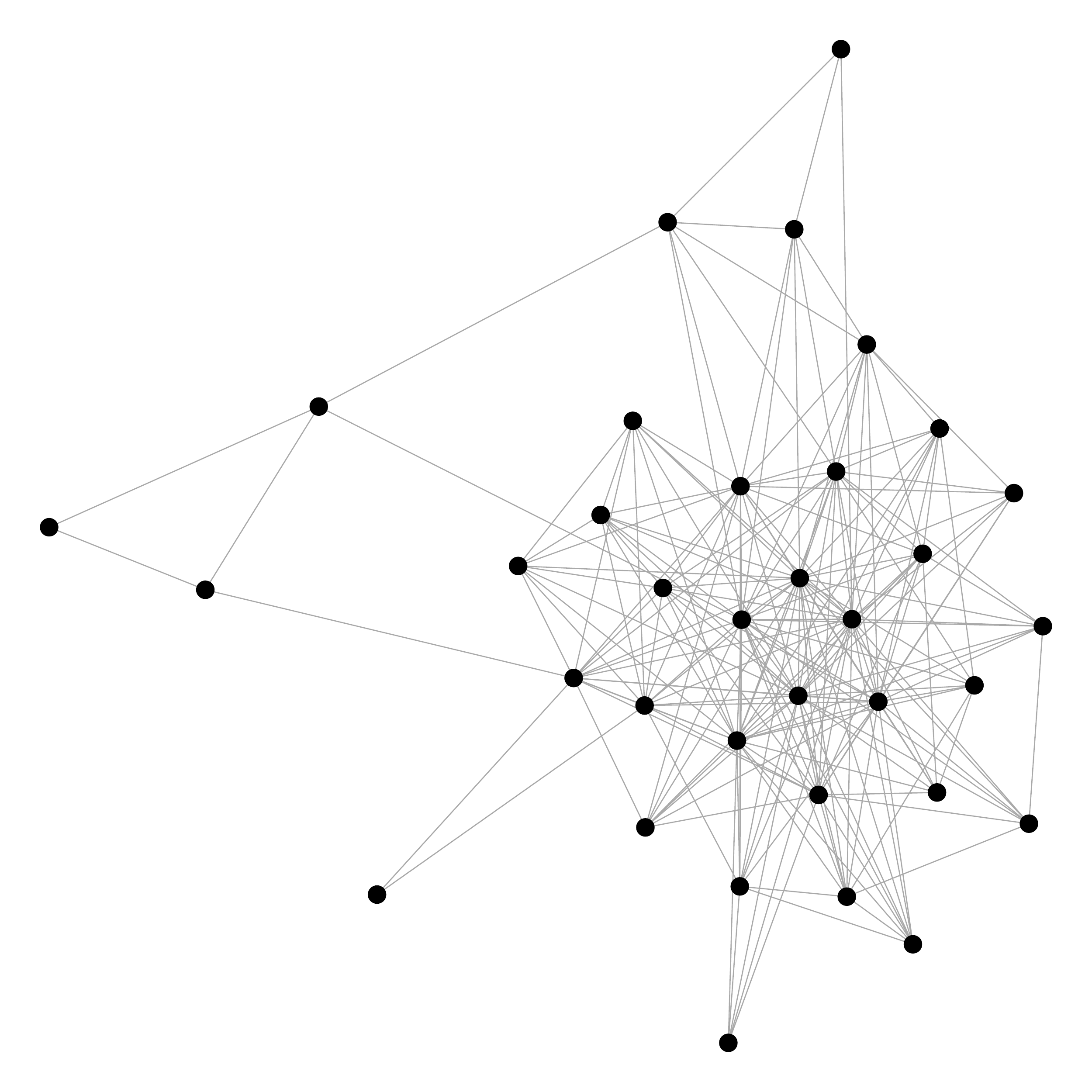} &
\includegraphics[width=.17\textwidth]{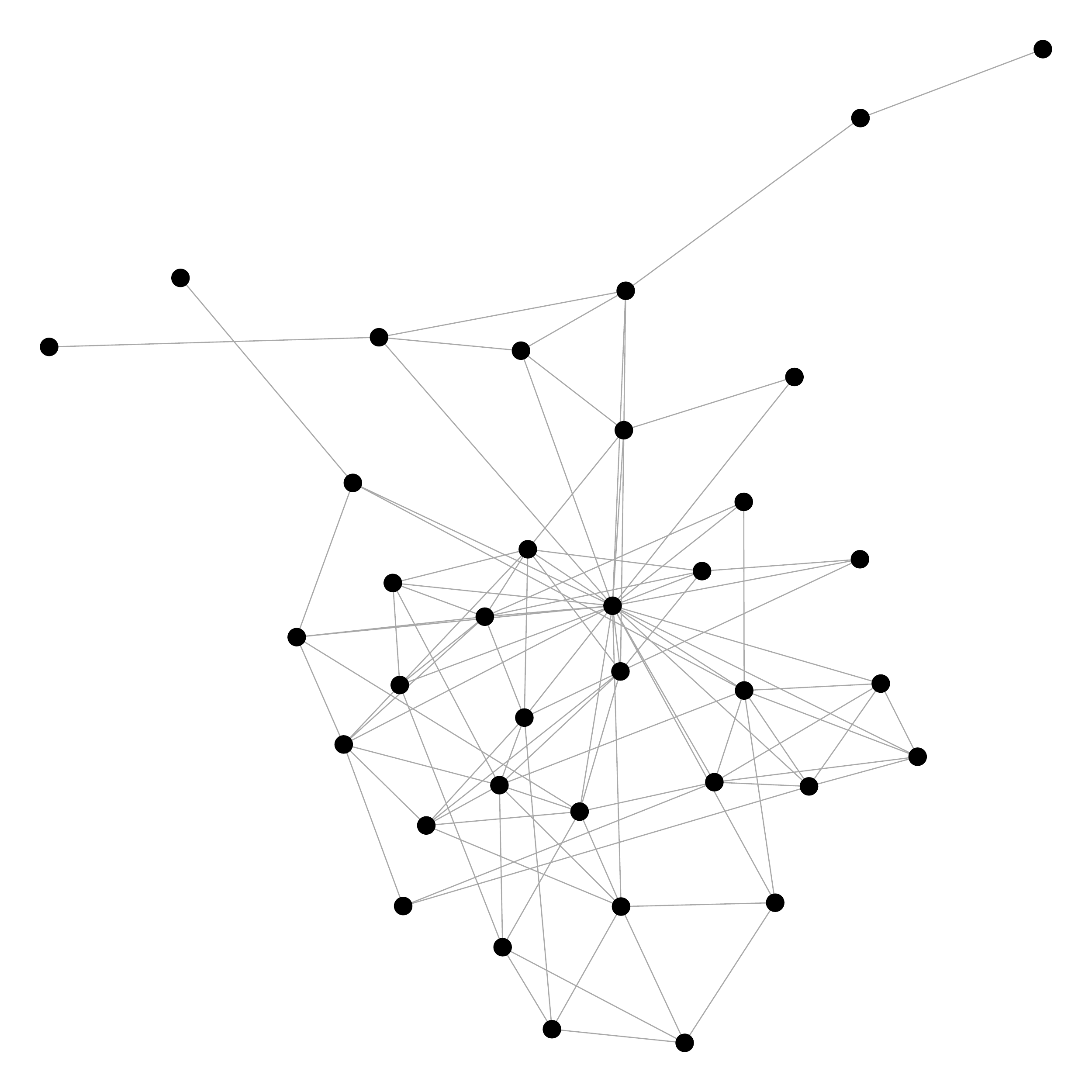} \\
\includegraphics[width=.17\textwidth]{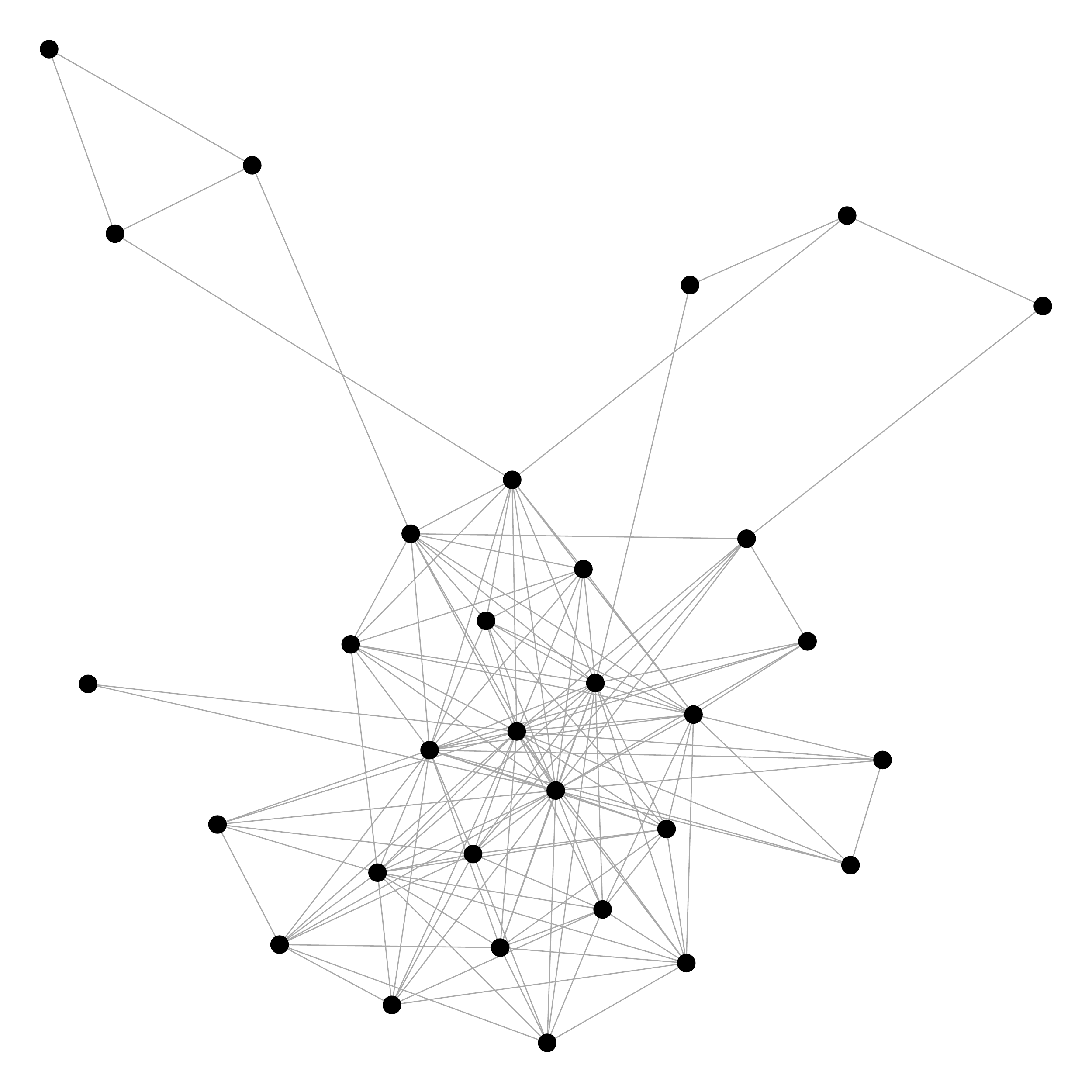} &
\includegraphics[width=.17\textwidth]{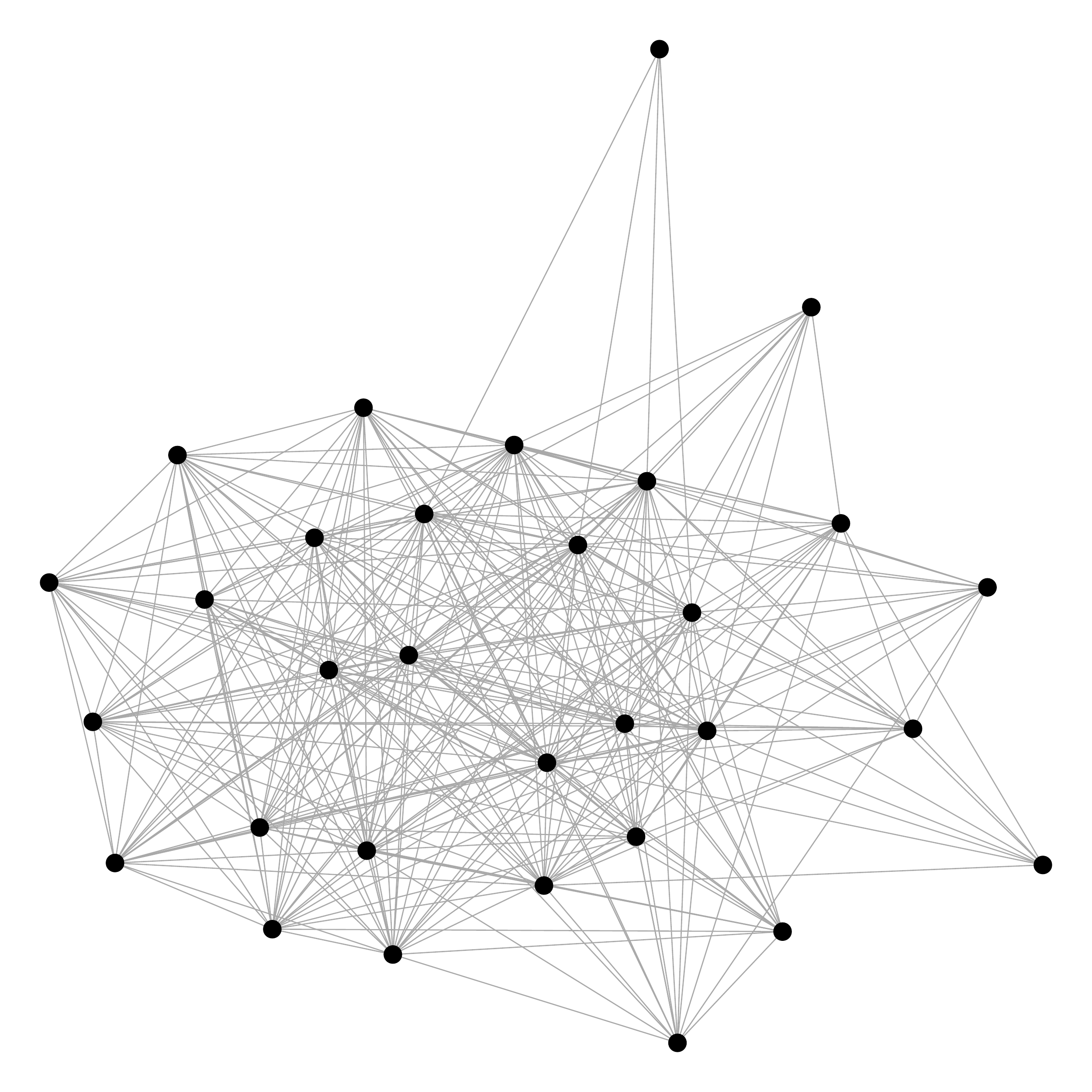} &
\includegraphics[width=.17\textwidth]{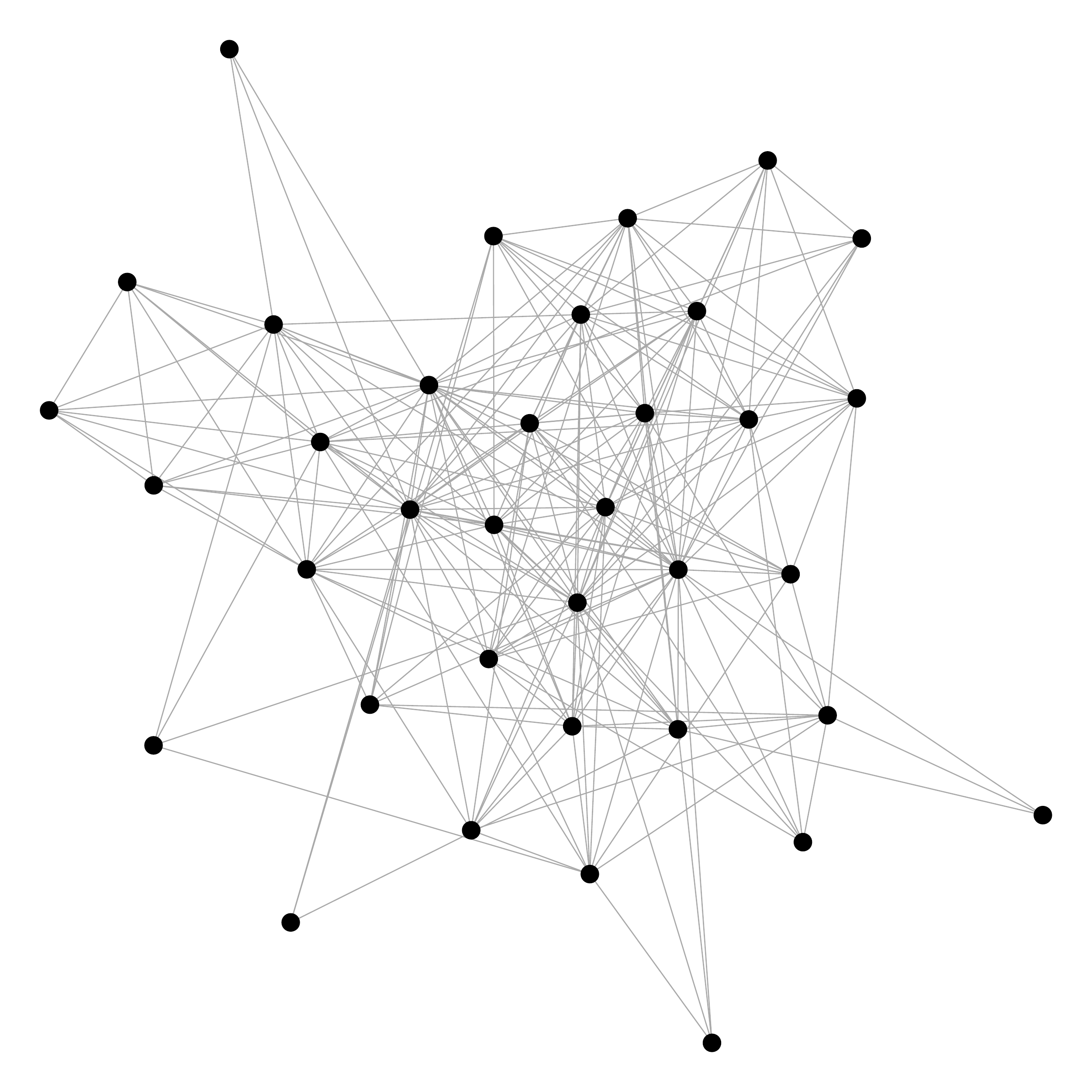} &
\includegraphics[width=.17\textwidth]{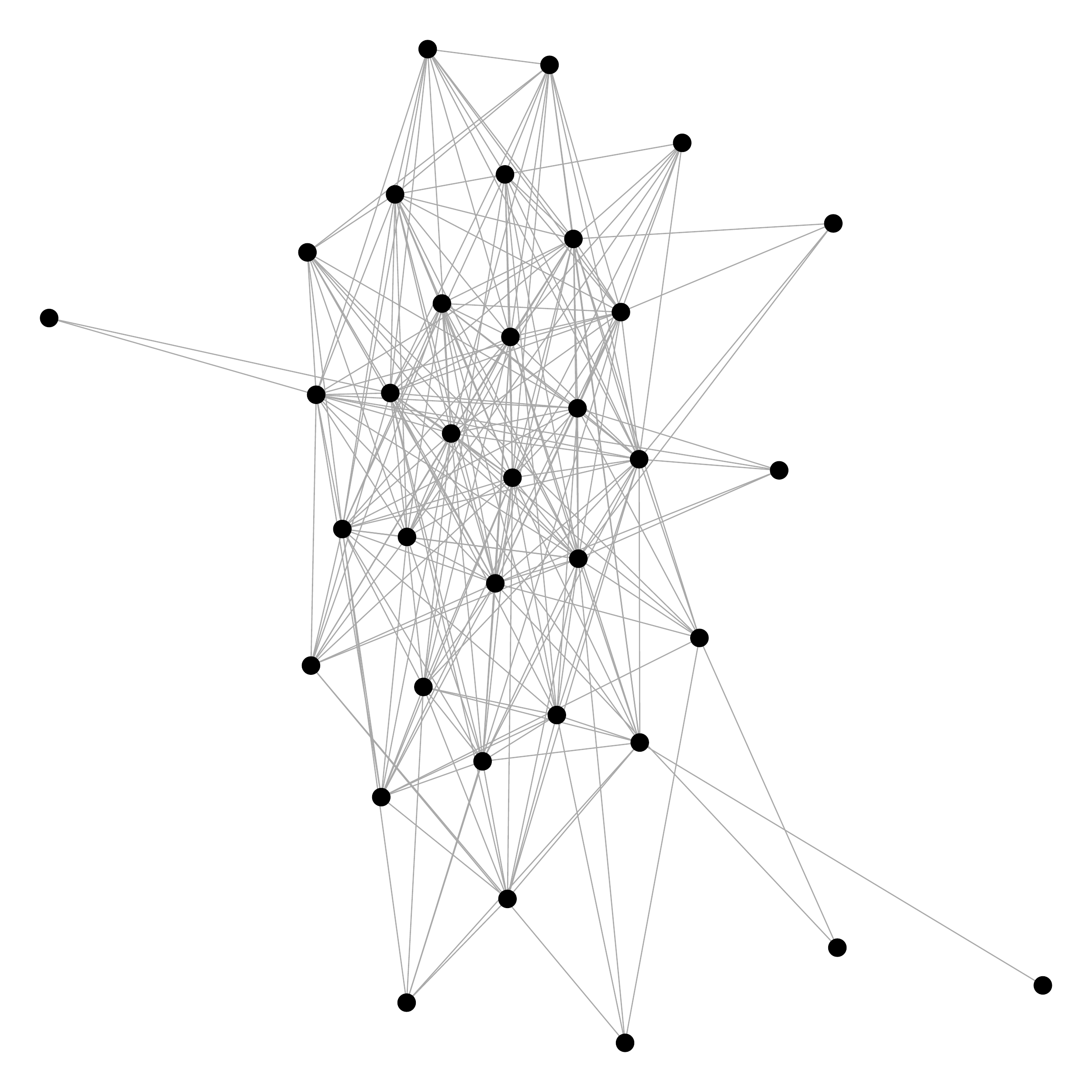} \\
\includegraphics[width=.17\textwidth]{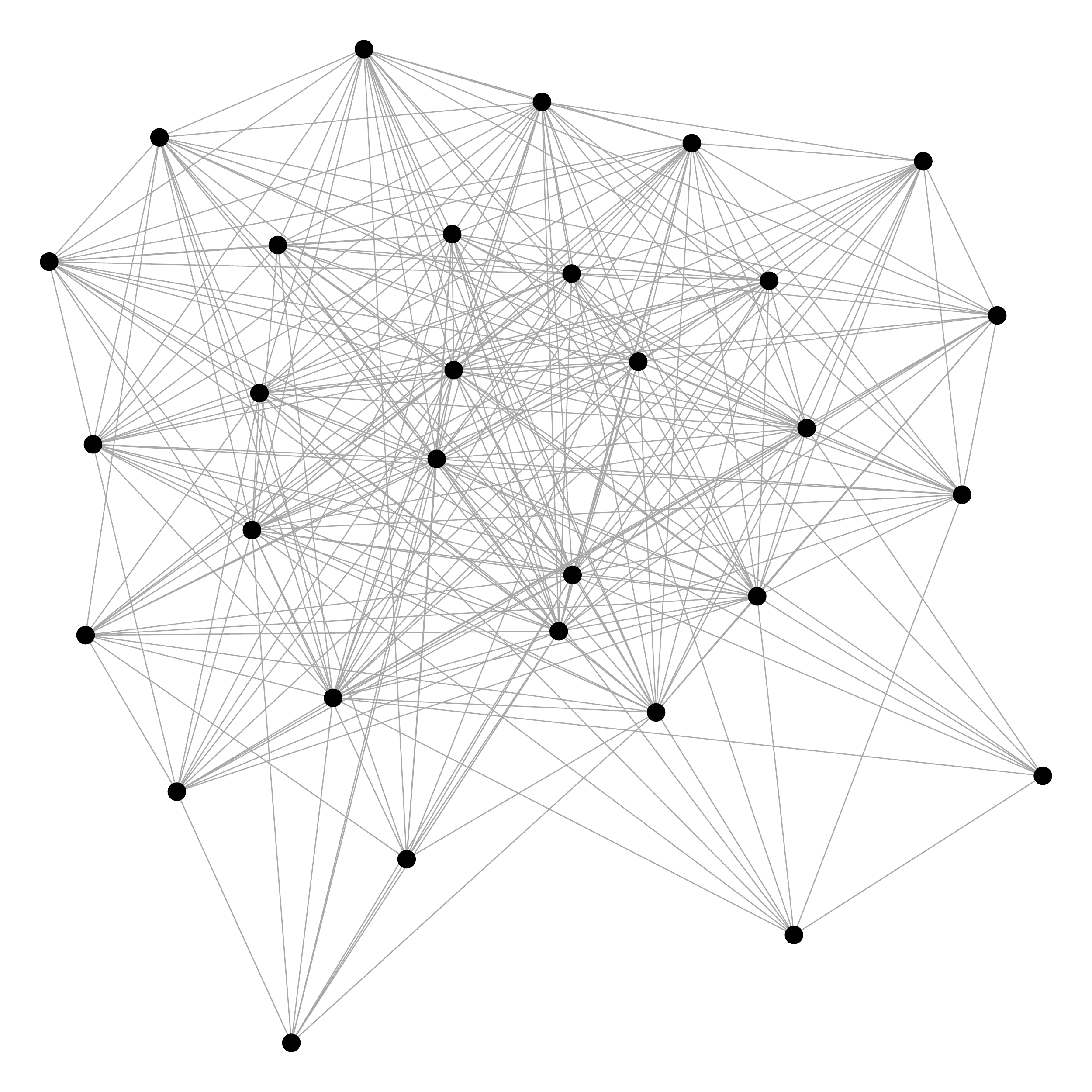} &
\includegraphics[width=.17\textwidth]{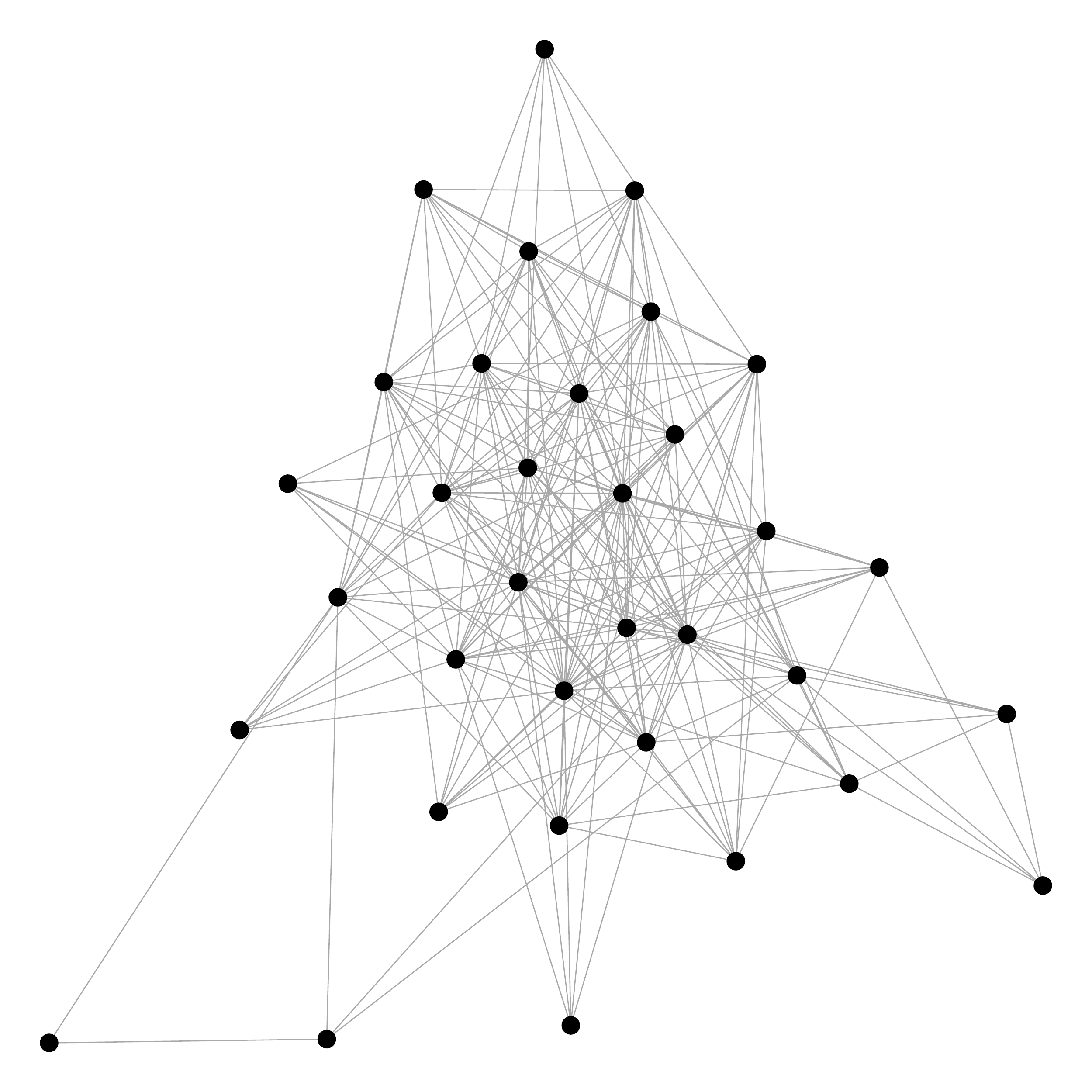} &
\includegraphics[width=.17\textwidth]{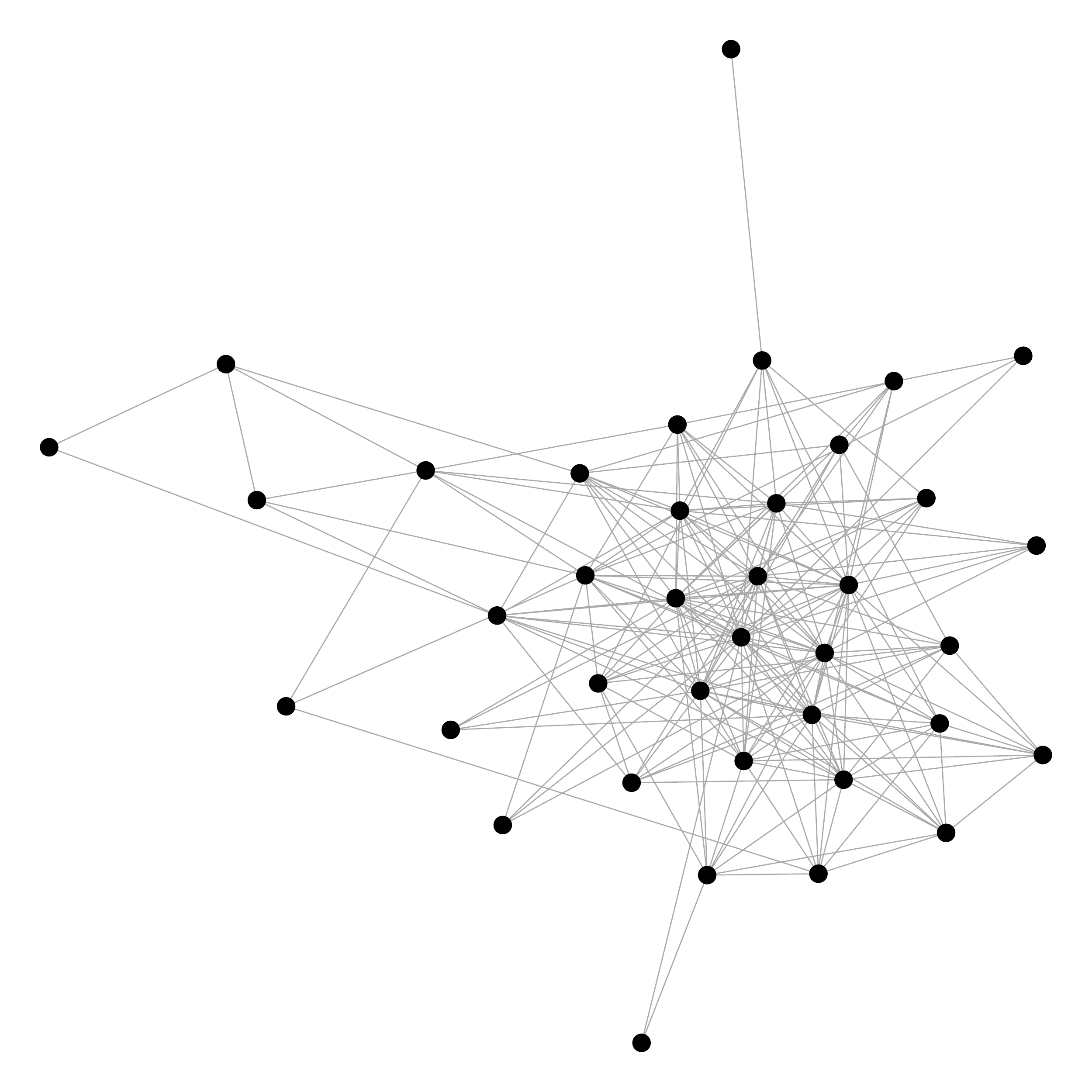} &
\includegraphics[width=.17\textwidth]{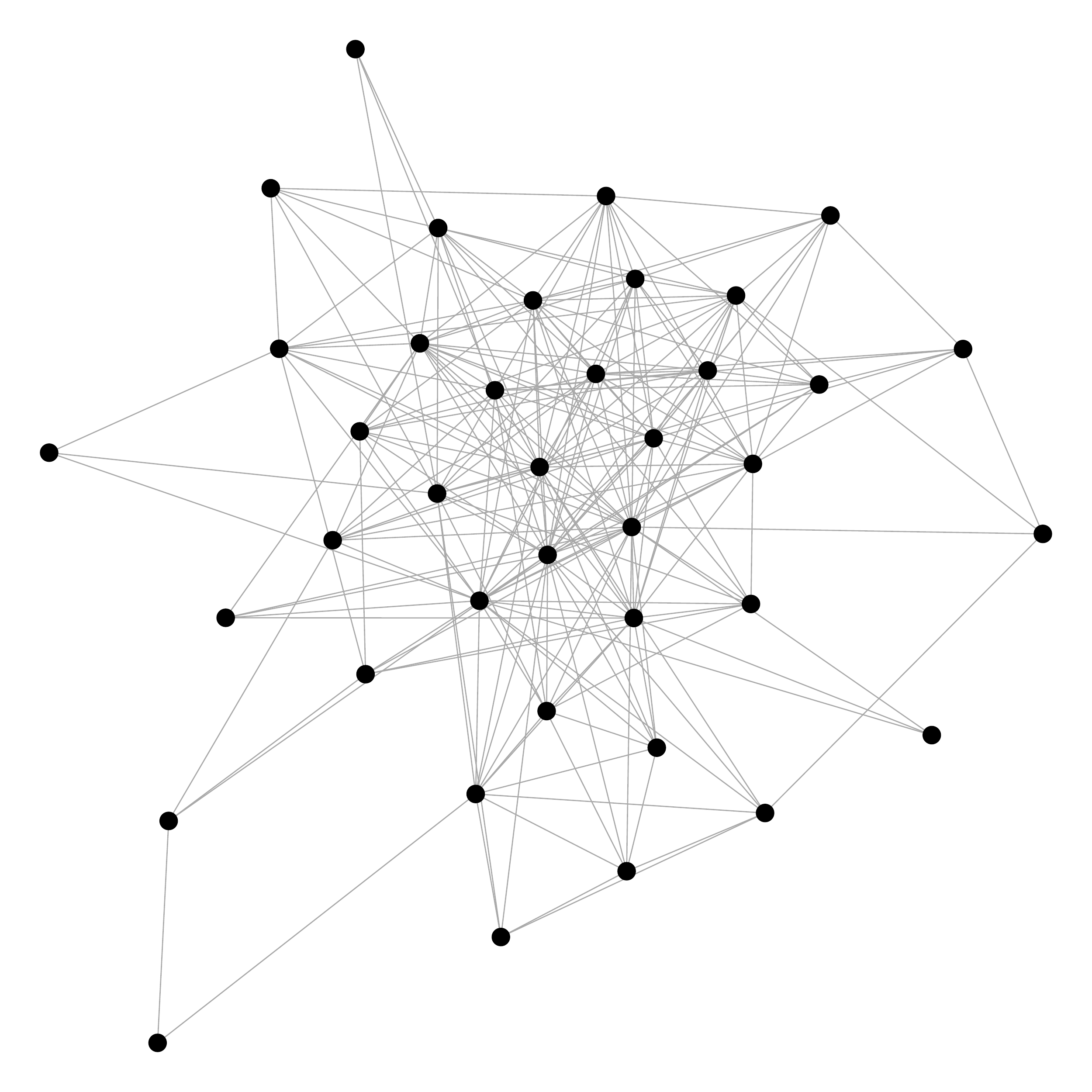} \\
\end{tabular}
\caption{Facebook neighborhood graphs, each with between 30 and 40 nodes,
  extracted from the Carnegie Mellon University portion of the
  Facebook 100 dataset. Top: 12 randomly selected graphs with  p-values smaller than
  $10^{-12}$, under the EZ test. Bottom: 12 graphs
  with p-values larger than 0.1. The results for other
  universities are similar. Community structure is readily apparent in
  the top graphs, and lacking in the bottom graphs.}
\end{center}
\label{fig:fbgraphs}
\end{figure}

The data we use are from the ``Facebook 100'' dataset, comprised
of Facebook friend networks from 100 U.S.~universities, collected
in 2005.
%\footnote{The fb100 dataset, described
%in \cite{fb100:12}, is currently not readily available online.}
In addition to the friend relations
user attributes such as dorm, gender, graduation year, and
academic major are included in the data; however, we do not use
these attributes in our analysis.

The data are divided into separate networks for each of the 100
universities. For a given university, we form the induced graph of a
given user $e$ by forming an adjacency matrix $A^{(e)} =
\bigl(A^{(e)}_{ij}\bigr)$ with respect to the friends of $e$, with
$A^{(e)}_{ij} = 1$ if $i$ is a friend of $j$ (or vice-versa), and
$A^{(e)}_{ij} = 0$ otherwise. A discussion of the inferential
properties of selecting neighborhood graphs this way is given 
in the supplementary material (see Section~\ref{sec:nbhds}).

We display sample results for the Carnegie Mellon University
subnetwork; the results for other universities are qualitatively very
similar. Restricting to subgraphs having between
30 and 40 nodes results in 556 graphs for the CMU subnetwork. The
p-values were computed according to the \EZ{} test implied by equation
\eqref{eq:vst}.  A histogram of these 556 p-values is displayed in
Figure~\ref{fig:pvalues}, where it is seen that most of the p-values
are very small, indicating significant structure.
Figure~\ref{fig:fbgraphs} shows 12 randomly selected graphs 
having large ($> 0.1$) and small ($< 10^{-12}$) p-values under the 
test.  Community structure is readily apparent in the graphs with
small p-values. Structure is absent in the graphs with large
p-values, while they clearly have degree heterogeneity, 
as modeled by the configuration model.

\begin{figure}
\begin{center}
\begin{tabular}{cc}
\includegraphics[width=.42\textwidth]{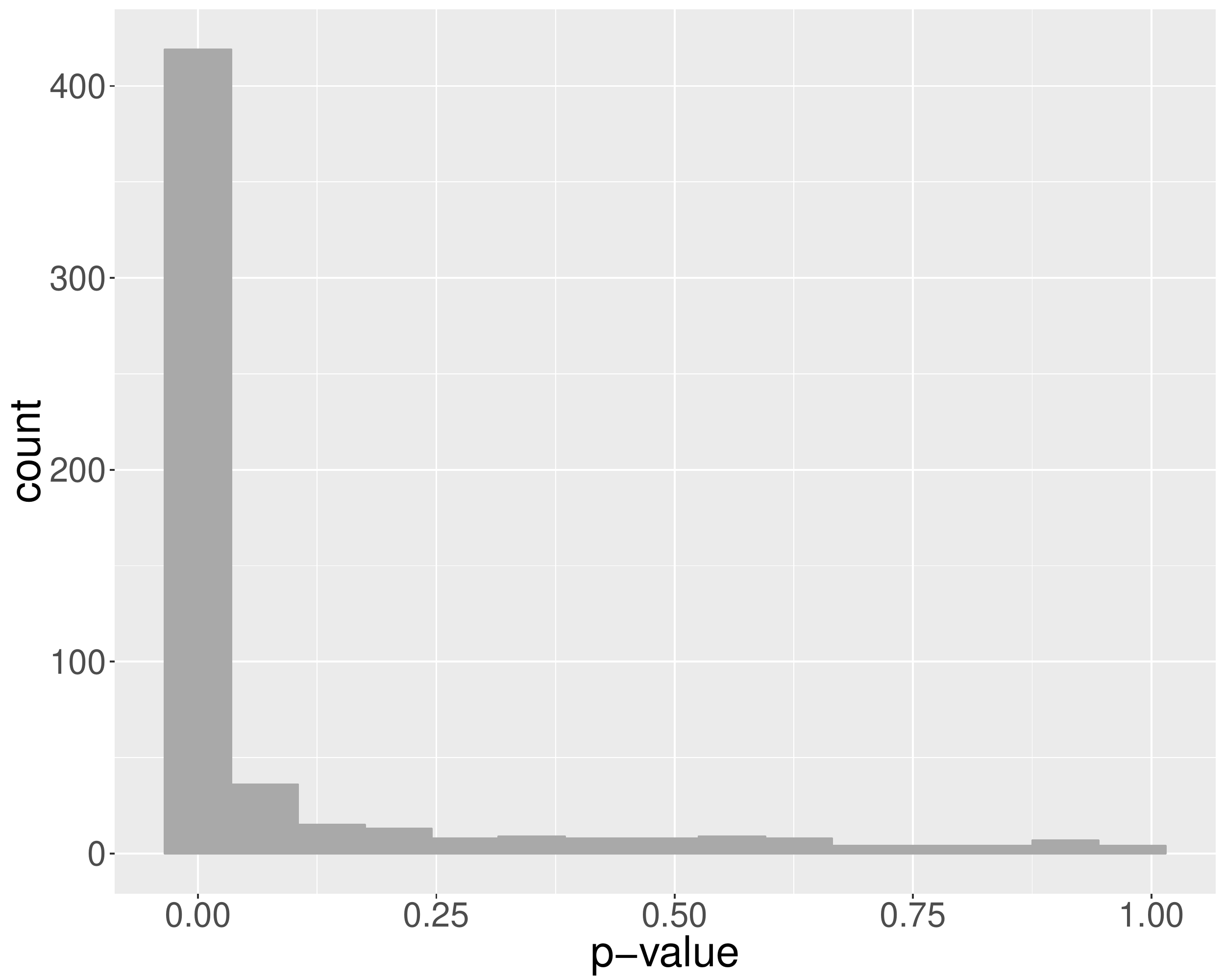}&
\hskip.3in 
\includegraphics[width=.42\textwidth]{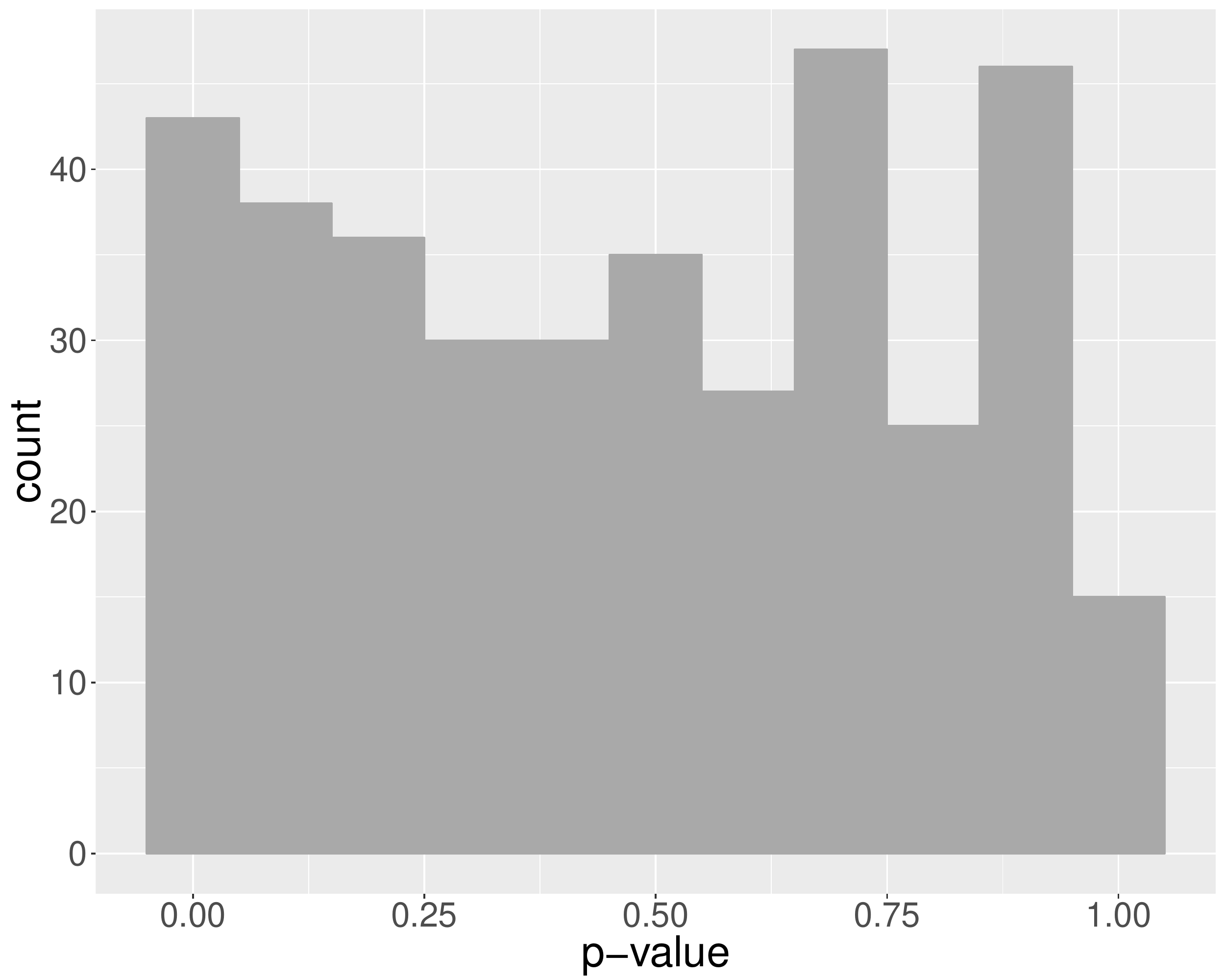}
\end{tabular}
\end{center}
\caption{Histograms of p-values of the Facebook graphs and journal 
  citation graphs. Left: histogram 
of the 556 neighborhood graphs in the Carnegie Mellon subnetwork that
have between 30 and 40 nodes. Most of the graphs have significant
community structure according to the EZ test. Right:
histogram of the p-values for subnetworks of the 
citation data. In the case, the p-values are typically large, indicating
a lack of community structure.}
\label{fig:pvalues}
\end{figure}

\subsection{Citation networks from statistics journals}

The data used to illustrate the proposed test in this section are 
associated with citations from several
statistics journals, including
the Annals of Statistics, Biometrika, the Journal of the American
Statistical Association, and the Journal of the Royal Statistical
Society, Series B. The citations are from papers published between 2003 and
2012 \citep{ji2016}.

We work with the ``giant component'' of the citation network from this
dataset, where each node in the network corresponds to one of 2,654
authors.  A directed edge from author $i$ to $j$
indicates that author $i$ has cited one or more papers by author $j$.
We extract subnetworks $A^{(a)} = (A^{(a)}_{ij})$ for a given author
$a$. This graph is over the authors cited by $a$, with
$A^{(a)}_{ij}=1$ if $i$ cites $j$ (or vice-versa), and
$A^{(a)}_{ij}=0$ otherwise.

\begin{figure}[!ht]
\begin{center}
\vskip-20pt
\begin{tabular}{c}
\includegraphics[width=.58\textwidth]{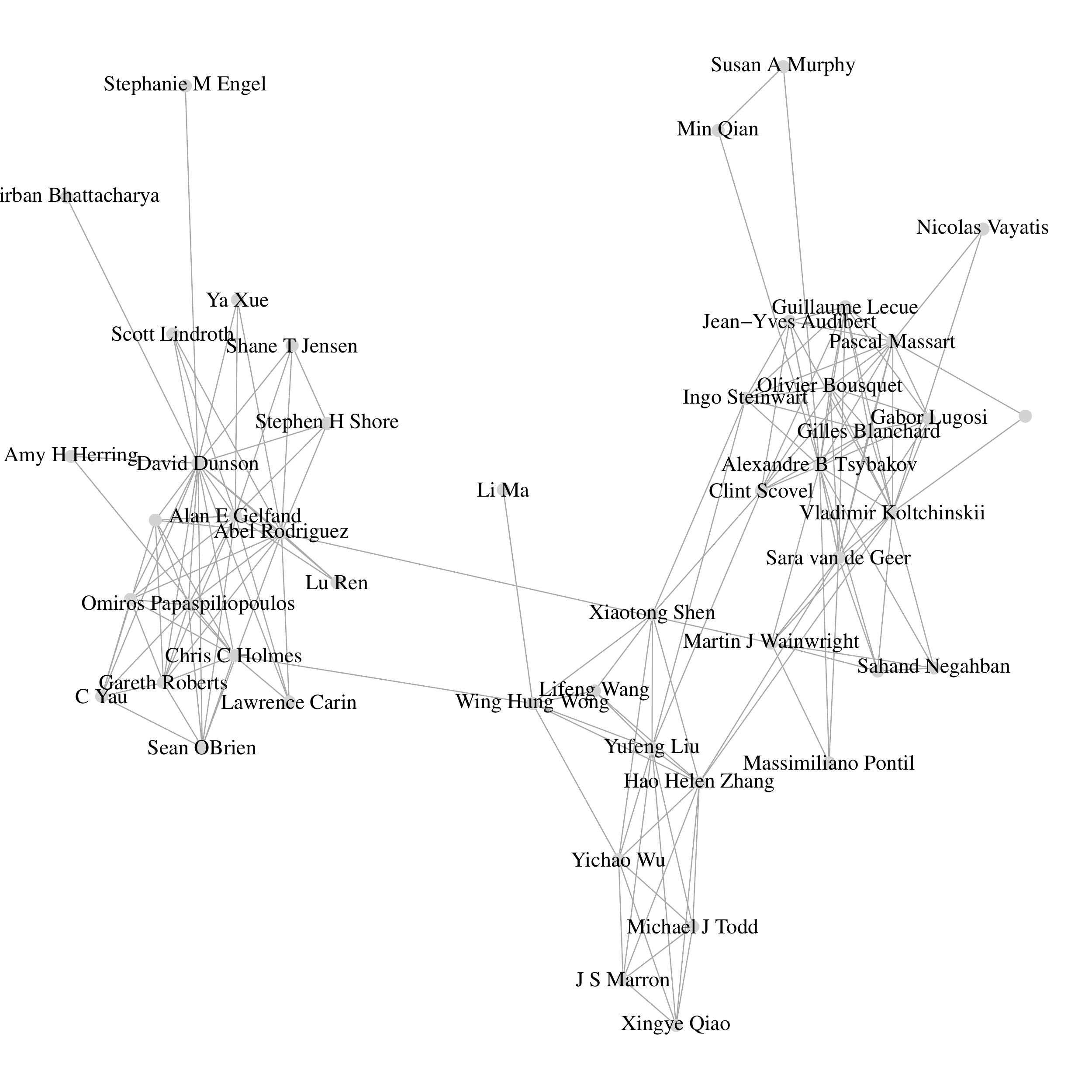}\\[-25pt]
\includegraphics[width=.58\textwidth]{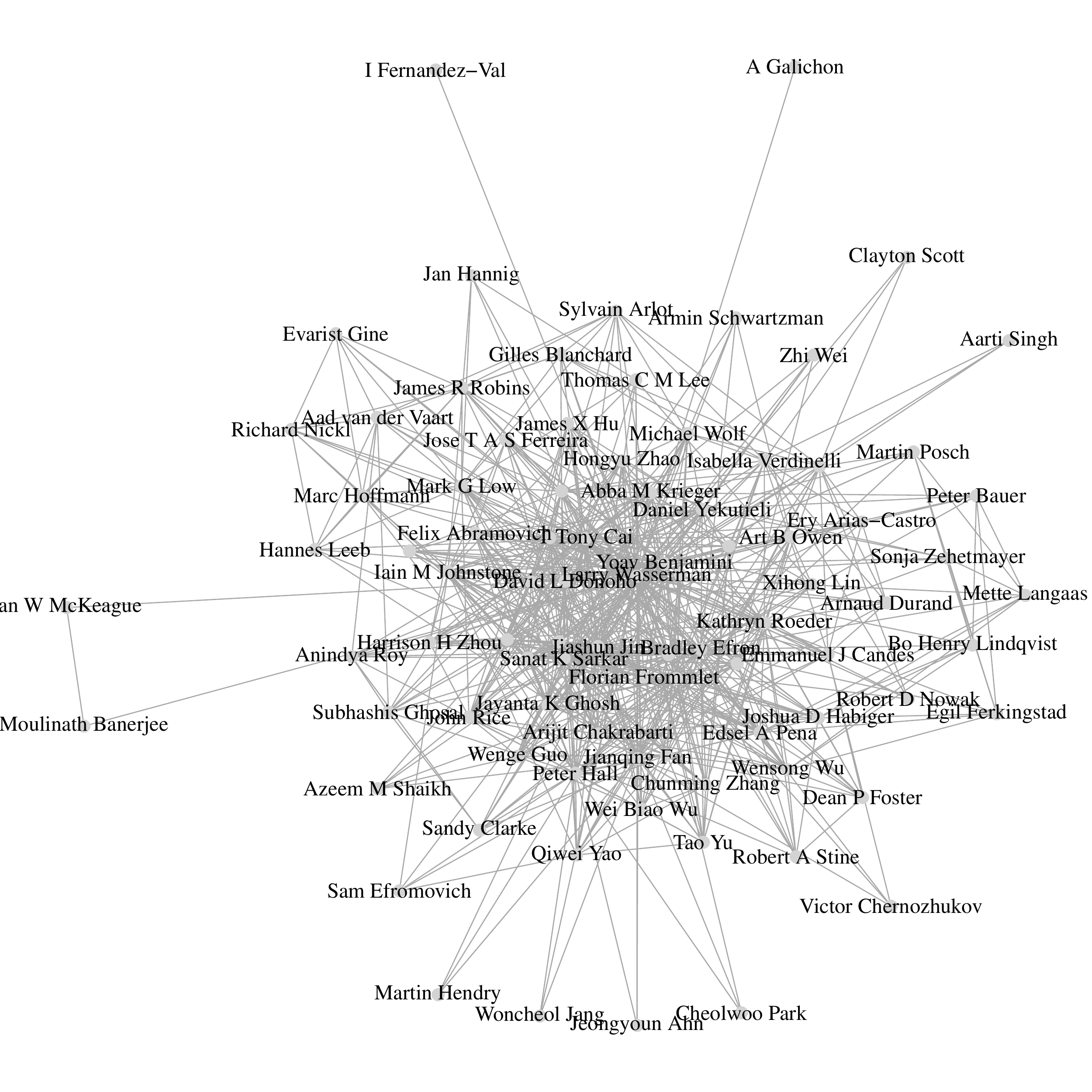}
\end{tabular}
\end{center}
\caption{Two networks from the statistics citation data. Top: induced
  network for the statistician Michael~I.~Jordan (Berkeley), with a
  small p-value of $9\times 10^{-5}$ and an EZ score of
  $3.91$. Clustering structure is evident in the subnetwork, with
  three groups that might be labeled ``nonparametric Bayes,''
  ``statistical learning theory'' and ``North Carolina statistics.''
  Bottom: network for Christopher Genovese (Carnegie Mellon), with a
  p-value of $2\times 10^{-12}$. The EZ score for this network 
  has a negative value of $-7.03$, indicating a model with $a < b$,
  and a smaller number of triangles compared with a configuration model.  The
  nodes on the periphery are highly connected to the inner cluster,
  but not highly connected among themselves.  (Some of the node labels
  have been removed for better readability.)}
\label{fig:citation-graphs}
\end{figure}

For these data, 387 authors have induced graphs of size 15 or larger.
A histogram of the p-values for these graphs is shown in the right
plot of Figure~\ref{fig:pvalues}. This shows that the p-values
are much closer to uniform, indicating a general lack of community structure,
compared with the Facebook networks. Two sample graphs that have small
p-values are shown in Figure~\ref{fig:citation-graphs}.
The top graph shows the induced graph for Berkeley statistician
Michael~I.~Jordan (Berkeley). This graph has 
a p-value of $9\times 10^{-5}$ and an EZ score of
$3.91$. Three communities are apparent in the subnetwork; 
knowledge of the authors in these groups leads one
to interpret them as ``nonparametric Bayes,''
``statistical learning theory'' and ``North Carolina statistics.''
A similar discovery is also made by \cite{jin2017estimating} with a different analysis under a mixed-membership model.
The bottom network is interesting because it has 
a large negative EZ score of $-7.03$. This is the network associated with
Christopher Genovese, from Carnegie Mellon University. The dense
core of the network includes researchers from CMU, and
other statisticians who work in similar areas. The 
nodes on the periphery are highly connected to the inner
cluster, but very weakly connected among themselves.
Thus, researchers in this outer group tend not to cite each other, but
cite (or are cited by) researchers in the inner group.

\subsection{Correlations among stocks on the S\&P 500}

%\begin{figure}
%\begin{center}
%\begin{tabular}{cc}
%\includegraphics[width=.45\textwidth]{sp500/sp500-corr-histogram} &
%\raise80pt\hbox{\texttt{  [p-value plot will go here]}}
%\end{tabular}
%\end{center}
%\caption{Left: Histograms of the Spearman rank correlations 
%for the S\&P 500 data. When we compute p-values for subnetworks of
%companies, we restrict to graphs among companies having correlations 
%of 0.45 or higher, corresponding to the 95th percentile. Right:
%histogram of p-values of all subnetworks having beween 20 and
%100 nodes, using the EZ score of equation \eqref{eq:gezscore}.}
%\label{fig:pvalues2}
%\end{figure}
%

In this third illustration, we treat correlations among the returns of
stocks on the S\&P 500.  The data are derived from prices posted on
the Yahoo Finance site, finance.yahoo.com.  The daily closing prices
were obtained for 452 stocks that consistently were in the S\&P 500
index between January 1, 2003 through January 1, 2011.  We restrict to
the subset of the data between January 1, 2003 to January 1, 2008,
before the onset of the 2008--2009 financial crisis. After this event,
stocks tended to become much more tightly correlated as investors
became more cautious.  We consider the variables $X_{t,j} =
\log(S_{t,j} /S_{t−1,j})$ where $S_{t,j}$ denotes the closing price of
stock $j$ on day $t$. The 452 stocks are categorized into 10 Global
Industry Classification Standard (GICS) sectors, including Consumer
Discretionary (70 stocks), Consumer Staples (35 stocks), Energy (37
stocks), Financials (74 stocks), Health Care (46 stocks), Industrials
(59 stocks), Information Technology (64 stocks), Materials (29
stocks), Telecommunications Services (6 stocks), and Utilities (32
stocks).

Stocks within a sector are generally strongly correlated--stocks
within an industry tend to move together. For a given company,
we form a subnetwork that is analogous to the Facebook subnetworks
considered in Section~\ref{sec:friends} (see also Section~\ref{sec:nbhds}).
Specifically, we create an edge $A_{ij}=1$ between companies $i$ 
and $j$ if their returns are strongly correlated, and
each is strongly correlated with the given company. We
take ``strongly correlated'' to mean
a Spearman rank correlation in the 95th percentile,  
which is a correlation above 0.45.

Example networks are shown in Figure~\ref{fig:sp500graphs}; three have
apparent community structure, one shows no apparent community
structure. For example, the upper right figure shows the correlation
graph for Consol Energy, which is based outside of Pittsburgh, PA, and
has interests in coal and natural gas production; its customers
include electric utilities and steel mills.\footnote{Wikipedia,
  https://en.wikipedia.org/wiki/Consol\_Energy} Two communities are
seen in the graph, with companies from the Energy sector (including
oil companies and other energy companies) and the Materials sector
(including the US Steel Corporation and Alcoa Inc.).

\begin{figure}[!ht]
\begin{center}
\begin{tabular}{cc}
{\footnotesize SCANA Corp, $p=1.7\times 10^{-7}$} &
{\footnotesize CONSOL Energy, $p=0.0043$} \\
\includegraphics[width=.50\textwidth]{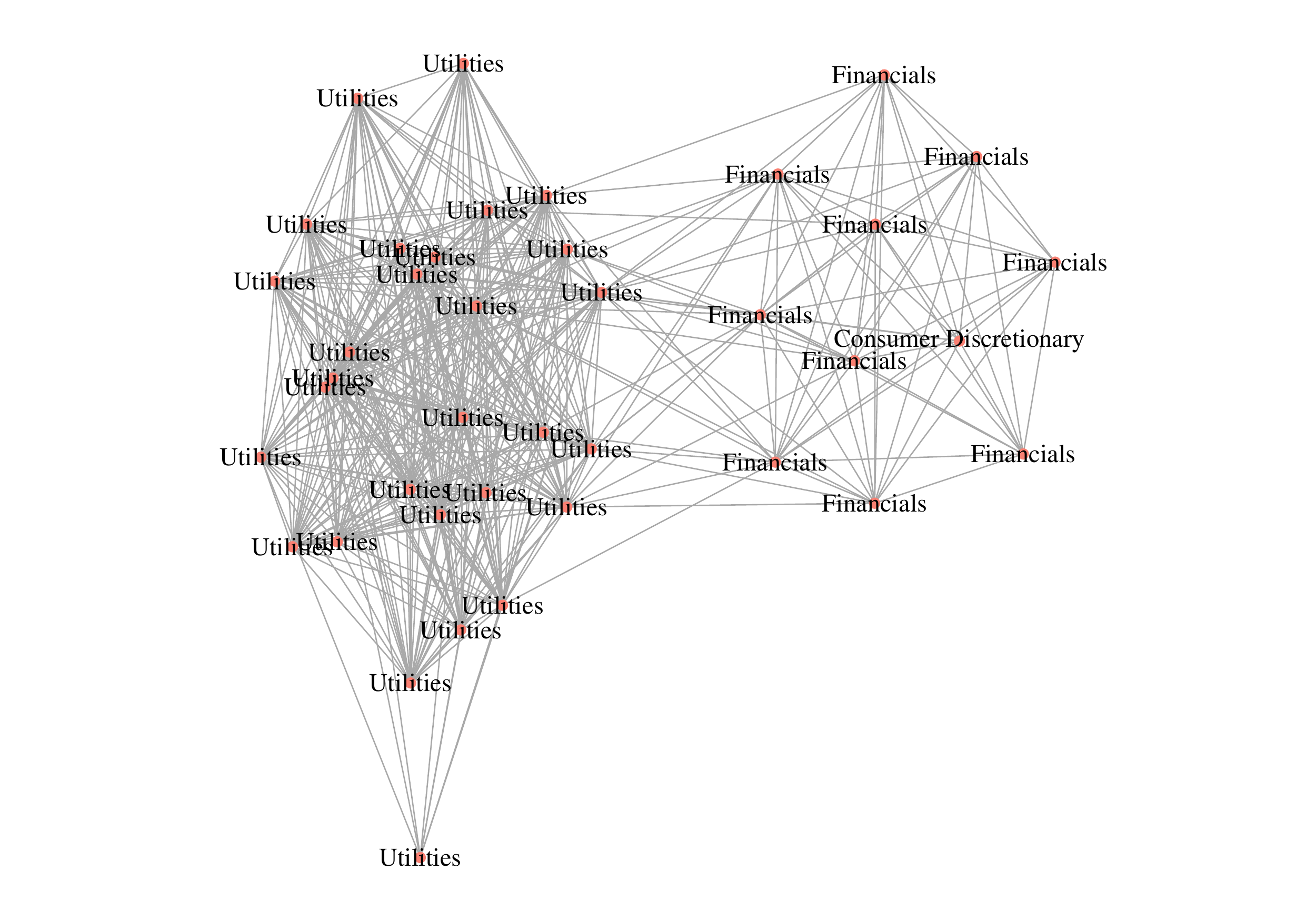}&
\includegraphics[width=.50\textwidth]{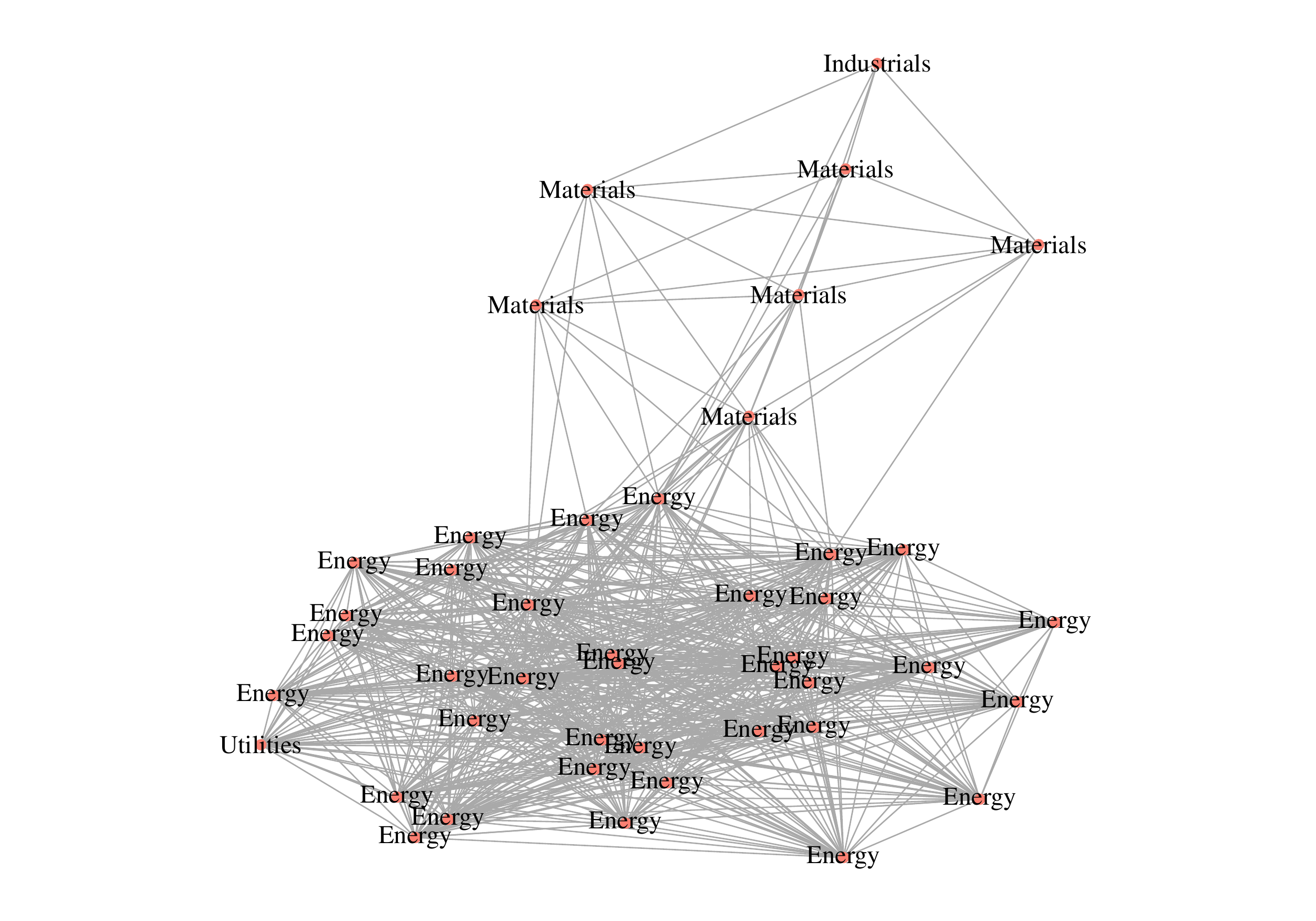}\\[20pt]
{\footnotesize Agilent Technologies, $p=0.0095$} &
{\footnotesize Family Dollar Stores, $p=0.97$}\\
\includegraphics[width=.50\textwidth]{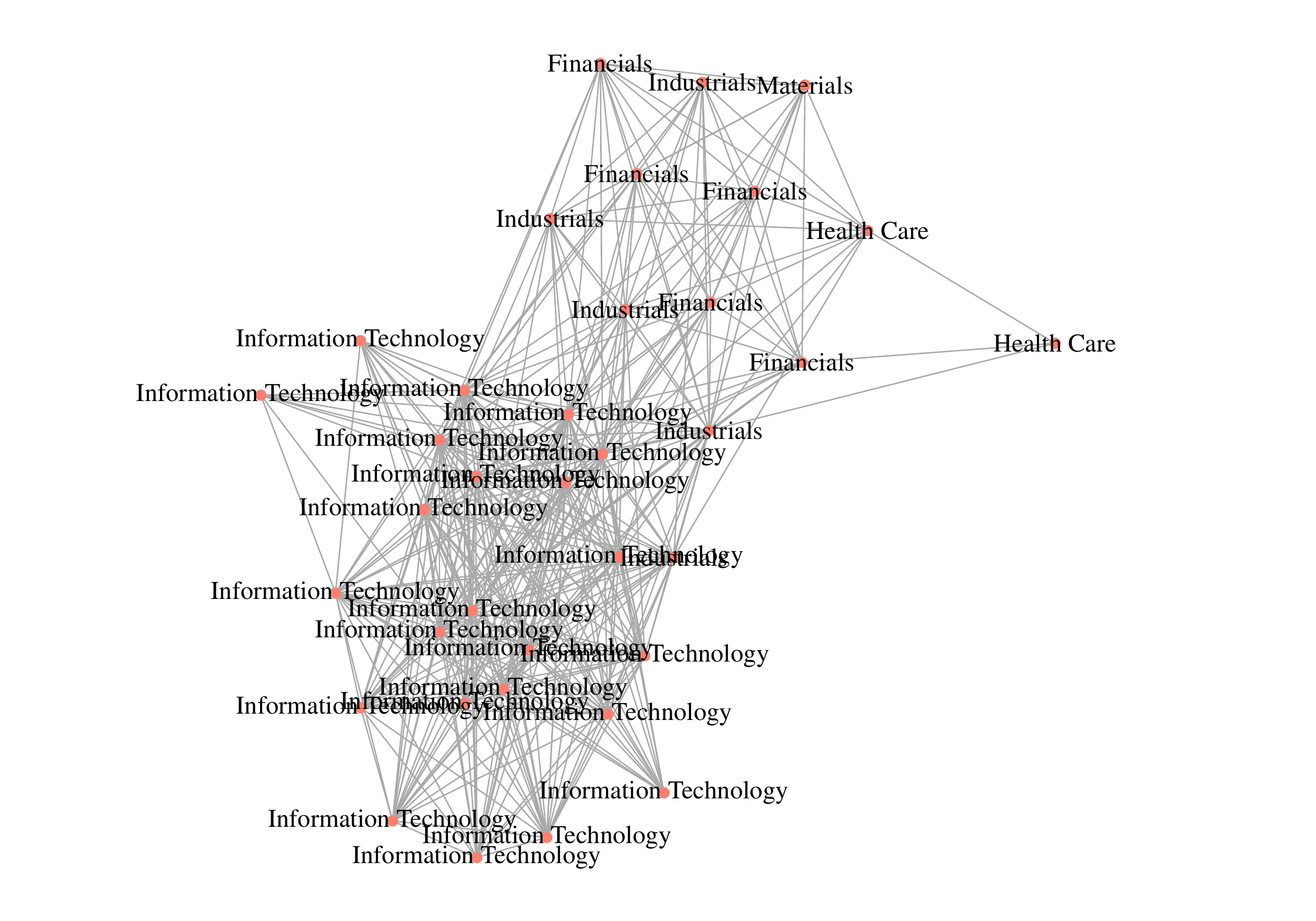} &
\includegraphics[width=.50\textwidth]{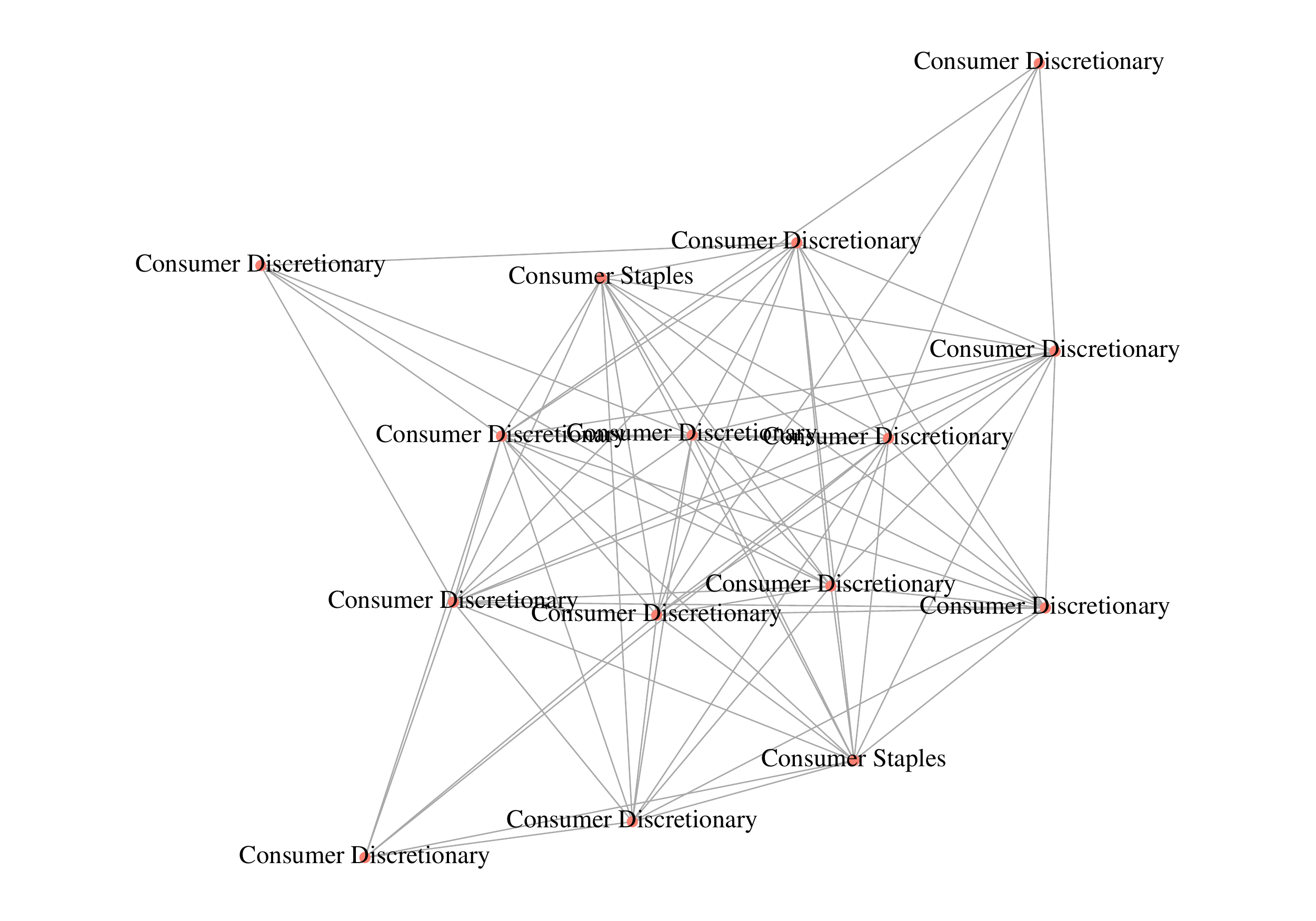} 
\end{tabular}
\caption{Example correlation graphs for companies on the S\&P 500.
For a given company we form a subnetwork among all companies having
high correlation (larger than $0.45$) with that company. We then 
compute the EZ score and p-value for this correlation graph. This figure shows four representative
subnetworks, three with relatively small p-values and cluster
structure, and one with a large p-value and no apparent cluster
structure. The node labels are the GICS industry of the company
corresponding to that node.
}
\end{center}
\label{fig:sp500graphs}
\end{figure}

\section{Summary}

Our results show how global structural characteristics of networks can
be inferred from local subgraph frequencies, without requiring the
global community structure to be explicitly estimated.  We develop
a testing framework based on the simple invariant $\chi_{ez} = T - (V/E)^3 = 0$
satisfied by all configuration models that lack community
structure. Our theory indicates that the signal-to-noise ratio
required for the hypothesis test to find community structure, when it
is present, is weaker than what is required by existing procedures
that rely on explicitly estimating the communities. Experiments with
social network data, scientific citations and equity returns show
that the test can be very effective. Our findings shed light on the
question of how global graph properties are reflected in local
subgraph statistics.  Lower bounds for detecting community structure,
as well as the development of more powerful tests for particular
settings such as Gaussian data and time series, are promising
directions for further study.

\section*{Acknowledgements}

The authors thank Rina Barber and Tracy Ke for many helpful comments
on this work, including the suggestion of the configuration null
model. We also thank Fengnan Gao for suggesting the martingale central
limit theorem in \cite{hall2014martingale}, and Scarlett Li for help
with the simulations. The research of JL is supported in part by NSF grant
DMS-1513594 and ONR grant N00014-12-1-0762. The research of CG is supported in part by NSF grant DMS-1712957.

\setlength{\bibsep}{3pt}
\bibliographystyle{apalike}
\def\refname{{\normalsize References}}
\begin{small}
\bibliography{ref}
\end{small}

\appendix
\section{An Application to Neighborhood Graphs}
\label{sec:nbhds}

The Facebook friend networks we study in the paper are all neighborhood graphs. In this section, we 
introduce a rigorous probabilistic setting for studying neighborhood graphs and derive an analogous EZ characterization.

Consider an adjacency matrix $\{A_{ij}\}_{0\leq i<j\leq n}$ of $n+1$ nodes. For the $0$th node, its neighborhood nodes are $\mathcal{M}=\left\{i\in[n]:A_{0i}=1\right\}$. Then, the neighborhood graph is the induced subgraph of $\mathcal{M}$.

Suppose there is a clustering structure in the network. We are interested in testing whether the $0$th node belongs a single cluster or multiple clusters. We propose the following natural setting for this problem. First, assume $\{A_{ij}\}_{1\leq i<j\leq n}$ are sampled from a DCBM with parameters $a,b,k$ and distribution $\mathcal{W}$ for the degree latent variable with $\mathbb{E}W^2=1$. Recall that $Z_i$ is a uniform random variable that takes value in $[k]$, and it stands for the label of the cluster that the $i$th node belongs to. Then, there is a subset $\mathcal{R}\subset[k]$ with cardinality $|\mathcal{R}|=r$, such that $A_{0i}\sim\text{Bernoulli}(p)$ if $Z_i\in\mathcal{R}$ and $A_{0i}=0$ otherwise independently for each $i\in[n]$. Because of the symmetry of the model, we can assume $\mathcal{R}=[r]$ without loss of generality. The model setting implicitly requires $r\leq k$.

The subgraph frequencies of edge, V-shape and triangle in the neighborhood graph are
\begin{eqnarray}
\label{eq:E|} E &=& \mathbb{P}(A_{12}=1|1,2\in\mathcal{M}),\\
\label{eq:V|} V &=& \mathbb{P}(A_{12}A_{13}=1|1,2,3\in\mathcal{M}),\\
\label{eq:T|} T &=& \mathbb{P}(A_{12}A_{13}A_{23}=1|1,2,3\in\mathcal{M}).
\end{eqnarray}
With direct calculation, the Erd\H{o}s-Zuckerberg characterization also holds for the above definitions.
\begin{proposition}
Under the setting described above,
\begin{eqnarray}
\label{eq:E+|} E &=&  (\mathbb{E}W)^2\left(\frac{1}{r}a+\frac{r-1}{r}b\right),\\
\label{eq:V+|} V &=& (\mathbb{E}W)^2\left(\frac{1}{r}a+\frac{r-1}{r}b\right)^2, \\
\label{eq:T+|} T &=& \frac{1}{r^2}a^3+\frac{3(r-1)}{r^2}ab^2+\frac{(r-1)(r-2)}{r^2}b^3.
\end{eqnarray}
As a consequence,
$$T-\left(\frac{V}{E}\right)^3=\frac{(r-1)(a-b)^3}{r^3}.$$
\end{proposition}
We can use this result to test whether the $0$th node belongs a single cluster or not, which is equivalent to testing whether $(r-1)(a-b)=0$ or not.
The forms given by (\ref{eq:E|})-(\ref{eq:T|}) suggests empirical subgraph frequencies using only the neighborhood graph. For example, since
$$E=\frac{\mathbb{P}(A_{12}=1, A_{01}=1, A_{02}=1)}{\mathbb{P}(A_{01}=1, A_{02}=1)}=\frac{\mathbb{E}A_{12}A_{01}A_{02}}{\mathbb{E}A_{01}A_{02}},$$
a natural estimator for $E$ is
\begin{equation}
\hat{E}=\frac{\frac{1}{{n\choose 2}}\sum_{1\leq i<j\leq n}A_{ij}A_{0i}A_{0j}}{\frac{1}{{n\choose 2}}\sum_{1\leq i<j\leq n}A_{0i}A_{0j}}=\frac{1}{{m\choose 2}}\sum_{1\leq i<j\leq n}A_{ij}A_{0i}A_{0j},\label{eq:def-E-NG}
\end{equation}
where $m=|\mathcal{M}|$ so that $\sum_{1\leq i<j\leq n}A_{0i}A_{0j}={m\choose 2}$. Similarly, estimators for $V$ and $T$ are
\begin{eqnarray}
\label{eq:def-V-NG}\hat{V} &=& \frac{1}{3{m\choose 3}}\sum_{1\leq i<j<k\leq n}(A_{ij}A_{ik}+A_{ij}A_{jk}+A_{ik}A_{jk})A_{0i}A_{0j}A_{0k}, \\
\label{eq:def-T-NG}\hat{T} &=& \frac{1}{{m\choose 3}}\sum_{1\leq i<j<k\leq n}A_{ij}A_{jk}A_{ik}A_{0i}A_{0j}A_{0k}.
\end{eqnarray}
The $\hat{E},\hat{V},\hat{T}$ are subgraph frequencies for the neighborhood graph.

\begin{theorem}\label{thm:NG}
Assume $\mathbb{E}W^4=O(1)$, $\frac{nrpa}{k}\rightarrow\infty$, $p=o(1)$ and $a^6\asymp b^6=o\left(\left(\frac{k}{nr}\right)^3\wedge p^2\left(\frac{k}{nr}\right)^4\right)$. Suppose
$$\delta=\lim_n\frac{(r-1)(a-b)^3}{\sqrt{6}}\left(\frac{\mathbb{E}m}{r(a+(r-1)b)}\right)^{3/2}\in[0,\infty).$$
Then,
$$2\sqrt{{m\choose 3}}\left(\sqrt{\hat{T}}-\left(\frac{\hat{V}}{\hat{E}}\right)^{3/2}\right)\leadsto N(\delta,1).$$
\end{theorem}

Theorem \ref{thm:NG} can be viewed as an extension of Theorem \ref{thm:main}. The results of Theorem \ref{thm:main} are recovered when $r=k$ and $p=1$. In the setting of neighborhood graphs, the roles of $k$ and $n$ in Theorem \ref{thm:main} are replaced by $r$ and $m$ (or $\mathbb{E}m$) in Theorem \ref{thm:NG}. Similar to Theorem \ref{thm:power}, we also present a result for the asymptotic power of the test.
\begin{theorem}\label{thm:NG-power}
Assume $\mathbb{E}W^4=O(1)$, $\frac{nrpa}{k}\rightarrow\infty$, $p=o(1)$ and $a^6\asymp b^6=o\left(\left(\frac{k}{nr}\right)^3\wedge p^2\left(\frac{k}{nr}\right)^4\right)$. Suppose
\begin{equation}
\frac{(\mathbb{E}m)(a-b)^2}{r^{4/3}(a+b)}\rightarrow\infty.\label{eq:condition-m}
\end{equation}
Then, for any constant $t\asymp 1$, we have
$$\mathbb{P}\left(\left|2{\sqrt{{m\choose 3}}\left(\sqrt{\hat{T}}-\left(\frac{\hat{V}}{\hat{E}}\right)^{3/2}\right)}\right|>t\right)\rightarrow 1.$$
\end{theorem}

\section{Testing for Structure in Gaussian Correlations}
\label{sec:gauss}

In multivariate analysis, it is important to identify community structures in variables.
In this section, we show how an analogous EZ test can be derived to find community structure 
in multivariate data $X=(X_1,...,X_p)^T\in\R^p$. 
Specifically, we consider the multivariate Gaussian model $X \given \theta\sim
N(0,\Sigma)$, where $\Sigma$ is a $p\times p$ covariance matrix with
diagonal entries $1$, and off-diagonal entries 
$\Sigma_{jl}=\theta_{jl}$ where $\theta$ follows the DCBM. Thus, we model
the correlation matrix as
$\text{Corr}(X_j,X_l \given \theta_{jl})=\theta_{jl}$ for $j\neq l$, 
where 
\begin{equation}
\theta_{jl} \given W, Z = W_j W_l (a \indicator(Z_j=Z_l) + b \indicator(Z_j
\neq Z_l),
\end{equation}
as in \eqref{eq:def-dcbm}.
A similar model for variable clustering was considered in \cite{bunea2015minimax} without the
latent variables $\{W_j\}$.
In this setting, under the null hypothesis of no community structure
($a=b$ or $k=1$), the covariance $\Sigma$ can be decomposed into
the sum of a diagonal matrix and a rank-one matrix. This is the
spiked covariance model commonly adopted in the PCA literature
\citep{tipping1999probabilistic,johnstone2009consistency}.

Defining
\begin{align}
E &= \E\bigl(\E(X_1 X_2 \given \theta)\bigr), \\[5pt]
V &= \E\bigl(\E(X_1 X_2 \given \theta)\, \E(X_2 X_3 \given \theta)\bigr), \\[5pt]
T &= \E\bigl(\E(X_1 X_2 \given \theta)\, \E(X_2 X_3 \given \theta)\, \E(X_1 X_3 \given \theta)\bigr),
\end{align}
we see that the same relations \eqref{eq:E+}--\eqref{eq:EZ} 
stated in Proposition~\ref{prop:EZ} hold in this Gaussian setting.
To estimate $E,V$ and $T$ from data,
we exploit Wick's formula \citep{isserlis:18,wick:50} in the form
$\E\left(\prod_{j=1}^{2h-1} X_j\right)  = 0$ and 
\begin{align}
\E\left(\prod_{j=1}^{2h} X_j\right) & = \sum \E \left(\prod \E(X_j X_l\given \theta)\right),
\end{align}
where the sum is over all ways of partitioning the components $X_1,\ldots, X_{2h}$ into disjoint pairs,
and the product is over those $h$ pairs. In particular, we have
\begin{align}
\nonumber
\E\left( X_1 X_2^2 X_3\right) &= 2\E\bigl(\E(X_1 X_2 \given \theta)\, \E(X_2 X_3 \given \theta)\bigr) 
 + \E\left( \E(X_1 X_3\given \theta)\, \E(X_2^2 \given \theta)\right) \\
&= 2V + E,
\end{align}
since we assume that $\E(X_j^2\given \theta)=1$ and $\E(W^2)=1$. Similarly, we have that
\begin{align}
\E\left(X_1^2 X_2^2 X_3^2\right) &= \E\left( \E(X_1^2\given \theta)^3\right) + 
6\E\left( \E(X_1 X_2 \given \theta)^2\E(X_3^2\given \theta)\right) + 8T\\
&= 1 + 6 \E\left( \E(X_1 X_2 \given \theta)^2 \right) + 8T\\
&= 1 + 3\left(\E(X_1^2 X_2^2) - 1\right) + 8T\\
&= 3\E\left(X_1^2 X_2^2\right) + 8T - 2.
\end{align}
Therefore, given an i.i.d. sample $\{X_i\}_{i=1}^n$ of size $n$, unbiased estimates of $E$, $V$, and $T$ are given by
$\hat E = \frac{1}{n}\sum_{i=1}^n \hat E_i$, 
$\hat V = \frac{1}{n}\sum_{i=1}^n \hat V_i$,
and $\hat T = \frac{1}{n}\sum_{i=1}^n \hat T_i$,
%\begin{equation}
%\hat E = \frac{1}{n}\sum_{i=1}^n \hat E_i, \;\; 
%\hat V = \frac{1}{n}\sum_{i=1}^n \hat V_i,\; \mbox{and\ \ }
%\hat T = \frac{1}{n}\sum_{i=1}^n \hat T_i,
%\end{equation}
where 
\begin{align}
\hat E_i &= \frac{1}{\binom{p}{2}} \sum_{j<l} X_{ij} X_{il}, \\[5pt]
\hat V_i &= \frac{1}{6 \binom{p}{3}} \sum_{j<l<m} 
   \left(X_{ij}^2 X_{il} X_{im} + X_{ij} X_{il}^2 X_{im} + X_{ij} X_{il} X_{im}^2\right)
- \frac{1}{2} \hat E_i,\\
\hat T_i &= \frac{1}{8\binom{p}{3}} \sum_{j<l<m}  X_{ij}^2 X_{il}^2 X_{im}^2 
- \frac{3}{8 \binom{p}{2}} \sum_{j<l} X_{ij}^2 X_{il}^2 + \frac{1}{4}.
\end{align}
Let $A$ be a block-diagonal matrix with the $i$th block 
\begin{equation*}
A^{(i)}_{jl} = 
\begin{cases}
X_{ij} X_{il} & \text{if $j \neq l$}, \\
0 & \text{otherwise.}
\end{cases}
\end{equation*}
Then, as in \eqref{eq:comp1}--\eqref{eq:comp3}, the quantities $\hat E$, $\hat V$, and $\hat
T$ can be computed using matrix operations as 
\begin{align}
\hat E &= \frac{1}{2 n\binom{p}{2}} \mip{\ones}{A},\\
\hat V &= \frac{1}{12 n \binom{p}{3}} \left(\mip{\ones}{A^2} -
\tr(A^2)\right) - \frac{1}{2} \hat E, \\
\hat T &= \frac{1}{48 n\binom{p}{3}} \tr(A^3)
- \frac{3}{16 n\binom{p}{2}} \|A\|^2 + \frac{1}{4},
\end{align}
where $\|A\| = \sqrt{\mip{A}{A}}$ is the
Frobenius norm.

As before, we reject the null hypothesis once the magnitude of the testing
statistic $\hat\chi_{ez} = \hat{T}-\left(\sfrac{\hat{V}}{\hat{E}}\right)^3$ passes a threshold.
For the network models discussed in Section~\ref{sec:eztest}, the
square root transformation
automatically normalizes the testing statistic, as shown in Theorem
\ref{thm:main}. Here, we need a different normalization for the
Gaussian covariance model. Under some mild conditions, we have the
decomposition
$$\hat\chi_{ez} - \chi_{ez} = \frac{1}{n}\sum_{i=1}^n\Delta_i+r_n,$$
where $r_n$ is a negligible term, and 
$$\Delta_i=\hat T_i-\E (\hat T_i\given \theta)-3\frac{V^2}{E^3}(\hat V_i-\E
(\hat V_i\given \theta))+3\frac{V^3}{E^4}(\hat E_i-\E (\hat E_i \given \theta)).$$
This suggests a natural estimator of the variance given by
$$\hat{\sigma}^2=\frac{1}{n-1}\sum_{i=1}^n(Q_i-\bar{Q}_n)^2,$$
where 
$$Q_i=T_i-3\frac{\hat{V}^2}{\hat{E}^3}\hat V_i+3\frac{\hat{V}^3}{\hat{E}^4}\hat E_i,$$
and $\bar{Q}_n=\frac{1}{n}\sum_{i=1}^nQ_i$.

\begin{theorem}\label{thm:CLT-cov}
Assume $\E W^4=O(1)$, $p^{-1}\vee n^{-1/2}\ll |a|\asymp |b|$, and
$a^4=o(1\wedge (p/n))$, and suppose that
$$\delta=\lim_n\frac{\sqrt{n}(k-1)(a-b)^3}{k^3\sqrt{\frac{9}{32}\left(\frac{1}{k}a^2+\frac{k-1}{k}b^2\right)}}\in[0,\infty).$$
Then as $n\longrightarrow\infty$, we have
\begin{equation}
\frac{\sqrt{n}\left(\hat{T}-\left(\frac{\hat{V}}{\hat{E}}\right)^3\right)}{\sqrt{\hat{\sigma}^2}}\leadsto N(\delta,1).
\label{eq:gezscore}
\end{equation}
\end{theorem}

Theorem \ref{thm:CLT-cov} characterizes the asymptotic behavior of the
testing statistic. Note that for a growing number of communities $k$,
the magnitude of the mean $|\delta|$ is of order
$\sfrac{\sqrt{n}|a-b|^3}{k^2a}$, which is the natural signal-to-noise
ratio of the problem. When this quantity tends to infinity, the power
of the test approaches one.
\begin{theorem}\label{thm:power-cov}
Assume $\E W^4=O(1)$, $p^{-1}\vee n^{-1/2}\ll |a|\asymp |b|$, and $a^4=o(1\wedge (p/n))$. Suppose
$$\frac{\sqrt{n}|a-b|^3}{k^2a}\longrightarrow\infty.$$
Then, for any constant $t\asymp 1$, we have
$$\P\left(\left|\sqrt{\frac{n}{\hat{\sigma}^2}} \left(\hat{T}-\left(\frac{\hat{V}}{\hat{E}}\right)^3\right)\right|>t\right)\longrightarrow 1.$$
\end{theorem}

\section{Proofs}

In this section, we give proofs of all theorems in the paper. The main results are proved in Section \ref{sec:pf-main} with the assistance of some technical lemmas, whose proofs are deferred to Section \ref{sec:aux} and Section \ref{sec:pf-tech}.

\subsection{Proofs of Main Results}\label{sec:pf-main}

The proofs of Theorem \ref{thm:main} and Theorem \ref{thm:power} requires the following lemmas, whose proofs will be given in Section \ref{sec:aux}.
\begin{lemma}\label{lem:EV-order}
Assume $\mathbb{E}W^4=O(1)$ and $n^{-1}\ll a\asymp b=o(1)$. Then
$$\mathbb{E}(\hat{E}-E)^2=O\left(\frac{a^2}{n}\right),\quad\text{and}\quad\mathbb{E}(\hat{V}-V)^2=O\left(\frac{a^4}{n}\right).$$
\end{lemma}

\begin{lemma}\label{lem:T-order}
Assume $\mathbb{E}W^4=O(1)$ and $n^{-1}\ll a\asymp b\ll n^{-2/3}$. Then
$$\mathbb{E}(\hat{T}-T)^2\asymp \frac{a^3}{n^3}\quad \text{and}\quad \frac{\sqrt{{n\choose 3}}(\hat{T}-T)}{\sqrt{T}}\leadsto N(0,1).$$
\end{lemma}

\begin{proof}[Proof of Theorem \ref{thm:main}]
We first give the expansion 
\begin{eqnarray}
\label{eq:expdecomp} \hat{T}-\left(\frac{\hat{V}}{\hat{E}}\right)^3 &=& T-\left(\frac{V}{E}\right)^3 + (\hat{T}-T) \\
\nonumber && - 3\left(\frac{V}{E}\right)^2\frac{\hat{V}-V}{E} - 3\left(\frac{V}{E}\right)^2\left(\frac{1}{\hat{E}}-\frac{1}{E}\right)V \\
\nonumber && - 3\left(\frac{V}{E}\right)^2(\hat{V}-V)\left(\frac{1}{\hat{E}}-\frac{1}{E}\right) \\
\nonumber && -3\left(\frac{V}{E}\right)\left(\frac{\hat{V}}{\hat{E}}-\frac{V}{E}\right)^2 - \left(\frac{\hat{V}}{\hat{E}}-\frac{V}{E}\right)^3.
\end{eqnarray}
By Lemma \ref{lem:EV-order} and Lemma \ref{lem:T-order}, $T-\left(\frac{V}{E}\right)^3+\hat{T}-T$ is the dominating term. Therefore,
$$\frac{\sqrt{{n\choose 3}}\left(\hat{T}-\left(\frac{\hat{V}}{\hat{E}}\right)^3\right)}{\sqrt{T}}\leadsto N(\delta,1),$$
where
$$\delta=\lim_n \frac{\sqrt{{n\choose 3}}\left({T}-\left(\frac{{V}}{{E}}\right)^3\right)}{\sqrt{T}}.$$
Moreover, under the assumption, we have
$$\frac{\hat{T}}{T}=(1+o_P(1))\frac{\left(\frac{\hat{V}}{\hat{E}}\right)^3}{T}=(1+o_P(1)).$$
This gives the desired results.
\end{proof}

\begin{proof}[Proof of Theorem \ref{thm:power}]
Similar to the argument that we have used in the proof of Theorem \ref{thm:main}, $T-\left(\frac{V}{E}\right)^3+\hat{T}-T$ is the dominating term of $\hat{T}-\left(\frac{\hat{V}}{\hat{E}}\right)^3$. Thus, the dominating term of $2{\sqrt{{n\choose 3}}\left(\sqrt{\hat{T}}-\left(\frac{\hat{V}}{\hat{E}}\right)^{3/2}\right)}$ is
$$2\sqrt{{n\choose 3}}\frac{T-\left(\frac{V}{E}\right)^3+\hat{T}-T}{\sqrt{\hat{T}}+\left(\frac{\hat{V}}{\hat{E}}\right)^{3/2}}.$$
Under the assumption,
$$\left|2\sqrt{{n\choose 3}}\frac{T-\left(\frac{V}{E}\right)^3}{\sqrt{\hat{T}}+\left(\frac{\hat{V}}{\hat{E}}\right)^{3/2}}\right|\rightarrow\infty,$$
in probability, and
$$\left|2\sqrt{{n\choose 3}}\frac{\hat{T}-T}{\sqrt{\hat{T}}+\left(\frac{\hat{V}}{\hat{E}}\right)^{3/2}}\right|=O_P(1).$$
Hence, the desired result follows.
\end{proof}

To prove Theorem \ref{thm:NG} and Theorem \ref{thm:NG-power}, we need the following lemmas, and their proofs will be given in Section \ref{sec:aux}. Recall the definitions of $\hat{E}$, $\hat{V}$ and $\hat{T}$ in (\ref{eq:def-E-NG})-(\ref{eq:def-T-NG}).

\begin{lemma}\label{lem:EV-NG}
Assume $\mathbb{E}W^4=O(1)$, $\frac{nrpa}{k}\rightarrow\infty$, $p=o(1)$ and $a\asymp b=o(1)$. Then
$$(\hat{E}-E)^2=O_P\left(a^2\left(\frac{k}{prn}\right)\right)\quad\text{and}\quad(\hat{V}-V)^2=O_P\left(a^4\left(\frac{k}{prn}\right)\right).$$
\end{lemma}

\begin{lemma}\label{lem:T-NG}
Assume $\mathbb{E}W^4=O(1)$, $\frac{nrpa}{k}\rightarrow\infty$, $p=o(1)$ and $a^6\asymp b^6=o\left(\left(\frac{k}{nr}\right)^3\wedge p^2\left(\frac{k}{nr}\right)^4\right)$. Then,
$$\mathbb{E}(\hat{T}-T)^2\asymp \frac{a^3k^3}{n^3r^3}\quad \text{and}\quad \frac{\sqrt{{n\choose 3}\left(\frac{r}{k}\right)^3}(\hat{T}-T)}{\sqrt{T}}\leadsto N(0,1).$$
\end{lemma}

\begin{proof}[Proofs of Theorem \ref{thm:NG} and Theorem \ref{thm:NG-power}]
The results follow the same argument used in the proofs of Theorem \ref{thm:main} and Theorem \ref{thm:power}. We note that under the required conditions, we have
$$\frac{{m\choose 3}}{{n\choose 3}\left(\frac{r}{k}\right)^3}=1+o_P(1),$$
so that ${n\choose 3}\left(\frac{r}{k}\right)^3$ can be replaced by ${m\choose 3}$.
\end{proof}

To prove Theorem \ref{thm:CLT-cov} and Theorem \ref{thm:power-cov}, we need the following lemmas, and their proofs will be given in Section \ref{sec:aux}.

\begin{lemma}\label{lem:EV-cov}
Assume $\mathbb{E}W^4=O(1)$ and $p^{-1}\ll |a|\asymp |b|=o(1)$. Then
$$(\hat{E}-E)^2=O_P\left(\frac{a^2}{n}+\frac{a^2}{p}\right),\quad (\hat{V}-V)^2=O_P\left(\frac{a^2}{n}+\frac{a^4}{p}\right).$$
\end{lemma}

\begin{lemma}\label{lem:T-cov}
Assume $\mathbb{E}W^4=O(1)$ and $p^{-1}\vee n^{-1/2}\ll |a|\asymp |b|=o(1\wedge (p/n)^{1/4})$. Then,
$$(\hat{T}-T)^2\asymp_P \frac{a^2}{n}\quad \text{and}\quad \frac{\hat{T}-T}{\sqrt{\mathbb{E}\Var(\hat{T}|\theta)}}\leadsto N(0,1).$$
\end{lemma}

\begin{proof}[Proofs of Theorem \ref{thm:CLT-cov} and Theorem \ref{thm:power-cov}]
By Lemma \ref{lem:EV-cov}, Lemma \ref{lem:T-cov}, and the expansion (\ref{eq:expdecomp}), $T-\left(\frac{V}{E}\right)^3+\hat{T}-T$ is the dominating term that we need to analyze under the conditions. Lemma \ref{lem:T-cov} implies that
$$\frac{\hat{T}-\left(\frac{\hat{V}}{\hat{E}}\right)^3}{\sqrt{\mathbb{E}\Var(\hat{T}|\theta)}}\leadsto N(\delta,1),$$
where $\delta=\lim_n\frac{\sqrt{n}\left(T-\left(\frac{V}{E}\right)^3\right)}{\sqrt{\frac{9}{32}\left(\frac{1}{k}a^2+\frac{k-1}{k}b^2\right)}}$. Finally, it is sufficient to show $\hat{\sigma}^2/\mathbb{E}\Var(\hat{T}|\theta)=1+o_P(1)$. By Lemma \ref{lem:EV-cov}, $\hat{\sigma}^2=(1+o_P(1))\frac{1}{n-1}\sum_{i=1}^n(T_i-\bar{T})^2$. Since $\mathbb{E}T_i^4=O(a^2)$, $\hat{\sigma}^2/\mathbb{E}\Var(\hat{T}|\theta)=1+o_P(1)$ holds under the condition. The proof of Theorem \ref{thm:power-cov} follows the same argument used in the proof of Theorem \ref{thm:power}.
\end{proof}

\subsection{Proofs of Auxiliary Results}\label{sec:aux}

This section gives proofs of all lemmas in Section \ref{sec:pf-main}. To better organize the proofs, we delay some lengthy calculations in propositions in Section \ref{sec:pf-tech}.

\begin{proof}[Proof of Lemma \ref{lem:EV-order}]
We introduce the notation $\eta_{ij}=a\mathbb{I}\{Z_i=Z_j\}+b\mathbb{I}\{Z_i\neq Z_j\}$. Then, $\theta_{ij}=W_iW_j\eta_{ij}$.

First, note that
$$\mathbb{E}(\hat{E}-E)^2=\mathbb{E}(\hat{E}-\mathbb{E}(\hat{E}|W,Z))^2+\mathbb{E}(\mathbb{E}(\hat{E}|W,Z)-\mathbb{E}(\hat{E}|Z))^2+\mathbb{E}(\mathbb{E}(\hat{E}|Z)-E)^2.$$
Thus, we need to bound the three terms on the right hand side of the above equality respectively. The first term has bound
$$\mathbb{E}(\hat{E}-\mathbb{E}(\hat{E}|W,Z))^2=\mathbb{E}\Var(\hat{E}|W,Z)\leq \frac{1}{{n\choose 2}^2}\sum_{i<j}\mathbb{E}\theta_{ij}=O\left(\frac{a}{n^2}\right).$$
By Proposition \ref{prop:wij}, the second term is
$$\mathbb{E}(\mathbb{E}(\hat{E}|W,Z)-\mathbb{E}(\hat{E}|Z))^2=\mathbb{E}\left(\frac{1}{{n\choose 2}}\sum_{i<j}\eta_{ij}\left(W_iW_j-(\mathbb{E}W)^2\right)\right)^2=O\left(\frac{a^2}{n}\right).$$
The last term has bound
\begin{eqnarray*}
\mathbb{E}(\mathbb{E}(\hat{E}|Z)-E)^2 &=& (\mathbb{E}W)^4\mathbb{E}\left(\frac{1}{{n\choose 2}}\sum_{i<j}(\eta_{ij}-\mathbb{E}\eta_{ij})\right)^2 \\
&\leq& 2(\mathbb{E}W)^4(a^2+b^2)\mathbb{E}\left(\frac{1}{{n\choose 2}}\sum_{i<j}(\mathbb{I}\{Z_i=Z_j\}-\mathbb{P}(Z_i=Z_j))\right)^2 \\
&=& O\left(\frac{a^2}{n}\right),
\end{eqnarray*}
by Proposition \ref{prop:z-foundation}.
Under the condition $n^{-1}\ll a$, we get $\mathbb{E}(\hat{E}-E)^2=O\left(\frac{a^2}{n}\right)$.

Next, we study $\mathbb{E}(\hat{V}-V)^2$. Again, it can be decomposed into the sum of three terms.
$$\mathbb{E}(\hat{V}-V)^2=\mathbb{E}(\hat{V}-\mathbb{E}(\hat{V}|W,Z))^2+\mathbb{E}(\mathbb{E}(\hat{V}|W,Z)-\mathbb{E}(\hat{V}|Z))^2+\mathbb{E}(\mathbb{E}(\hat{V}|Z)-V)^2.$$
To study the first term, note that
\begin{eqnarray*}
&& \hat{V}-\mathbb{E}(\hat{V}|W,Z) \\
&=& \frac{1}{{n\choose 3}}\sum_{i<j<l}\frac{(A_{ij}A_{jl}-\theta_{ij}\theta_{jl}) + (A_{ji}A_{il}-\theta_{ij}\theta_{il}) + (A_{il}A_{kj}-\theta_{il}\theta_{jl})}{3} \\
&=& \frac{1}{{n\choose 3}}\sum_{i<j<l}\frac{(A_{ij}-\theta_{ij})(A_{jl}-\theta_{jl})+(A_{ij}-\theta_{ij})(A_{il}-\theta_{il})+(A_{il}-\theta_{il})(A_{jl}-\theta_{jl})}{3} \\
&& + \frac{1}{{n\choose 2}}\sum_{i<j}\left[\frac{1}{n-2}\sum_{l\notin\{i,j\}}(\theta_{jl}+\theta_{il})\right](A_{ij}-\theta_{ij}).
\end{eqnarray*}
Write
$$L_{ijl}=\frac{(A_{ij}-\theta_{ij})(A_{jl}-\theta_{jl})+(A_{ij}-\theta_{ij})(A_{il}-\theta_{il})+(A_{il}-\theta_{il})(A_{jl}-\theta_{jl})}{3}.$$
It is not hard to see that $L_{ijl}$ and $L_{i'j'l'}$ are uncorrelated if the sets $\{i,j,l\}$ and $\{i',j',l'\}$ are different. Thus,
$$\mathbb{E}\left(\frac{1}{{n\choose 3}}\sum_{i<j<l}L_{ijl}\right)^2=\frac{1}{{n\choose 3}^2}\sum_{i<j<l}\mathbb{E}L_{ijl}^2=O\left(\frac{a^2}{n^3}\right).$$
Moreover,
$$\mathbb{E}\left(\frac{1}{{n\choose 2}}\sum_{i<j}\left[\frac{1}{n-2}\sum_{l\notin\{i,j\}}(\theta_{jl}+\theta_{il})\right](A_{ij}-\theta_{ij})\right)^2=O\left(\frac{a^3}{n^2}\right).$$
Thus,
$$\mathbb{E}(\hat{V}-\mathbb{E}(\hat{V}|W,Z))^2=O\left(\frac{a^3}{n^2}\right),$$
under the condition $n^{-1}\ll a$. The second term is
\begin{eqnarray*}
&& \mathbb{E}(\mathbb{E}(\hat{V}|W,Z)-\mathbb{E}(\hat{V}|Z))^2 \\
&=& \mathbb{E}\left(\frac{1}{{n\choose 3}}\sum_{i<j<l}\frac{(W_i^2W_jW_l-\mathbb{E}W^2(\mathbb{E}W)^2)\eta_{ij}\eta_{il}}{3}\right.\\
&& \left. + \frac{(W_j^2W_iW_l-\mathbb{E}W^2(\mathbb{E}W)^2)\eta_{ij}\eta_{jl}+(W_l^2W_jW_i-\mathbb{E}W^2(\mathbb{E}W)^2)\eta_{jl}\eta_{il}}{3}\right)^2. \\
&=& O\left(\frac{a^4}{n}\right),
\end{eqnarray*}
by Proposition \ref{prop:meixiangdao}.
The third term is
$$(\mathbb{E}W^2)(\mathbb{E}W)^2\Var\left(\frac{1}{{n\choose 3}}\sum_{i<j<l}\frac{\eta_{ij}\eta_{il}+\eta_{ij}\eta_{jl}+\eta_{il}\eta_{jl}}{3}\right)=O\left(\frac{a^4}{n}\right),$$
by Proposition \ref{prop:z-turbo}.
Under the condition $n^{-1}\ll a$, we have $\mathbb{E}(\hat{V}-V)^2=O\left(\frac{a^4}{n}\right)$.
\end{proof}

\begin{proof}[Proof of Lemma \ref{lem:T-order}]
Recall the notation $\eta_{ij}=a\mathbb{I}\{Z_i=Z_j\}+b\mathbb{I}\{Z_i\neq Z_j\}$ that we have used in Lemma \ref{lem:EV-order}. It is helpful to state the decomposition
$$\hat{T}-T=\hat{T}-\mathbb{E}(\hat{T}|W,Z)+\mathbb{E}(\hat{T}|W,Z)-\mathbb{E}(\hat{T}|Z)+\mathbb{E}(\hat{T}|Z)-T.$$
We are going to argue that $\hat{T}-\mathbb{E}(\hat{T}|W,Z)$ is the dominating term. We first give bounds on the order of $\mathbb{E}(\hat{T}|W,Z)-\mathbb{E}(\hat{T}|Z)$ and $\mathbb{E}(\hat{T}|Z)-T$. First, by Proposition \ref{prop:wijk},
$$\mathbb{E}\left(\mathbb{E}(\hat{T}|W,Z)-\mathbb{E}(\hat{T}|Z)\right)^2=\mathbb{E}\left(\frac{1}{{n\choose 3}}\sum_{i<j<l}\eta_{ij}\eta_{jl}\eta_{il}(W_i^2W_j^2W_l^2-(\mathbb{E}W^2)^3)\right)^2=O\left(\frac{a^6}{n}\right).$$
Then, by Proposition \ref{prop:z-v8},
$$\mathbb{E}\left(\mathbb{E}(\hat{T}|Z)-T\right)^2=(\mathbb{E}W^2)^3\Var\left(\frac{1}{{n\choose 3}}\sum_{i<j<l}\eta_{ij}\eta_{jl}\eta_{il}\right)=O\left(\frac{a^6}{n}\right).$$

Now we study $\hat{T}-\mathbb{E}(\hat{T}|W,Z)$. It has the following expansion.
\begin{eqnarray*}
\hat{T}-\mathbb{E}(\hat{T}|W,Z) &=& \frac{1}{{n\choose 3}}\sum_{i<j<l}(A_{ij}A_{jl}A_{il}-\theta_{ij}\theta_{jl}\theta_{il}) \\
&=& \frac{1}{{n\choose 3}}\sum_{i<j<l}(A_{ij}-\theta_{ij})(A_{jl}-\theta_{jl})(A_{il}-\theta_{il}) \\
&& + \frac{1}{{n\choose 3}}\sum_{i<j<l} \left[\theta_{ij}(A_{jl}-\theta_{jl})(A_{il}-\theta_{il})+\theta_{jl}(A_{ij}-\theta_{ij})(A_{il}-\theta_{il})\right. \\
&& \left.+ \theta_{il}(A_{ij}-\theta_{ij})(A_{jl}-\theta_{jl})\right] \\
&& + \frac{3}{{n\choose 2}}\sum_{i<j}\left[\frac{1}{n-2}\sum_{l\notin\{i,j\}}\theta_{il}\theta_{jl}\right](A_{ij}-\theta_{ij}).
\end{eqnarray*}
Write
$$H_{ijl}=\theta_{ij}(A_{jl}-\theta_{jl})(A_{il}-\theta_{il})+\theta_{jl}(A_{ij}-\theta_{ij})(A_{il}-\theta_{il})+ \theta_{il}(A_{ij}-\theta_{ij})(A_{jl}-\theta_{jl}).$$
It is not hard to see that $H_{ijl}$ and $H_{i'j'l'}$ are uncorrelated if the sets $\{i,j,l\}$ and $\{i',j',l'\}$ are different. Thus,
$$\mathbb{E}\left(\frac{1}{{n\choose 3}}\sum_{i<j<l}H_{ijl}\right)^2=\frac{1}{{n\choose 3}^2}\sum_{i<j<l}\mathbb{E}H_{ijl}^2=O\left(\frac{a^4}{n^3}\right).$$
Moreover,
$$\mathbb{E}\left(\frac{3}{{n\choose 2}}\sum_{i<j}\left[\frac{1}{n-2}\sum_{l\notin\{i,j\}}\theta_{il}\theta_{jl}\right](A_{ij}-\theta_{ij})\right)^2=O\left(\frac{a^5}{n^2}\right).$$
Write
$$G_{ijl}=(A_{ij}-\theta_{ij})(A_{jl}-\theta_{jl})(A_{il}-\theta_{il}).$$
It is not hard to see that $G_{ijl}$ and $G_{i'j'l'}$ are uncorrelated if the sets $\{i,j,l\}$ and $\{i',j',l'\}$ are different. Therefore,
$$\mathbb{E}\left(\frac{1}{{n\choose 3}}\sum_{i<j<l}G_{ijl}\right)^2=\frac{1}{{n\choose 3}^2}\sum_{i<j<l}\mathbb{E}\theta_{ij}\theta_{jl}\theta_{il}(1-\theta_{ij})(1-\theta_{jl})(1-\theta_{il})=(1+o(1))\frac{T}{{n\choose 3}}.$$
Since $T/{n\choose 3}\asymp\frac{a^3}{n^3}$, $\frac{1}{{n\choose 3}}\sum_{i<j<l}G_{ijl}$ is the dominating term of $\hat{T}-T$ under the condition $n^{-1}\ll a\asymp b\ll n^{-2/3}$. Thus, the asymptotic distribution of $\hat{T}-T$ is the same as that of $\frac{1}{{n\choose 3}}\sum_{i<j<l}G_{ijl}$.

Define the set
$$\tilde{W}_n=\left\{\left|\frac{1}{n}\sum_{i=1}^nW_i^2-\mathbb{E}W^2\right|\leq n^{-1/3}\right\}.$$
Under the condition $\mathbb{E}W^4=O(1)$, we obtain $\mathbb{P}(\tilde{W}_n^c)=O(n^{-1/3})$ by Chebyshev's ienquality. Moreover, for any $(W_1,...,W_n)\in \tilde{W}_n$, we have $\frac{1}{n}\sum_{i=1}^nW_i^2=(1+o(1))\mathbb{E}W^2=1+o(1)$. In order to show $\frac{1}{{n\choose 3}}\sum_{i<j<l}G_{ijl}$ converges to a Gaussian distribution, we construct a martingale conditioning on $Z$ and $W$, and then apply the martingale central limit theorem in \cite{hall2014martingale}. Define
$$S_{n,m}=\frac{\sum_{i=1}^{m-2}\sum_{j=i+1}^{m-1}\sum_{l=j+1}^m(A_{ij}A_{jl}A_{il}-\theta_{ij}\theta_{jl}\theta_{il})}{\sqrt{\sum_{1\leq i<j<l\leq n}\theta_{ij}\theta_{jl}\theta_{il}(1-\theta_{ij})(1-\theta_{il})(1-\theta_{jl})}}.$$
Then, $\{S_{n,m}\}$ is a conditional martingale with respect to the filtration $\sigma(\{A_{ij}\}_{1\leq i<j\leq m}|W,Z)$. When $\theta_{ij}$ is a constant for all $(i,j)$, this martingale was analyzed in \cite{gao2017testing}. In addition, the same analysis in \cite{gao2017testing} can also be applied to the setting here almost without any change. The only difference we need to check here is a lower bound for $\sum_{1\leq i<j<l\leq n}\theta_{ij}\theta_{jl}\theta_{il}(1-\theta_{ij})(1-\theta_{il})(1-\theta_{jl})$. Since $a\asymp b=o(1)$, it is of the same order as $\sum_{1\leq i<j<l\leq n}\theta_{ij}\theta_{jl}\theta_{il}$, which is lower bounded by ${n\choose 3}(a\wedge b)^3\asymp n^3a^3$. Hence, we can apply the same argument in \cite{gao2017testing} and obtain
$$\frac{\sum_{i<j<l}G_{ijl}}{\sqrt{\sum_{1\leq i<j<l\leq n}\theta_{ij}\theta_{jl}\theta_{il}(1-\theta_{ij})(1-\theta_{il})(1-\theta_{jl})}}\Big| (W,Z)\leadsto N(0,1).$$
This asymptotic result holds for all $w\in \tilde{W}_n$ and $Z$. Since $\mathbb{P}(\tilde{W}_n^c)=O(n^{-1/3})$, it also holds without conditioning on $W$ and $Z$. Moreover, since
$$\frac{1}{{n\choose 3}}\sum_{1\leq i<j<l\leq n}\theta_{ij}\theta_{jl}\theta_{il}(1-\theta_{ij})(1-\theta_{il})(1-\theta_{jl})=(1+o_P(1))\frac{1}{{n\choose 3}}\sum_{1\leq i<j<l\leq n}\theta_{ij}\theta_{jl}\theta_{il},$$
which is exactly $\mathbb{E}(\hat{T}|W,Z)$. We have just obtained that $\mathbb{E}(\hat{T}-\mathbb{E}(\hat{T}|W,Z))^2=O\left(\frac{a^3}{n^3}\right)=o(T^2)$. Thus,
$$\frac{\frac{1}{\sqrt{{n\choose 3}}}\sum_{i<j<l}G_{ijl}}{\sqrt{{T}}}\leadsto N(0,1),$$
by Slutsky Theorem. Since $\frac{1}{{n\choose 3}}\sum_{i<j<l}G_{ijl}$ is the dominating term of $\hat{T}-T$, the desired result is obtained.
\end{proof}

\begin{proof}[Proof of Lemma \ref{lem:EV-NG}]
We first introduce some notations. Define
\begin{eqnarray*}
\tilde{E} &=& \frac{1}{{n\choose 2}\left(\frac{rp}{k}\right)^2}\sum_{1\leq i<j\leq n}A_{ij}A_{0i}A_{0j}, \\
\tilde{V} &=& \frac{1}{3{n\choose 3}\left(\frac{rp}{k}\right)^3}\sum_{1\leq i<j<l\leq n}(A_{ij}A_{il}+A_{ij}A_{jl}+A_{il}A_{jl})A_{0i}A_{0j}A_{0l}.
\end{eqnarray*}
By the definition of each $A_{0i}$, it can be written as $A_{0i}=\mathbb{I}\{Z_i\in\mathcal{R}\}S_i$, where $S_i\sim\text{Bernoulli}(p)$. We use the notations
$$\tilde{A}_{ij}=A_{ij}\mathbb{I}\{Z_i\in\mathcal{R}\}\mathbb{I}\{Z_j\in\mathcal{R}\}\quad\text{and}\quad\tilde{\theta}_{ij}=\theta_{ij}\mathbb{I}\{Z_i\in\mathcal{R}\}\mathbb{I}\{Z_j\in\mathcal{R}\}.$$
It is easy to see that $\tilde{A}_{ij}|(W,Z)\sim\text{Bernoulli}(\tilde{\theta}_{ij})$. We also define $\tilde{\eta}_{ij}=\eta_{ij}\mathbb{I}\{Z_i\in\mathcal{R}\}\mathbb{I}\{Z_j\in\mathcal{R}\}$, so that $\tilde{\theta}_{ij}=W_iW_j\tilde{\eta}_{ij}$.

We first bound $(\hat{E}-\tilde{E})^2$ and $(\hat{V}-\tilde{V})^2$. Recall that $m=\sum_{i=1}^nA_{0i}$. Since $A_{0i}\sim\text{Bernoulli}\left(\frac{rp}{k}\right)$, by Proposition \ref{prop:m-performance}, we have
$$\left({n\choose 2}\left(\frac{rp}{k}\right)^2-\sum_{i<j}A_{0i}A_{0j}\right)^2=\left(\sum_{i<j}\left(A_{0i}A_{0j}-\left(\frac{rp}{k}\right)^2\right)\right)^2=O_P\left(\left(\frac{nrp}{k}\right)^3\right),$$
and
$$\left({n\choose 3}\left(\frac{rp}{k}\right)^3-\sum_{i<j<l}A_{0i}A_{0j}A_{0l}\right)^2=\left(\sum_{i<j<l}\left(A_{0i}A_{0j}A_{0l}-\left(\frac{rp}{k}\right)^3\right)\right)^2=O_P\left(\left(\frac{nrp}{k}\right)^5\right),$$
under the condition $\frac{nrp}{k}\rightarrow\infty$. This leads to
$$(\hat{E}-\tilde{E})^2=\left(\sum_{1\leq i<j\leq n}A_{ij}A_{0i}A_{0j}\right)^2\left|\frac{1}{\sum_{i<j}A_{0i}A_{0j}}-\frac{1}{{n\choose 2}\left(\frac{rp}{k}\right)^2}\right|^2=O_P\left(\frac{a^2k}{nrp}\right).$$
A similar argument leads to
$$(\hat{V}-\tilde{V})^2=O_P\left(\frac{a^4k}{nrp}\right).$$

Next, we give bounds for $(\tilde{E}-E)^2$ and $(\tilde{V}-V)^2$. We use $A$ to represent $\{A_{ij}\}_{1\leq i<j\leq n}$.
We have the decomposition
\begin{eqnarray*}
\mathbb{E}(\tilde{E}-E)^2 &=& \mathbb{E}(\tilde{E}-\mathbb{E}(\tilde{E}|A,W,Z))^2 + \mathbb{E}(\mathbb{E}(\tilde{E}|A,W,Z)-\mathbb{E}(\tilde{E}|W,Z))^2 \\
&& +\mathbb{E}(\mathbb{E}(\tilde{E}|W,Z)-\mathbb{E}(\tilde{E}|Z))^2+\mathbb{E}(\mathbb{E}(\tilde{E}|Z)-E)^2,
\end{eqnarray*}
Note that
$$\mathbb{E}(\tilde{E}|A,W,Z)=\left(\frac{k}{r}\right)^2\frac{1}{{n\choose 2}}\sum_{1\leq i<j\leq n}\tilde{A}_{ij}.$$
This gives,
$$\tilde{E}-\mathbb{E}(\tilde{E}|A,W,Z)=\left(\frac{k}{rp}\right)^2\frac{1}{{n\choose 2}}\sum_{1\leq i<j\leq n}\tilde{A}_{ij}\left(S_iS_j-p^2\right),$$
where $\{S_i\}$ are independent of $\{\tilde{A}_{ij}\}$. Thus, by Proposition \ref{prop:Sij}, under the condition $\frac{nrpa}{k}\rightarrow\infty$,
$$\mathbb{E}\left(\tilde{E}-\mathbb{E}(\tilde{E}|A,W,Z)\right)^2=O\left(\frac{a^2k}{prn}\right).$$
For the second term,
$$\mathbb{E}(\tilde{E}|A,W,Z)-\mathbb{E}(\tilde{E}|W,Z)=\left(\frac{k}{r}\right)^2\frac{1}{{n\choose 2}}\sum_{1\leq i<j\leq n}(\tilde{A}_{ij}-\tilde{\theta}_{ij}).$$
Therefore,
$$\mathbb{E}\left(\mathbb{E}(\tilde{E}|A,W,Z)-\mathbb{E}(\tilde{E}|W,Z)\right)^2\leq \left(\frac{k}{r}\right)^4\frac{1}{{n\choose 2}^2}\sum_{1\leq i<j\leq n}\mathbb{E}\tilde{\theta}_{ij}=O\left(a\left(\frac{k}{rn}\right)^2\right).$$
For the third term,
$$\mathbb{E}(\tilde{E}|W,Z)-\mathbb{E}(\tilde{E}|Z)=\left(\frac{k}{r}\right)^2\frac{1}{{n\choose 2}}\sum_{1\leq i<j\leq n}\tilde{\eta}_{ij}\left(W_iW_j-(\mathbb{E}W)^2\right).$$
By Proposition \ref{prop:volvo},
$$\mathbb{E}\left(\mathbb{E}(\tilde{E}|W,Z)-\mathbb{E}(\tilde{E}|Z)\right)^2=O\left(a^2\left(\frac{k}{rn}\right)\right).$$
For the last term,
$$\mathbb{E}(\tilde{E}|Z)-E=\left(\frac{k}{r}\right)^2\frac{1}{{n\choose 2}}\sum_{1\leq i<j\leq n}(\tilde{\eta}_{ij}-\mathbb{E}\tilde{\eta}_{ij})(\mathbb{E}W)^2.$$
By Proposition \ref{prop:civic-si},
$$\mathbb{E}\left(\mathbb{E}(\tilde{E}|Z)-E\right)^2=O\left(a^2\left(\frac{k}{rn}\right)\right).$$
Combining above bounds, we get $\mathbb{E}(\tilde{E}-E)^2=O\left(a^2\left(\frac{k}{prn}\right)\right)$. This leads to $(\hat{E}-E)^2=O_P\left(a^2\left(\frac{k}{prn}\right)\right)$.

We analyze $(\tilde{V}-V)^2$ in a similar way. First, we have the decomposition
\begin{eqnarray*}
\mathbb{E}(\tilde{V}-V)^2 &=& \mathbb{E}(\tilde{V}-\mathbb{E}(\tilde{V}|A,W,Z))^2 + \mathbb{E}(\mathbb{E}(\tilde{V}|A,W,Z)-\mathbb{E}(\tilde{V}|W,Z))^2 \\
&& +\mathbb{E}(\mathbb{E}(\tilde{V}|W,Z)-\mathbb{E}(\tilde{V}|Z))^2+\mathbb{E}(\mathbb{E}(\tilde{V}|Z)-V)^2.
\end{eqnarray*}
Note that
$$\mathbb{E}(\tilde{V}|A,W,Z)=\left(\frac{k}{r}\right)^3\frac{1}{3{n\choose 3}}\sum_{1\leq i<j<l\leq n}(\tilde{A}_{ij}\tilde{A}_{il}+\tilde{A}_{ij}\tilde{A}_{jl}+\tilde{A}_{il}\tilde{A}_{jl}).$$
We can also obtain bounds for $\mathbb{E}(\tilde{V}-\mathbb{E}(\tilde{V}|A,W,Z))^2$, $\mathbb{E}(\mathbb{E}(\tilde{V}|A,W,Z)-\mathbb{E}(\tilde{V}|W,Z))^2$,
$\mathbb{E}(\mathbb{E}(\tilde{V}|W,Z)-\mathbb{E}(\tilde{V}|Z))^2$, and $\mathbb{E}(\mathbb{E}(\tilde{V}|Z)-V)^2$. Under the condition $\frac{nrpa}{k}\rightarrow\infty$, we have
$$\mathbb{E}(\tilde{V}-\mathbb{E}(\tilde{V}|A,W,Z))^2=O\left(\frac{a^4k}{prn}\right),$$
$$\mathbb{E}(\mathbb{E}(\tilde{V}|A,W,Z)-\mathbb{E}(\tilde{V}|W,Z))^2=O\left(a^3\left(\frac{k}{rn}\right)^2\right),$$
$$\mathbb{E}(\mathbb{E}(\tilde{V}|W,Z)-\mathbb{E}(\tilde{V}|Z))^2=O\left(a^4\left(\frac{k}{rn}\right)\right),$$
and
$$\mathbb{E}(\mathbb{E}(\tilde{V}|Z)-V)^2=O\left(a^4\left(\frac{k}{rn}\right)\right).$$
See Proposition \ref{prop:polestar}, Proposition \ref{prop:v-tec} and Proposition \ref{prop:civic-type-R} for derivations.
Thus, $(\tilde{V}-V)^2=O_P\left(\frac{a^4k}{prn}\right)$, which leads to $(\hat{V}-V)^2=O_P\left(\frac{a^4k}{prn}\right)$.
\end{proof}

\begin{proof}[Proof of Lemma \ref{lem:T-NG}]
Similar to the proof of Lemma \ref{lem:EV-NG}, we define
$$\tilde{T}=\frac{1}{{n\choose 3}\left(\frac{rp}{k}\right)^3}\sum_{1\leq i<j<l\leq n}A_{ij}A_{il}A_{jl}A_{0i}A_{0j}A_{0l}.$$
By Proposition \ref{prop:m-performance}, we have
$$\left({n\choose 3}\left(\frac{rp}{k}\right)^3-\sum_{i<j<l}A_{0i}A_{0j}A_{0l}\right)^2=\left(\sum_{i<j<l}\left(A_{0i}A_{0j}A_{0l}-\left(\frac{rp}{k}\right)^3\right)\right)^2=O_P\left(\left(\frac{nrp}{k}\right)^5\right),$$
under the condition $\frac{nrp}{k}\rightarrow\infty$. Thus,
\begin{eqnarray*}
(\hat{T}-\tilde{T})^2 &=& \left(\sum_{1\leq i<j<l\leq n}A_{ij}A_{il}A_{jl}A_{0i}A_{0j}A_{0l}\right)^2\left|\frac{1}{{n\choose 3}\left(\frac{rp}{k}\right)^3}-\frac{1}{\sum_{i<j<l}A_{0i}A_{0j}A_{0l}}\right|^2 \\
&=& O\left(\frac{a^6k}{nrp}\right).
\end{eqnarray*}
Recall the definition of $\{S_i\}$, $\tilde{A}_{ij}$, $\tilde{\theta}_{ij}$ and $\tilde{\eta}_{ij}$ in the proof of Lemma \ref{lem:EV-NG}. Note that
$$\mathbb{E}(\tilde{T}|A,W,Z)=\left(\frac{k}{r}\right)^3\frac{1}{{n\choose 2}}\sum_{1\leq i<j<l\leq n}\tilde{A}_{ij}\tilde{A}_{il}\tilde{A}_{jl}.$$
Then,
$$\tilde{T}-\mathbb{E}(\tilde{T}|A,W,Z)=\frac{1}{{n\choose 3}\left(\frac{rp}{k}\right)^3}\sum_{1\leq i<j<l\leq n}\tilde{A}_{ij}\tilde{A}_{il}\tilde{A}_{jl}(S_iS_jS_l-p^3).$$
By Proposition \ref{prop:cls63}, we have
$$\mathbb{E}(\tilde{T}-\mathbb{E}(\tilde{T}|A,W,Z))^2=O\left(\frac{a^6k}{rnp}\right),$$
under the condition $\frac{nrpa}{k}\rightarrow\infty$.

Now we study $\mathbb{E}(\tilde{T}|A,W,Z)-T$. It has the following expansion,
\begin{eqnarray*}
\mathbb{E}(\tilde{T}|A,W,Z)-T &=& \left(\frac{k}{r}\right)^3\frac{1}{{n\choose 3}}\sum_{i<j<l}(\tilde{A}_{ij}-\tilde{\theta}_{ij})(\tilde{A}_{jl}-\tilde{\theta}_{jl})(\tilde{A}_{il}-\tilde{\theta}_{il}) \\
&& + \left(\frac{k}{r}\right)^3\frac{1}{{n\choose 3}}\sum_{i<j<l} \left[\tilde{\theta}_{ij}(\tilde{A}_{jl}-\tilde{\theta}_{jl})(\tilde{A}_{il}-\tilde{\theta}_{il})+\tilde{\theta}_{jl}(\tilde{A}_{ij}-\tilde{\theta}_{ij})(\tilde{A}_{il}-\tilde{\theta}_{il})\right. \\
&& \left.+ \tilde{\theta}_{il}(\tilde{A}_{ij}-\tilde{\theta}_{ij})(\tilde{A}_{jl}-\tilde{\theta}_{jl})\right] \\
&& + \left(\frac{k}{r}\right)^3\frac{3}{{n\choose 2}}\sum_{i<j}\left[\frac{1}{n-2}\sum_{l\notin\{i,j\}}\tilde{\theta}_{il}\tilde{\theta}_{jl}\right](\tilde{A}_{ij}-\tilde{\theta}_{ij}) \\
&& + \mathbb{E}(\tilde{T}|W,Z)-T.
\end{eqnarray*}
Among the four terms in the expansion, we will argue that the first term is the dominating term. The variance of the first term is
\begin{eqnarray*}
&& \mathbb{E}\left(\left(\frac{k}{r}\right)^3\frac{1}{{n\choose 3}}\sum_{i<j<l}(\tilde{A}_{ij}-\tilde{\theta}_{ij})(\tilde{A}_{jl}-\tilde{\theta}_{jl})(\tilde{A}_{il}-\tilde{\theta}_{il})\right)^2 \\
&=& \left(\frac{k}{r}\right)^6\frac{1}{{n\choose 3}^2}\sum_{i<j<l}\mathbb{E}\tilde{\theta}_{ij}\tilde{\theta}_{jl}\tilde{\theta}_{il}(1-\tilde{\theta}_{ij})(1-\tilde{\theta}_{jl})(1-\tilde{\theta}_{il}) \\
&=& (1+o(1))\left(\frac{k}{r}\right)^3\frac{1}{{n\choose 3}^2}\sum_{i<j<l}\mathbb{E}\theta_{ij}\theta_{jl}\theta_{il} \\
&\asymp& \frac{a^3k^3}{n^3r^3}.
\end{eqnarray*}
Using similar calculation, the variances of the second and the third terms can be bounded by $O\left(\frac{a^4k^3}{n^3r^3}\right)$ and $O\left(\frac{a^5k^2}{n^2r^2}\right)$, respectively. For the fourth term, we have
\begin{eqnarray*}
&& \mathbb{E}\left(\mathbb{E}(\tilde{T}|W,Z)-T\right)^2 \\
&=& \mathbb{E}\left(\mathbb{E}(\tilde{T}|W,Z)-\mathbb{E}(\tilde{T}|Z)\right)^2 + \mathbb{E}(\mathbb{E}(\tilde{T}|Z)-T)^2 \\
&=& \mathbb{E}\left(\left(\frac{k}{r}\right)^3\frac{1}{{n\choose 3}}\sum_{i<j<l}\tilde{\eta}_{ij}\tilde{\eta}_{jl}\tilde{\eta}_{il}(W_i^2W_j^2W_l^2-(\mathbb{E}W^2)^3)\right)^2 \\
&& + (\mathbb{E}W^2)^3\Var\left(\left(\frac{k}{r}\right)^3\frac{1}{{n\choose 3}}\sum_{i<j<l}\tilde{\eta}_{ij}\tilde{\eta}_{jl}\tilde{\eta}_{il}\right) \\
&=& O\left(\frac{a^6k}{nr}\right),
\end{eqnarray*}
by Proposition \ref{prop:glorymu} and Proposition \ref{prop:E63}.

Combining all the bounds above, under the condition $a^6=o\left(\left(\frac{k}{nr}\right)^3\wedge p^2\left(\frac{k}{nr}\right)^4\right)$, the dominating term of $\hat{T}-T$ is 
$$\left(\frac{k}{r}\right)^3\frac{1}{{n\choose 3}}\sum_{i<j<l}(\tilde{A}_{ij}-\tilde{\theta}_{ij})(\tilde{A}_{jl}-\tilde{\theta}_{jl})(\tilde{A}_{il}-\tilde{\theta}_{il}),$$
and we will find its asymptotic distribution. By the same martingale argument that we have used in the proof of Lemma \ref{lem:EV-NG}, we obtain
$$\sqrt{\left(\frac{k}{r}\right)^3\frac{1}{{n\choose 3}}}\sum_{i<j<l}(\tilde{A}_{ij}-\tilde{\theta}_{ij})(\tilde{A}_{jl}-\tilde{\theta}_{jl})(\tilde{A}_{il}-\tilde{\theta}_{il})/\sqrt{T}\leadsto N(0,1).$$
This completes the proof.
\end{proof}

\begin{proof}[Proof of Lemma \ref{lem:EV-cov}]
Note that
$$\mathbb{E}(\hat{E}-E)^2=\mathbb{E}(\hat{E}-\mathbb{E}(\hat{E}|\theta))^2+\mathbb{E}(\mathbb{E}(\hat{E}|\theta)-E)^2.$$
The second term $\mathbb{E}(\mathbb{E}(\hat{E}|\theta)-E)^2$ has the same expression as the one in the network setting, which has already been bounded in the proof of Lemma \ref{lem:EV-order}. Therefore, $\mathbb{E}(\mathbb{E}(\hat{E}|\theta)-E)^2=O(\frac{a^2}{p})$ under the condition $p^{-1}\ll |a|$. The first term is
$$\mathbb{E}(\hat{E}-\mathbb{E}(\hat{E}|\theta))^2=\frac{1}{n^2}\sum_{i=1}^n\mathbb{E}(E_i-\mathbb{E}(E_i|\theta))^2.$$
For each $i$,
$$\mathbb{E}(E_i-\mathbb{E}(E_i|\theta))^2\leq \mathbb{E}E_i^2=\frac{1}{{p\choose 2}^2}\sum_{j<l}\sum_{j'<l'}\mathbb{E}X_jX_lX_{j'}X_{l'}.$$
There are three situations for calculating $\mathbb{E}X_jX_lX_{j'}X_{l'}$. When the sets $\{j,l\}$ and $\{j',l'\}$ do not have intersection,
$$\mathbb{E}X_jX_lX_{j'}X_{l'}=3(\mathbb{E}W)^4\left(\frac{1}{k}a+\frac{k-1}{k}b\right)^2=O(a^2).$$
When the sets $\{j,l\}$ and $\{j',l'\}$ are intersected by one element, we have
$$\mathbb{E}X_jX_lX_{j'}X_{l'}=2(\mathbb{E}W)^2\left(\frac{1}{k}a+\frac{k-1}{k}b\right)^2+(\mathbb{E}W)^2\left(\frac{1}{k}a+\frac{k-1}{k}b\right)=O(|a|).$$
When the sets $\{j,l\}$ and $\{j',l'\}$ are identical,
$$\mathbb{E}X_jX_lX_{j'}X_{l'}=1+2\left(\frac{1}{k}a^2+\frac{k-1}{k}b^2\right)=O(1).$$
All of the above moments calculation can be done through Wick's formula. Under the condition $p^{-1}\ll |a|$, we have $\mathbb{E}(E_i-\mathbb{E}(E_i|\theta))^2=O(a^{2})$, and thus $\mathbb{E}(\hat{E}-\mathbb{E}(\hat{E}|\theta))^2=O\left(\frac{a^2}{n}\right)$.

We study $(\hat{V}-V)^2$ in the same way. First, we have
$$\mathbb{E}(\hat{V}-V)^2=\mathbb{E}(\hat{V}-\mathbb{E}(\hat{V}|\theta))^2+\mathbb{E}(\mathbb{E}(\hat{V}|\theta)-V)^2.$$
By the proof of Lemma \ref{lem:EV-order}, we have $\mathbb{E}(\mathbb{E}(\hat{V}|\theta)-V)^2=O(\frac{a^4}{p})$ under the condition $p^{-1}\ll |a|$. For the first term, we have
$$\mathbb{E}(\hat{V}-\mathbb{E}(\hat{V}|\theta))^2=\frac{1}{n^2}\sum_{i=1}^n\mathbb{E}(V_i-\mathbb{E}(V_i|\theta))^2.$$
For each $i$, we have $\mathbb{E}(V_i-\mathbb{E}(V_i|\theta))^2\leq \mathbb{E}V_i^2$. By the definition of $V_i$, it is sufficient to bound $\mathbb{E}\left(\frac{1}{{p\choose 3}}\sum_{j<h<l}X_{ij}X_{ih}X_{il}(X_{ij}+X_{ih}+X_{il})\right)^2$. By its definition, it can be written as
$$\frac{1}{{p\choose 3}^2}\sum_{j<h<l}\sum_{j'<h'<l'}\mathbb{E}X_{ij}X_{ih}X_{il}(X_{ij}+X_{ih}+X_{il})X_{ij'}X_{ih'}X_{il'}(X_{ij'}+X_{ih'}+X_{il'}).$$
To make the presentation concise, we will omit some details in the application of Wick's formula when calculating various moments. When the sets $\{j,h,l\}$ and $\{j',h',l'\}$ do not intersect, we have
$$\mathbb{E}X_{ij}X_{ih}X_{il}(X_{ij}+X_{ih}+X_{il})X_{ij'}X_{ih'}X_{il'}(X_{ij'}+X_{ih'}+X_{il'})=O(a^2).$$
Therefore, using the similar argument in bounding $\mathbb{E}E_i^2$, we have $\mathbb{E}V_i^2=O(a^2)$ under the condition $p^{-1}\ll |a|$. This completes the proof. 
\end{proof}

\begin{proof}[Proof of Lemma \ref{lem:T-cov}]
First, observe that
$$\hat{T}-T=\hat{T}-\mathbb{E}(\hat{T}|\theta)+\mathbb{E}(\hat{T}|\theta)-T.$$
According to the proof of Lemma \ref{lem:T-order}, $(\mathbb{E}(\hat{T}|\theta)-T)^2=O_P\left(\frac{a^6}{p}\right)$ under the condition $p^{-1}\ll |a|$. Next, we calculate the expected variance of $\hat{T}-\mathbb{E}(\hat{T}|\theta)$. Note that
$$\mathbb{E}\Var(T_i|\theta)=\mathbb{E}T_i^2-(\mathbb{E}T_i)^2-\mathbb{E}(\mathbb{E}(\hat{T}|\theta)-T)^2.$$
The first term is
\begin{eqnarray*}
\mathbb{E}T_i^2 &=& \frac{1}{64{p\choose 3}^2}\sum_{j<h<l}\sum_{j'<h'<l'}\mathbb{E}X_{ij}^2X_{ih}^2X_{il}^2X_{ij'}^2X_{ih'}^2X_{il'}^2 \\
&& + \frac{9}{64{p\choose 2}^2}\sum_{j<h}\sum_{j'<h'}\mathbb{E}X_{ij}^2X_{ih}^2X_{ij'}^2X_{ih'}^2 \\
&& - \frac{6}{64{p\choose 2}{p\choose 3}}\sum_{j<h<l}\sum_{j'<h'}\mathbb{E}X_{ij}^2X_{ih}^2X_{il}^2X_{ij'}^2X_{ih'}^2 \\
&& + \frac{1}{16{p\choose 3}}\sum_{j<h<l}\mathbb{E}X_{ij}^2X_{ih}^2X_{il}^2 - \frac{3}{16{p\choose 2}}\sum_{j<h}\mathbb{E}X_{ij}^2X_{ih}^2 + \frac{1}{16}
\end{eqnarray*}
The following calculations by Wick's formula are helpful
\begin{eqnarray*}
\mathbb{E}X_1^2X_2^2 &=& 1 + 2\mathbb{E}\theta_{12}^2, \\
\mathbb{E}X_1^2X_2^2X_3^2 &=& 1 + 6\mathbb{E}\theta_{12}^2 + o(a^2), \\
\mathbb{E}X_1^2X_2^2X_3^2X_4^2 &=& 1 + 12\mathbb{E}\theta_{12}^2 + o(a^2), \\
\mathbb{E}X_1^2X_2^2X_3^2X_4^2X_5^2 &=& 1 + 20\mathbb{E}\theta_{12}^2 + o(a^2) \\
\mathbb{E}X_1^2X_2^2X_3^2X_4^2X_5^2X_6^2 &=& 1 + 30\mathbb{E}\theta_{12}^2 + o(a^2).
\end{eqnarray*}
Therefore, under the condition $p^{-1}\ll |a|$, we have $\mathbb{E}T_i^2=(1+o(1))\frac{9}{32}\mathbb{E}\theta_{12}^2$. Since $(\mathbb{E}T_i)^2=O(a^6)$ and $\mathbb{E}(\mathbb{E}(\hat{T}|\theta)-T)^2=O(a^6/p)$, we have $\mathbb{E}\Var(T_i|\theta)=(1+o(1))\frac{9}{32}\mathbb{E}\theta_{12}^2\asymp a^2$. This leads to
$$\mathbb{E}(\hat{T}-\mathbb{E}(\hat{T}|\theta))^2=\frac{1}{n^2}\sum_{i=1}^n\mathbb{E}\Var(T_i|\theta)\asymp \frac{a^2}{n}.$$
Under the condition $a^4=o(p/n)$, $\hat{T}-\mathbb{E}(\hat{T}|\theta)$ is the dominating term of $\hat{T}-T$. Thus, the asymptotic distribution of $\hat{T}-T$ is determined by that of $\hat{T}-\mathbb{E}(\hat{T}|\theta)$.

To prove a central limit theorem for $\hat{T}-\mathbb{E}(\hat{T}|\theta)$, we need to establish the Lyapunov's condition,
\begin{equation}
\frac{\sum_{i=1}^n\mathbb{E}\left[(T_i-\mathbb{E}(\hat{T}|\theta))^4|\theta\right]}{\left(\sum_{i=1}^n\Var(T_i|\theta)\right)^2}=o_P(1).\label{eq:lyapunov}
\end{equation}
Note that
\begin{eqnarray*}
&& \left|\Var(T_i|\theta)-\mathbb{E}\Var(T_i|\theta)\right|^2 \\
&\leq& 3(\mathbb{E}(T_i^2|\theta)-\mathbb{E}T_i^2)^2 + 3\left((\mathbb{E}(T_i|\theta))^2-(\mathbb{E}T_i)^2\right)^2 + 3\left(\mathbb{E}(\mathbb{E}(\hat{T}|\theta)-T)^2\right)^2.
\end{eqnarray*}
The first term in the above bound can be bounded by $O_P(a^4/p)$ by using similar analysis as in the proof of previous lemmas. The second and the third term can be bounded by $O_P(a^{12}/p)$ and $O_P(a^{12}/p^2)$ by $\mathbb{E}(\mathbb{E}(\hat{T}|\theta)-T)^2=O(a^6/p)$. Since $\mathbb{E}\Var(T_i|\theta)\asymp a^2$, we get $\Var(T_i|\theta)=(1+o_P(1))\mathbb{E}\Var(T_i|\theta)$. As a result $\left(\sum_{i=1}^n\Var(T_i|\theta)\right)^2\asymp_P n^2a^4$, which serves as a lower bound for the denominator of (\ref{eq:lyapunov}).

To bound the numerator of (\ref{eq:lyapunov}), it is sufficient to give a bound for $\mathbb{E}T_i^4$. Since
$$\mathbb{E}\left(\frac{1}{{p\choose 3}}\sum_{j<h<l}X_{ij}^2X_{ih}^2X_{il}^2-1\right)^4=O(a^2),$$
and
$$\mathbb{E}\left(\frac{1}{{p\choose 2}}\sum_{j<h}X_{ij}^2X_{ih}^2-1\right)^4=O(a^2),$$
we have $\mathbb{E}T_i^4=O(a^2)$. This implies that Lyapunov's condition holds if $n^{-1}\ll a^2$. Therefore, $\frac{\hat{T}-\mathbb{E}(\hat{T}|\theta)}{\sqrt{\Var(\hat{T}|\theta)}}\leadsto N(0,1)$, which leads to the desired result.
\end{proof}

\subsection{Proofs of Technical Results}\label{sec:pf-tech}

\begin{proposition}\label{prop:wij}
Under the same setting of Lemma \ref{lem:EV-order},
$$\mathbb{E}\left(\frac{1}{{n\choose 2}}\sum_{i<j}\eta_{ij}\left(W_iW_j-(\mathbb{E}W)^2\right)\right)^2=O\left(\frac{a^2}{n}\right).$$
\end{proposition}
\begin{proof}
Note that
\begin{equation}
W_iW_j-(\mathbb{E}W)^2=(\mathbb{E}W)(W_i-\mathbb{E}W+W_j-\mathbb{E}W)+(W_i-\mathbb{E}W)(W_j-\mathbb{E}W).\label{eq:2-decomp}
\end{equation}
We observe that the two terms on the right hand side of the above display are uncorrelated. Thus,
\begin{eqnarray*}
&& \mathbb{E}\left(\frac{1}{{n\choose 2}}\sum_{i<j}\eta_{ij}\left(W_iW_j-(\mathbb{E}W)^2\right)\right)^2 \\
&=& \mathbb{E}\left(\frac{1}{{n\choose 2}}\sum_{i<j}\eta_{ij}(\mathbb{E}W)(W_i-\mathbb{E}W+W_j-\mathbb{E}W)\right)^2 \\
&& + \mathbb{E}\left(\frac{1}{{n\choose 2}}\sum_{i<j}\eta_{ij}(W_i-\mathbb{E}W)(W_j-\mathbb{E}W)\right)^2.
\end{eqnarray*}
Since $\eta_{ij}=O(a)$ and $\{\eta_{ij}\}$ are independent of $\{W_i\}$, we have
\begin{eqnarray*}
&& \mathbb{E}\left(\frac{1}{{n\choose 2}}\sum_{i<j}\eta_{ij}(\mathbb{E}W)(W_i-\mathbb{E}W+W_j-\mathbb{E}W)\right)^2 \\
&\leq& \mathbb{E}\left(\frac{1}{{n\choose 2}}\sum_{i\neq j}\eta_{ij}(\mathbb{E}W)(W_i-\mathbb{E}W)\right)^2 \\
&=& \frac{1}{{n\choose 2}^2}\sum_{i=1}^n\mathbb{E}\left(\sum_{\{j:j\neq i\}}\eta_{ij}(\mathbb{E}W)\right)^2\Var(W_i) \\
&=&O\left(\frac{a^2}{n}\right).
\end{eqnarray*}
For the second term, observe that $(W_i-\mathbb{E}W)(W_j-\mathbb{E}W)$ and $(W_{i'}-\mathbb{E}W)(W_{j'}-\mathbb{E}W)$ are uncorrelated if the sets $\{i,j\}$ and $\{i',j'\}$ do not intersect. Thus,
\begin{eqnarray*}
&& \mathbb{E}\left(\frac{1}{{n\choose 2}}\sum_{i<j}\eta_{ij}(W_i-\mathbb{E}W)(W_j-\mathbb{E}W)\right)^2 \\
&=& \frac{1}{{n\choose 2}^2}\sum_{i<j}\mathbb{E}\eta_{ij}^2\Var(W_i)\Var(W_j) \\
&=&O\left(\frac{a^2}{n^2}\right),
\end{eqnarray*}
which leads to the desired result.
\end{proof}

\begin{proposition}\label{prop:z-foundation}
Under the same setting of Lemma \ref{lem:EV-order},
$$\mathbb{E}\left(\frac{1}{{n\choose 2}}\sum_{i<j}(\mathbb{I}\{Z_i=Z_j\}-\mathbb{P}(Z_i=Z_j))\right)^2=O\left(\frac{1}{n}\right).$$
\end{proposition}
\begin{proof}
Define $z_{il}=\mathbb{I}\{Z_i=l\}$. It is easy to see that $z_{il}\sim\text{Bernoulli}(k^{-1})$. Moreover, $z_{il}$ and $z_{jl}$ are independent if $i\neq j$, which implies $\mathbb{E}(z_{il}z_{jl})=(\mathbb{E}z_{il})(\mathbb{E}z_{jl})$. Then, $\mathbb{I}\{Z_i=Z_j\}=\sum_{l=1}^kz_{il}z_{jl}$. Observe the inequality
\begin{eqnarray*}
&& \mathbb{E}\left(\frac{1}{{n\choose 2}}\sum_{i<j}(\mathbb{I}\{Z_i=Z_j\}-\mathbb{P}(Z_i=Z_j))\right)^2 \\
&\leq& \frac{1}{{n\choose 2}^2}k\sum_{l=1}^k\left(\sum_{i<j}z_{il}z_{jl}-(\mathbb{E}z_{il})(\mathbb{E}z_{jl})\right)^2.
\end{eqnarray*}
Similar to (\ref{eq:2-decomp}), we have
\begin{eqnarray*}
&& z_{il}z_{jl}-(\mathbb{E}z_{il})(\mathbb{E}z_{jl}) \\
&=& (\mathbb{E}z_{jl})(z_{il}-\mathbb{E}z_{il})+(\mathbb{E}z_{il})(z_{jl}-(\mathbb{E}z_{jl}))+(z_{il}-\mathbb{E}z_{il})(z_{jl}-(\mathbb{E}z_{jl})) \\
&=& k^{-1}(z_{il}-\mathbb{E}z_{il})+k^{-1}(z_{jl}-(\mathbb{E}z_{jl}))+(z_{il}-\mathbb{E}z_{il})(z_{jl}-(\mathbb{E}z_{jl})).
\end{eqnarray*}
Therefore,
\begin{eqnarray*}
&& \mathbb{E}\left(\sum_{i<j}z_{il}z_{jl}-(\mathbb{E}z_{il})(\mathbb{E}z_{jl})\right)^2 \\
&\leq& \frac{3}{4}\mathbb{E}\left(\sum_{i\neq j}k^{-1}(z_{il}-\mathbb{E}z_{il})\right)^2 + \frac{3}{4}\mathbb{E}\left(\sum_{i\neq j}k^{-1}(z_{jl}-\mathbb{E}z_{jl})\right)^2 \\
&& +\frac{3}{4}\mathbb{E}\left(\sum_{i\neq j}(z_{il}-\mathbb{E}z_{il})(z_{jl}-(\mathbb{E}z_{jl})\right)^2 \\
&=& O\left(\frac{n^3}{k^3}+\frac{n^2}{k^2}\right).
\end{eqnarray*}
Hence,
$$\mathbb{E}\left(\frac{1}{{n\choose 2}}\sum_{i<j}(\mathbb{I}\{Z_i=Z_j\}-\mathbb{P}(Z_i=Z_j))\right)^2=O\left(\frac{1}{nk}\right)=O\left(\frac{1}{n}\right).$$
\end{proof}

\begin{proposition}\label{prop:meixiangdao}
Under the same setting of Lemma \ref{lem:EV-order},
\begin{eqnarray*}
&& \mathbb{E}\left(\frac{1}{{n\choose 3}}\sum_{i<j<l}\frac{(W_i^2W_jW_l-\mathbb{E}W^2(\mathbb{E}W)^2)\eta_{ij}\eta_{il}}{3}\right.\\
&& \left. + \frac{(W_j^2W_iW_l-\mathbb{E}W^2(\mathbb{E}W)^2)\eta_{ij}\eta_{jl}+(W_l^2W_jW_i-\mathbb{E}W^2(\mathbb{E}W)^2)\eta_{jl}\eta_{il}}{3}\right)^2=O\left(\frac{a^4}{n}\right). 
\end{eqnarray*}
\end{proposition}
\begin{proof}
It is sufficient to bound
\begin{equation}
\mathbb{E}\left(\frac{1}{{n\choose 3}}\sum_{i<j<l}(W_i^2W_jW_l-\mathbb{E}W^2(\mathbb{E}W)^2)\eta_{ij}\eta_{il}\right)^2.\label{eq:911turbo}
\end{equation}
Observe the following decomposition
\begin{eqnarray*}
&& W_i^2W_jW_l-\mathbb{E}W^2(\mathbb{E}W)^2 \\
&=& (\mathbb{E}W)^2(W_i^2-\mathbb{E}W^2) + \mathbb{E}W^2(\mathbb{E}W)(W_j-\mathbb{E}W) + \mathbb{E}W^2(\mathbb{E}W)(W_l-\mathbb{E}W) \\
&& + (\mathbb{E}W^2)(W_j-\mathbb{E}W)(W_l-\mathbb{E}W) + (\mathbb{E}W)(W_i^2-\mathbb{E}W^2)(W_j-\mathbb{E}W) \\
&& + (\mathbb{E}W)(W_i^2-\mathbb{E}W^2)(W_l-\mathbb{E}W) + (W_i^2-\mathbb{E}W^2)(W_i^2-\mathbb{E}W^2)(W_l-\mathbb{E}W).
\end{eqnarray*}
Let $\xi_{ijl}$ be an arbitrary array such that $0\leq \xi_{ijl}\leq a^2$ for all $i,j,l$. Then,
$$\mathbb{E}\left(\sum_{i=1}^n\left(\sum_{\{(j,l):j,l\neq i\}}(\mathbb{E}W)^2\xi_{ijl}\right)(W_i^2-\mathbb{E}W^2)\right)^2=O\left(n^5a^4\right),$$
$$\mathbb{E}\left(\sum_{i=1}^n\left(\sum_{\{(j,l):j,l\neq i\}}(\mathbb{E}W)(\mathbb{E}W^2)\xi_{ijl}\right)(W_i-\mathbb{E}W)\right)^2=O\left(n^5a^4\right),$$
$$\mathbb{E}\left(\sum_{i\neq j}\left(\sum_{\{l:l\neq i,j\}}(\mathbb{E}W^2)\xi_{ijl}\right)(W_j-\mathbb{E}W)(W_i-\mathbb{E}W)\right)^2=O\left(n^4a^4\right),$$
$$\mathbb{E}\left(\sum_{i\neq j}\left(\sum_{\{l:l\neq i,j\}}(\mathbb{E}W)\xi_{ijl}\right)(W_j-\mathbb{E}W)(W_i^2-\mathbb{E}W^2)\right)^2=O\left(n^4a^4\right),$$
$$\mathbb{E}\left(\sum_{i<j<l}\xi_{ijl}(W_i^2-\mathbb{E}W^2)(W_i^2-\mathbb{E}W^2)(W_l-\mathbb{E}W)\right)=O(n^3a^4).$$
Therefore, (\ref{eq:911turbo}) is bounded by $O\left(\frac{a^4}{n}\right)$.
\end{proof}

\begin{proposition}\label{prop:z-turbo}
Under the same setting of Lemma \ref{lem:EV-order},
$$\Var\left(\frac{1}{{n\choose 3}}\sum_{i<j<l}\frac{\eta_{ij}\eta_{il}+\eta_{ij}\eta_{jl}+\eta_{il}\eta_{jl}}{3}\right)=O\left(\frac{a^4}{n}\right).$$
\end{proposition}
\begin{proof}
We first give a very general result. Define $z_{ih}=\mathbb{I}\{Z_i=h\}\sim\text{Bernoulli}(k^{-1})$. Let $H$ be a subset of $[k]^3$. Then,
\begin{eqnarray}
\label{eq:amg-c43} && \mathbb{E}\left(\sum_{i<j<l}\sum_{(h_1,h_2,h_3)\in H}(z_{ih_1}z_{jh_2}z_{lh_3}-\mathbb{E}z_{ih_1}\mathbb{E}z_{jh_2}\mathbb{E}z_{lh_3})\right)^2 \\
\label{eq:amg-c63s} &\leq& |H|\sum_{(h_1,h_2,h_3)\in H}\mathbb{E}\left(\sum_{i<j<l}(z_{ih_1}z_{jh_2}z_{lh_3}-\mathbb{E}z_{ih_1}\mathbb{E}z_{jh_2}\mathbb{E}z_{lh_3})\right)^2.
\end{eqnarray}
We observe the decomposition,
\begin{eqnarray*}
&& z_{ih_1}z_{jh_2}z_{lh_3}-\mathbb{E}z_{ih_1}\mathbb{E}z_{jh_2}\mathbb{E}z_{lh_3} \\
&=& k^{-2}(z_{ih_1}-k^{-1}) + k^{-2}(z_{jh_2}-k^{-1}) + k^{-2}(z_{lh_3}-k^{-1}) \\
&& + k^{-1}(z_{ih_1}-k^{-1})(z_{jh_2}-k^{-1}) + k^{-1}(z_{ih_1}-k^{-1})(z_{lh_3}-k^{-1}) + k^{-1}(z_{jh_3}-k^{-1})(z_{lh_2}-k^{-1}) \\
&& + (z_{ih_1}-k^{-1})(z_{jh_2}-k^{-1})(z_{lh_3}-k^{-1}).
\end{eqnarray*}
It is not hard to see
$$\mathbb{E}\left(\sum_{i<j<l}k^{-2}(z_{ih_1}-k^{-1})\right)^2=O\left(\frac{n^5}{k^5}\right),$$
$$\mathbb{E}\left(\sum_{i<j<l}k^{-1}(z_{ih_1}-k^{-1})(z_{jh_2}-k^{-1})\right)^2=O\left(\frac{n^4}{k^4}\right),$$
$$\mathbb{E}\left(\sum_{i<j<l}(z_{ih_1}-k^{-1})(z_{jh_2}-k^{-1})(z_{lh_3}-k^{-1})\right)^2=O\left(\frac{n^3}{k^3}\right).$$
By the inequality (\ref{eq:amg-c63s}),
\begin{equation}
\mathbb{E}\left(\sum_{i<j<l}\sum_{(h_1,h_2,h_3)\in H}(z_{ih_1}z_{jh_2}z_{lh_3}-\mathbb{E}z_{ih_1}\mathbb{E}z_{jh_2}\mathbb{E}z_{lh_3})\right)^2=O\left(|H|^2\frac{n^5}{k^5}\right).\label{eq:amg-gtr}
\end{equation}

Now we derive a bound for
$$\Var\left(\frac{1}{{n\choose 3}}\sum_{i<j<l}\eta_{ij}\eta_{il}\right).$$
By the definition of $\eta_{ij}$,
\begin{eqnarray*}
\Var\left(\frac{1}{{n\choose 3}}\sum_{i<j<l}\eta_{ij}\eta_{il}\right) &\leq& O(a^4)\Var\left(\frac{1}{{n\choose 3}}\sum_{i<j<l}\mathbb{I}\{Z_i=Z_j=Z_l\}\right) \\
&& + O(a^2b^2)\Var\left(\frac{1}{{n\choose 3}}\sum_{i<j<l}\mathbb{I}\{Z_j=Z_i\neq Z_l\}\right) \\
&& + O(a^2b^2)\Var\left(\frac{1}{{n\choose 3}}\sum_{i<j<l}\mathbb{I}\{Z_j\neq Z_i=Z_l\}\right) \\
&& + O(b^4)\Var\left(\frac{1}{{n\choose 3}}\sum_{i<j<l}\mathbb{I}\{Z_j\neq Z_i\neq Z_l\}\right),
\end{eqnarray*}
The four terms above are all special cases of (\ref{eq:amg-c43}). For example, $\mathbb{I}\{Z_j=Z_i\neq Z_l\}$ can be written as
\begin{equation}
\sum_{(h_1,h_2,h_3)\in H}z_{ih_1}z_{jh_2}z_{lh_3},\label{eq:general-H}
\end{equation}
where $H=\{(h_1,h_2,h_3)\in[k]^3: h_1=h_2=h_3\}$ and $|H|=k$. The indictor $\mathbb{I}\{Z_j=Z_i\neq Z_l\}$ can be written as (\ref{eq:general-H}) with an $H$ such that $|H|=O(k^2)$. The indicator $\mathbb{I}\{Z_j\neq Z_i=Z_l\}$  can also be written as (\ref{eq:general-H}) with an $H$ such that $|H|=O(k^2)$. Finally, note that
\begin{eqnarray*}
\mathbb{I}\{Z_j\neq Z_i\neq Z_l\} &=& 1-\mathbb{I}\{Z_i=Z_j=Z_l\} - \mathbb{I}\{Z_i\neq Z_j=Z_l\} \\
&& - \mathbb{I}\{Z_j=Z_i\neq Z_l\}.
\end{eqnarray*}
Thus, the variance of $\frac{1}{{n\choose 3}}\sum_{i<j<l}\mathbb{I}\{Z_j\neq Z_i\neq Z_l\}$ can be further decomposed according to the above equality. Each term in the decomposition can be represented by (\ref{eq:general-H}) with an $H$ such that $|H|=O(k^2)$. Hence, we obtain
$$\Var\left(\frac{1}{{n\choose 3}}\sum_{i<j<l}\eta_{ij}\eta_{il}\right)=O\left(\frac{a^4}{nk}\right)=O\left(\frac{a^4}{n}\right).$$
This completes the proof.
\end{proof}

\begin{proposition}\label{prop:wijk}
Under the same setting of Lemma \ref{lem:T-order},
$$\mathbb{E}\left(\frac{1}{{n\choose 3}}\sum_{i<j<l}\eta_{ij}\eta_{jl}\eta_{il}(W_i^2W_j^2W_l^2-(\mathbb{E}W^2)^3)\right)^2=O\left(\frac{a^6}{n}\right).$$
\end{proposition}
\begin{proof}
Observe the decomposition
\begin{eqnarray*}
&& W_i^2W_j^2W_l^2-(\mathbb{E}W^2)^3 \\
&=& (\mathbb{E}W^2)^2(W_i^2-\mathbb{E}W^2) + (\mathbb{E}W^2)^2(W_j^2-\mathbb{E}W^2) + (\mathbb{E}W^2)^2(W_l^2-\mathbb{E}W^2) \\
&& + (\mathbb{E}W^2)(W_i^2-\mathbb{E}W^2)(W_j^2-\mathbb{E}W^2) + (\mathbb{E}W^2)(W_i^2-\mathbb{E}W^2)(W_l^2-\mathbb{E}W^2) \\
&&+ (\mathbb{E}W^2)(W_l^2-\mathbb{E}W^2)(W_j^2-\mathbb{E}W^2)+ (W_i^2-\mathbb{E}W^2)(W_j^2-\mathbb{E}W^2)(W_l^2-\mathbb{E}W^2).
\end{eqnarray*}
Let $\xi_{ijl}$ be an arbitrary array such that $0\leq \xi_{ijl}\leq a^3$ for all $i,j,l$. Then,
$$\mathbb{E}\left(\sum_{i=1}^n\left(\sum_{\{(j,l):j,l\neq i\}}(\mathbb{E}W^2)^2\xi_{ijl}\right)(W_i^2-\mathbb{E}W^2)\right)^2=O\left(n^5a^6\right),$$
$$\mathbb{E}\left(\sum_{i\neq j}\left(\sum_{\{l:l\neq i,j\}}(\mathbb{E}W^2)\xi_{ijl}\right)(W_j^2-\mathbb{E}W^2)(W_i^2-\mathbb{E}W^2)\right)^2=O\left(n^4a^6\right),$$
$$\mathbb{E}\left(\sum_{i<j<l}\xi_{ijl}(W_i^2-\mathbb{E}W^2)(W_j^2-\mathbb{E}W^2)(W_l^2-\mathbb{E}W^2)\right)^2=O(n^3a^6).$$
Therefore, we obtain the bound $O\left(\frac{a^6}{n}\right)$.
\end{proof}

\begin{proposition}\label{prop:z-v8}
Under the same setting of Lemma \ref{lem:T-order},
$$\Var\left(\frac{1}{{n\choose 3}}\sum_{i<j<l}\eta_{ij}\eta_{jl}\eta_{il}\right)=O\left(\frac{a^6}{n}\right).$$
\end{proposition}
\begin{proof}
Since
\begin{eqnarray*}
\eta_{ij}\eta_{il}\eta_{jl} &=& \left(a\mathbb{I}\{Z_i=Z_j\}+b\mathbb{I}\{Z_i\neq Z_j\}\right) \\
&& \times \left(a\mathbb{I}\{Z_i=Z_l\}+b\mathbb{I}\{Z_i\neq Z_l\}\right) \\
&& \times \left(a\mathbb{I}\{Z_l=Z_j\}+b\mathbb{I}\{Z_l\neq Z_j\}\right),
\end{eqnarray*}
it can be expanded as
$$\eta_{ij}\eta_{il}\eta_{jl}=\sum_{t=1}^8\xi_t\sum_{(h_1,h_2,h_3)\in H_t}z_{ih_1}z_{jh_2}z_{lh_3},$$
where each $\xi_t$ takes value in $\{a^3,a^2b,ab^2,b^3\}$. For each $t\in[7]$, $|H_t|=O(k^2)$, and thus
\begin{equation}
\Var\left(\sum_{i<j<l}\sum_{(h_1,h_2,h_3)\in H_t}z_{ih_1}z_{jh_2}z_{lh_3}\right)=O\left(\frac{n^5}{k}\right),\label{eq:Ht1-7}
\end{equation}
by (\ref{eq:amg-gtr}). For $H_8$, we have
\begin{eqnarray*}
\sum_{(h_1,h_2,h_3)\in H_8}z_{ih_1}z_{jh_2}z_{lh_3} &=& \mathbb{I}\{Z_i\neq Z_j\neq Z_l\neq Z_i\} \\
&=& 1 - \mathbb{I}\{Z_i=Z_j=Z_l\} - \mathbb{I}\{Z_i\neq Z_j=Z_l\} \\
&& - \mathbb{I}\{Z_j=Z_i\neq Z_l\} - \mathbb{I}\{Z_i=Z_j\neq Z_l\}.
\end{eqnarray*}
Thus, the variance of $\frac{1}{{n\choose 3}}\sum_{i<j<l}\mathbb{I}\{Z_j\neq Z_i\neq Z_l\neq Z_j\}$ can be further decomposed according to the above equality. Each term in the decomposition can be represented by (\ref{eq:general-H}) with an $H$ such that $|H|=O(k^2)$. Therefore, (\ref{eq:Ht1-7}) also holds for $t=8$. Finally, we have the result
$$\Var\left(\frac{1}{{n\choose 3}}\sum_{i<j<l}\eta_{ij}\eta_{il}\eta_{jl}\right)=O\left(\frac{a^6}{nk}\right)=O\left(\frac{a^6}{n}\right).$$
This completes the proof.
\end{proof}

\begin{proposition}\label{prop:m-performance}
Under the same settings of Lemma \ref{lem:EV-NG} and Lemma \ref{lem:T-NG},
$$\mathbb{E}\left(\sum_{i<j}\left(A_{0i}A_{0j}-\left(\frac{rp}{k}\right)^2\right)\right)^2=O\left(\left(\frac{nrp}{k}\right)^3\right),$$
$$\mathbb{E}\left(\sum_{i<j<l}\left(A_{0i}A_{0j}A_{0l}-\left(\frac{rp}{k}\right)^3\right)\right)^2=O\left(\left(\frac{nrp}{k}\right)^5\right).$$
\end{proposition}
\begin{proof}
The proof is similar to those of Proposition \ref{prop:wij} and Proposition \ref{prop:wijk}. Similar to the argument in the proof of Proposition \ref{prop:wij}, to bound $\mathbb{E}\left(\sum_{i<j}\left(A_{0i}A_{0j}-\left(\frac{rp}{k}\right)^2\right)\right)^2$, it is sufficient to bound
$$\mathbb{E}\left(\sum_{i\neq j}\frac{rp}{k}\left(A_{0i}-\frac{rp}{k}\right)\right)^2,$$
and
$$\mathbb{E}\left(\sum_{i\neq j}\left(A_{0i}-\frac{rp}{k}\right)\left(A_{0j}-\frac{rp}{k}\right)\right)^2.$$
By the fact that $A_{0i}\sim\text{Bernoulli}\left(\frac{rp}{k}\right)$, the above two terms can be bounded by $O\left(\left(\frac{nrp}{k}\right)^3\right)$ and $O\left(\left(\frac{nrp}{k}\right)^2\right)$, respectively. Therefore,
$$\mathbb{E}\left(\sum_{i<j}\left(A_{0i}A_{0j}-\left(\frac{rp}{k}\right)^2\right)\right)^2=O\left(\left(\frac{nrp}{k}\right)^3\right),$$
if $\frac{nrp}{k}>1$.

Similar to the argument in the proof of Proposition \ref{prop:wijk}, to bound $\mathbb{E}\left(\sum_{i<j<l}\left(A_{0i}A_{0j}A_{0l}-\left(\frac{rp}{k}\right)^3\right)\right)^2$, it is sufficient to bound
$$\mathbb{E}\left(\sum_{i=1}^n\left(\sum_{\{(j,l):j,l\neq i\}}\left(\frac{rp}{k}\right)^2\right)\left(A_{0i}-\frac{rp}{k}\right)\right)^2,$$
$$\mathbb{E}\left(\sum_{i\neq j}\left(\sum_{\{l:l\neq i,j\}}\frac{rp}{k}\right)\left(A_{0i}-\frac{rp}{k}\right)\left(A_{0j}-\frac{rp}{k}\right)\right)^2,$$
$$\mathbb{E}\left(\sum_{i<j<l}\left(A_{0i}-\frac{rp}{k}\right)\left(A_{0j}-\frac{rp}{k}\right)\left(A_{0l}-\frac{rp}{k}\right)\right).$$
The three terms above can be bounded by $O\left(\left(\frac{nrp}{k}\right)^5\right)$, $O\left(\left(\frac{nrp}{k}\right)^4\right)$ and $O\left(\left(\frac{nrp}{k}\right)^3\right)$, respectively. Under the condition $\frac{nrp}{k}>1$, we have
$$\mathbb{E}\left(\sum_{i<j<l}\left(A_{0i}A_{0j}A_{0l}-\left(\frac{rp}{k}\right)^3\right)\right)^2=O\left(\left(\frac{nrp}{k}\right)^5\right).$$
\end{proof}

\begin{proposition}\label{prop:volvo}
Under the same setting of Lemma \ref{lem:EV-NG},
$$\mathbb{E}\left(\sum_{1\leq i<j\leq n}\tilde{\eta}_{ij}\left(W_iW_j-(\mathbb{E}W)^2\right)\right)^2=O\left(a^2r^3n^3/k^3\right).$$
\end{proposition}
\begin{proof}
We use the notation $R_i=\mathbb{I}\{Z_i\in\mathcal{R}\}\sim\text{Bernoulli}\left(\frac{r}{k}\right)$. Then, $\tilde{\eta}_{ij}=\eta_{ij}R_iR_j$. It is sufficient to bound the following two terms,
\begin{equation}
a^2\mathbb{E}\left(\sum_{i<j}\mathbb{I}\{Z_i=Z_j\}R_iR_j\left(W_iW_j-(\mathbb{E}W)^2\right)\right)^2,\label{eq:surp1}
\end{equation}
and
\begin{equation}
b^2\mathbb{E}\left(\sum_{i<j}\mathbb{I}\{Z_i\neq Z_j\}R_iR_j\left(W_iW_j-(\mathbb{E}W)^2\right)\right)^2,\label{eq:surp2}
\end{equation}
Note that $\mathbb{I}\{Z_i\neq Z_j\}=1-\mathbb{I}\{Z_i= Z_j\}$, (\ref{eq:surp2}) can be further bounded by the sum of the following two terms,
\begin{equation}
2b^2\mathbb{E}\left(\sum_{i<j}R_iR_j\left(W_iW_j-(\mathbb{E}W)^2\right)\right)^2,\label{eq:surp2.1}
\end{equation}
\begin{equation}
2b^2\mathbb{E}\left(\sum_{i<j}\mathbb{I}\{Z_i=Z_j\}R_iR_j\left(W_iW_j-(\mathbb{E}W)^2\right)\right)^2.\label{eq:surp2.2}
\end{equation}
Define $\tilde{z}_{ih}=\mathbb{I}\{Z_i=h,Z_i\in\mathcal{R}\}=z_{ih}R_i\sim\text{Bernoulli}(\mathbb{E}\tilde{z}_{ih})$, with $\mathbb{E}\tilde{z}_{ih}=k^{-1}$ if $h\in\mathcal{R}=[r]$, and $\mathbb{E}\tilde{z}_{ih}=0$, otherwise. To bound (\ref{eq:surp1}), we have
\begin{eqnarray*}
&& a^2\mathbb{E}\left(\sum_{i<j}\mathbb{I}\{Z_i=Z_j\}R_iR_j\left(W_iW_j-(\mathbb{E}W)^2\right)\right)^2 \\
&=& a^2\mathbb{E}\left(\sum_{i<j}\sum_{h=1}^r\tilde{z}_{ih}\tilde{z}_{jh}\left(W_iW_j-(\mathbb{E}W)^2\right)\right)^2 \\
&\leq& a^2r\sum_{h=1}^r\mathbb{E}\left(\sum_{i<j}\tilde{z}_{ih}\tilde{z}_{jh}\left(W_iW_j-(\mathbb{E}W)^2\right)\right)^2.
\end{eqnarray*}
Similarly (\ref{eq:surp2.2}) can be bounded by
$$2b^2r\sum_{h=1}^r\mathbb{E}\left(\sum_{i<j}\tilde{z}_{ih}\tilde{z}_{jh}\left(W_iW_j-(\mathbb{E}W)^2\right)\right)^2.$$
Therefore, it is essential to bound $\mathbb{E}\left(\sum_{i<j}\tilde{z}_{ih}\tilde{z}_{jh}\left(W_iW_j-(\mathbb{E}W)^2\right)\right)^2$ for $h\in[r]$. Observe the decomposition
\begin{eqnarray*}
&& \mathbb{E}\left(\sum_{i<j}\tilde{z}_{ih}\tilde{z}_{jh}\left(W_iW_j-(\mathbb{E}W)^2\right)\right)^2 \\
&=& \mathbb{E}\left(\sum_{i<j}(\tilde{z}_{ih}\tilde{z}_{jh}-k^{-2})\left(W_iW_j-(\mathbb{E}W)^2\right)\right)^2 + \mathbb{E}\left(\sum_{i<j}k^{-2}\left(W_iW_j-(\mathbb{E}W)^2\right)\right)^2,
\end{eqnarray*}
where the second term in the above decomposition can be bounded as $O\left(n^3k^{-4}\right)$. For the first term, $W_iW_j-(\mathbb{E}W)^2$ can be decomposed as what we have done in the proof of Proposition \ref{prop:wij}. We have a similar decomposition for $\tilde{z}_{ih}\tilde{z}_{jh}-k^{-2}$. As a result, the product $(\tilde{z}_{ih}\tilde{z}_{jh}-k^{-2})\left(W_iW_j-(\mathbb{E}W)^2\right)$ has the decomposition
\begin{eqnarray*}
&& (\tilde{z}_{ih}\tilde{z}_{jh}-k^{-2})\left(W_iW_j-(\mathbb{E}W)^2\right) \\
&=& \left((\mathbb{E}W)(W_i-\mathbb{E}W+W_j-\mathbb{E}W)+(W_i-\mathbb{E}W)(W_j-\mathbb{E}W)\right) \\
&& \times \left(k^{-1}(\tilde{z}_{il}-\mathbb{E}\tilde{z}_{il})+k^{-1}(\tilde{z}_{jl}-(\mathbb{E}\tilde{z}_{jl}))+(\tilde{z}_{il}-\mathbb{E}\tilde{z}_{il})(\tilde{z}_{jl}-(\mathbb{E}\tilde{z}_{jl}))\right) \\
&=& (\mathbb{E}W)k^{-1}(W_i-\mathbb{E}W)(\tilde{z}_{il}-\mathbb{E}\tilde{z}_{il}) + (\mathbb{E}W)k^{-1}(W_j-\mathbb{E}W)(\tilde{z}_{jl}-\mathbb{E}\tilde{z}_{jl}) + \cdots.
\end{eqnarray*}
We highlight the first two terms in the above expansion, and we only analyze the first term. The effects of other terms are negligible. We have
$$ \mathbb{E}\left(\sum_{i<j}(\mathbb{E}W)k^{-1}(W_i-\mathbb{E}W)(\tilde{z}_{il}-\mathbb{E}\tilde{z}_{il})\right)^2 = O\left(\left(\frac{n}{k}\right)^3\right).
$$
Thus, $\mathbb{E}\left(\sum_{i<j}\tilde{z}_{ih}\tilde{z}_{jh}\left(W_iW_j-(\mathbb{E}W)^2\right)\right)^2=O\left(\left(\frac{n}{k}\right)^3\right)$, and both (\ref{eq:surp1}) and (\ref{eq:surp2.2}) can be bounded by $O\left(a^2r^2\left(\frac{n}{k}\right)^3\right)$.

Finally, we still need to bound (\ref{eq:surp2.1}).  Using a similar argument that replaces $\tilde{z}_{ih}$ by $R_i$, we can bound this term by $O\left(a^2r^3n^3/k^3\right)$.
\end{proof}

\begin{proposition}\label{prop:polestar}
Under the same setting of Lemma \ref{lem:EV-NG},
\begin{eqnarray*}
&& \mathbb{E}\left(\sum_{i<j<l}\frac{(W_i^2W_jW_l-\mathbb{E}W^2(\mathbb{E}W)^2)\tilde{\eta}_{ij}\tilde{\eta}_{il}}{3}\right.\\
&& \left. + \frac{(W_j^2W_iW_l-\mathbb{E}W^2(\mathbb{E}W)^2)\tilde{\eta}_{ij}\tilde{\eta}_{jl}+(W_l^2W_jW_i-\mathbb{E}W^2(\mathbb{E}W)^2)\tilde{\eta}_{jl}\tilde{\eta}_{il}}{3}\right)^2=O(a^4r^5n^5k^{-5}). 
\end{eqnarray*}
\end{proposition}
\begin{proof}
It is sufficient to bound $\mathbb{E}\left(\sum_{i<j<l}(W_i^2W_jW_l-\mathbb{E}W^2(\mathbb{E}W)^2)\tilde{\eta}_{ij}\tilde{\eta}_{il}\right)^2$.
Recall the notations $R_i$, $z_{ih}$ and $\tilde{z}_{ih}$ in previous proofs. First, we realize that $\eta_{ij}\eta_{jl}$ can be decomposed into four terms. For example, the first term is $a^2\mathbb{I}\{Z_i=Z_j=Z_l\}$. Then, $\mathbb{I}\{Z_i=Z_j=Z_l\}R_iR_jR_l$ can be written as
\begin{equation}
\sum_{(h_1,h_2,h_3)\in H}\tilde{z}_{ih_1}\tilde{z}_{jh_2}\tilde{z}_{lh_3},\label{eq:zzzt}
\end{equation}
for some $H$. This is also the case for $\mathbb{I}\{Z_i\neq Z_j=Z_l\}R_iR_jR_l$ and $\mathbb{I}\{Z_i=Z_j\neq Z_l\}R_iR_jR_l$. Each of the three terms can be represented as (\ref{eq:zzzt}) with some $H$ such that $|H|\leq r^2$. The last term $\mathbb{I}\{Z_i\neq Z_j\neq Z_l\}$ can be analyzed by the relation $1-\mathbb{I}\{Z_i=Z_j=Z_l\}-\mathbb{I}\{Z_i=Z_j\neq Z_l\}-\mathbb{I}\{Z_i\neq Z_j=Z_l\}$. Therefore, the following two terms determine the order of the bound,
\begin{equation}
O(a^4)|H|\sum_{(h_1,h_2,h_3)\in H}\mathbb{E}\left(\sum_{i<j<l}\tilde{z}_{ih_1}\tilde{z}_{jh_2}\tilde{z}_{lh_3}\left(W_i^2W_jW_l-(\mathbb{E}W^2)(\mathbb{E}W)^2\right)\right)^2,\label{eq:golf-gti}
\end{equation}
\begin{equation}
O(a^4)\mathbb{E}\left(\sum_{i<j<l}R_iR_jR_l\left(W_i^2W_jW_l-(\mathbb{E}W^2)(\mathbb{E}W)^2\right)\right)^2.\label{eq:golf-r}
\end{equation}

We analyze $\mathbb{E}\left(\sum_{i<j<l}\tilde{z}_{ih_1}\tilde{z}_{jh_2}\tilde{z}_{lh_3}\left(W_i^2W_jW_l-(\mathbb{E}W^2)(\mathbb{E}W)^2\right)\right)^2$. It can be further decomposed into the following two terms,
$$\mathbb{E}\left(\sum_{i<j<l}k^{-3}\left(W_i^2W_jW_l-(\mathbb{E}W^2)(\mathbb{E}W)^2\right)\right)^2,$$
$$\mathbb{E}\left(\sum_{i<j<l}(\tilde{z}_{ih_1}\tilde{z}_{jh_2}\tilde{z}_{lh_3}-k^{-3})\left(W_i^2W_jW_l-(\mathbb{E}W^2)(\mathbb{E}W)^2\right)\right)^2,$$
where the first term can be bounded using the same argument in the proof of Proposition \ref{prop:meixiangdao}, which leads to the order $O\left(n^5k^{-6}\right)$. The detailed analysis of the second term is lengthy. The idea is to study the expansion
\begin{eqnarray*}
&& (\tilde{z}_{ih_1}\tilde{z}_{jh_2}\tilde{z}_{lh_3}-k^{-3})\left(W_i^2W_jW_l-(\mathbb{E}W^2)(\mathbb{E}W)^2\right) \\
&=& k^{-2}(\mathbb{E}W)^2(\tilde{z}_{ih_1}-k^{-1})(W_i^2-(\mathbb{E}W^2)) + \cdots.
\end{eqnarray*}
We only highlight the first term in the expansion. Its contribution is through
$$\mathbb{E}\left(\sum_{i<j<l}k^{-2}(\mathbb{E}W)^2(\tilde{z}_{ih_1}-k^{-1})(W_i^2-(\mathbb{E}W^2))\right)^2=O(n^5k^{-5}).$$
One can similarly analyze each term in the expansion, and the overall bound is of order $O(n^5k^{-5})$. The same analysis also applies to $\mathbb{E}\left(\sum_{i<j<l}R_iR_jR_l\left(W_i^2W_jW_l-(\mathbb{E}W^2)(\mathbb{E}W)^2\right)\right)^2$. The only difference is that $R_i\sim\text{Bernoulli}(r/k)$ compared with $\tilde{z}_{ih}\sim\text{Bernoulli}(k^{-1})$. Therefore, we can obtain a bound of order $O\left(n^5(r/k)^5\right)$. Finally, these bounds imply that (\ref{eq:golf-gti}) and (\ref{eq:golf-r}) have bounds $O(a^4r^4n^5k^{-5})$ and $O(a^4r^5n^5k^{-5})$, respectively. The proof is complete by realizing that $O(a^4r^5n^5k^{-5})$ is the dominating order.
\end{proof}

\begin{proposition}\label{prop:glorymu}
Under the same setting of Lemma \ref{lem:T-NG},
$$\mathbb{E}\left(\sum_{i<j<l}\tilde{\eta}_{ij}\tilde{\eta}_{jl}\tilde{\eta}_{il}(W_i^2W_j^2W_l^2-(\mathbb{E}W^2)^3)\right)^2=O(a^6r^5n^5k^{-5}).$$
\end{proposition}
\begin{proof}
This proof is very similar to that of Proposition \ref{prop:polestar}. Similar to the arguments used there, we also need to analyze two terms that are analogous to (\ref{eq:golf-gti}) and (\ref{eq:golf-r}). These two corresponding terms are
$$O(a^6)|H|\sum_{(h_1,h_2,h_3)\in H}\mathbb{E}\left(\sum_{i<j<l}\tilde{z}_{ih_1}\tilde{z}_{jh_2}\tilde{z}_{lh_3}\left(W_iW_jW_l-(\mathbb{E}W)^3\right)\right)^2,$$
$$O(a^6)\mathbb{E}\left(\sum_{i<j<l}R_iR_jR_l\left(W_iW_jW_l-(\mathbb{E}W)^3\right)\right)^2.$$
These two terms can be analyzed in the exactly same way as those for (\ref{eq:golf-gti}) and (\ref{eq:golf-r}).
They can be bounded by $O(a^6r^4n^5k^{-5})$ and $O(a^6r^5n^5k^{-5})$, respectively. Therefore, $O(a^6r^5n^5k^{-5})$ is the overall bound.
\end{proof}

\begin{proposition}\label{prop:Sij}
Under the same condition of Lemma \ref{lem:EV-NG},
$$\mathbb{E}\left(\sum_{i<j}\tilde{A}_{ij}(S_iS_j-p^2)\right)^2=O\left(a^2\left(\frac{npr}{k}\right)^3\right).$$
\end{proposition}
\begin{proof}
We decompose $\tilde{A}_{ij}$ as the sum of $\tilde{A}_{ij}-\tilde{\theta}_{ij}$ and $\tilde{\theta}_{ij}$. Then,
\begin{eqnarray*}
&& \mathbb{E}\left(\sum_{i<j}(\tilde{A}_{ij}-\tilde{\theta}_{ij})(S_iS_j-p^2)\right)^2 \\
&\leq& \sum_{i<j}\mathbb{E}[(S_iS_j-p^2)^2\tilde{\theta}_{ij}] \\
&=& O\left(a\left(\frac{npr}{k}\right)^2\right).
\end{eqnarray*}
Next, we study $\mathbb{E}\left(\sum_{i<j}\tilde{\theta}_{ij}(S_iS_j-p^2)\right)^2$, where $\tilde{\theta}_{ij}=\eta_{ij}W_iW_jR_iR_j$. With the same argument used in the proof of Proposition \ref{prop:volvo}, it is sufficient to bound the following two terms,
\begin{equation}
O(a^2)r\sum_{h=1}^r\mathbb{E}\left(\sum_{i<j}\tilde{z}_{ih}\tilde{z}_{jh}W_iW_j(S_iS_j-p^2)\right)^2,\label{eq:lukaku}
\end{equation}
\begin{equation}
O(a^2)\mathbb{E}\left(\sum_{i<j}R_iR_jW_iW_j(S_iS_j-p^2)\right)^2.\label{eq:pogba}
\end{equation}
We use the notation $\tilde{z}_{ih}W_i=\bar{z}_{ih}$ and $R_iW_i=\bar{R}_i$. Then, $\mathbb{E}\tilde{z}_{ih}=O(k^{-1})$, $\Var(\tilde{z}_{ih})=O(k^{-1})$, $\mathbb{E}\bar{R}_i=O(r/k)$ and $\Var(\bar{R}_i)=O(r/k)$. Then,  (\ref{eq:lukaku}) and (\ref{eq:pogba}) are of the same forms that we have already analyzed in the proof of Proposition \ref{prop:volvo}. Here, we have $\bar{z}_{ih}, S_i, p, \bar{R}_i$ instead of $\tilde{z}_{ih}, W_i,\mathbb{E}W, R_i$ in the proof of Proposition \ref{prop:volvo}. Using the same argument there, both (\ref{eq:lukaku}) and (\ref{eq:pogba}) can be bounded by $O\left(a^2\left(\frac{npr}{k}\right)^3\right)$, when $npr/k>1$. Finally, when $npra/k>1$, we have $\mathbb{E}\left(\sum_{i<j}\tilde{A}_{ij}(S_iS_j-p^2)\right)^2=O\left(a^2\left(\frac{npr}{k}\right)^3\right)$.
\end{proof}

\begin{proposition}\label{prop:v-tec}
Under the same setting of Lemma \ref{lem:EV-NG},
$$\mathbb{E}\left(\sum_{i<j<l}(\tilde{A}_{ij}\tilde{A}_{il}+\tilde{A}_{ij}\tilde{A}_{jl}+\tilde{A}_{jl}\tilde{A}_{il})(S_iS_jS_l-p^3)\right)^2=O\left(a^4\left(\frac{npr}{k}\right)^5\right).$$
\end{proposition}
\begin{proof}
We analyze $\mathbb{E}\left(\sum_{i<j<l}\tilde{A}_{ij}\tilde{A}_{il}(S_iS_jS_l-p^3)\right)^2$. Decompose $\tilde{A}_{ij}\tilde{A}_{il}$ as the sum of $\tilde{A}_{ij}\tilde{A}_{il}-\tilde{\theta}_{ij}\tilde{\theta}_{il}$ and $\tilde{\theta}_{ij}\tilde{\theta}_{il}$, and we first analyze $\mathbb{E}\left(\sum_{i<j<l}(\tilde{A}_{ij}\tilde{A}_{il}-\tilde{\theta}_{ij}\tilde{\theta}_{il})(S_iS_jS_l-p^3)\right)^2$. Use the decomposition
$$\tilde{A}_{ij}\tilde{A}_{il}-\tilde{\theta}_{ij}\tilde{\theta}_{il}=\tilde{\theta}_{il}(\tilde{A}_{ij}-\tilde{\theta}_{ij})+\tilde{\theta}_{ij}(\tilde{A}_{il}-\tilde{\theta}_{il})+(\tilde{A}_{ij}-\tilde{\theta}_{ij})(\tilde{A}_{il}-\tilde{\theta}_{il}).$$
Then, we have
\begin{eqnarray*}
&& \mathbb{E}\left(\sum_{i<j, l\neq i,j}\tilde{\theta}_{il}(\tilde{A}_{ij}-\tilde{\theta}_{ij})(S_iS_jS_l-p^3)\right)^2 \\
&=& \mathbb{E}\left(\sum_{i<j}\left(\sum_{\{l:l\neq i,j\}}\tilde{\theta}_{il}(S_iS_jS_l-p^3)\right)(\tilde{A}_{ij}-\tilde{\theta}_{ij})\right)^2 \\
&\leq& \sum_{i<j}\mathbb{E}\left(\sum_{\{l:l\neq i,j\}}\tilde{\theta}_{il}(S_iS_jS_l-p^3)\right)^2\tilde{\theta}_{ij} \\
&=& \sum_{i<j}\mathbb{E}\left(\sum_{\{l:l\neq i,j\}}\tilde{\theta}_{il}(S_iS_jS_l-S_iS_jp)\right)^2\tilde{\theta}_{ij} + \sum_{i<j}\mathbb{E}\left(\sum_{\{l:l\neq i,j\}}\tilde{\theta}_{il}(S_iS_jp-p^3)\right)^2\tilde{\theta}_{ij} \\
&\leq& \sum_{i<j}\sum_{\{l:l\neq i,j\}}\mathbb{E}\tilde{\theta}_{il}^2\tilde{\theta}_{ij}S_iS_jp + \sum_{i<j}n^2p^2\mathbb{E}\tilde{\theta}_{il}^2\tilde{\theta}_{ij}(S_iS_j-p^2)^2 \\
&=& O\left(\left(\frac{anpr}{k}\right)^3\right),
\end{eqnarray*}
and
\begin{eqnarray*}
&& \mathbb{E}\left(\sum_{i<j<l}(\tilde{A}_{ij}-\tilde{\theta}_{ij})(\tilde{A}_{il}-\tilde{\theta}_{il})(S_iS_jS_l-p^3)\right)^2 \\
&\leq& \sum_{i<j<l}\mathbb{E}(S_iS_jS_l-p^3)^2\tilde{\theta}_{ij}\tilde{\theta}_{il} \\
&=& O\left(a^2\left(\frac{npr}{k}\right)^3\right).
\end{eqnarray*}
Therefore,
$$\mathbb{E}\left(\sum_{i<j<l}(\tilde{A}_{ij}\tilde{A}_{il}-\tilde{\theta}_{ij}\tilde{\theta}_{il})(S_iS_jS_l-p^3)\right)^2=O\left(a^2\left(\frac{npr}{k}\right)^3\right).$$

Next, we study $\mathbb{E}\left(\sum_{i<j<l}\tilde{\theta}_{ij}\tilde{\theta}_{il}(S_iS_jS_l-p^3)\right)^2$. Note that $\tilde{\theta}_{ij}\tilde{\theta}_{il}=\eta_{ij}\eta_{il}W_i^2W_jW_lR_iR_jR_l$. With the same argument used in the proof of Proposition \ref{prop:polestar}, it is sufficient to bound the following two terms,
\begin{equation}
O(a^4)|H|\sum_{(h_1,h_2,h_3)\in H}\mathbb{E}\left(\sum_{i<j<l}\tilde{z}_{ih_1}\tilde{z}_{jh_2}W_i^2W_jW_l\tilde{z}_{lh_3}\left(S_iS_jS_l-p^3\right)\right)^2,\label{eq:matic}
\end{equation}
\begin{equation}
O(a^4)\mathbb{E}\left(\sum_{i<j<l}R_iR_jR_lW_i^2W_jW_l\left(S_iS_jS_l-p^3\right)\right)^2.\label{eq:mata}
\end{equation}
We use the notation $\tilde{z}_{ih}W_i=\bar{z}_{ih}$ and $R_iW_i=\bar{R}_i$,  (\ref{eq:matic}) and (\ref{eq:mata}) are of the same forms that we have already analyzed in the proof of Proposition \ref{prop:polestar}. Here, we have $\bar{z}_{ih}, S_i, p, \bar{R}_i$ instead of $\tilde{z}_{ih}, W_i,\mathbb{E}W, R_i$ in the proof of Proposition \ref{prop:polestar}. Using the same argument there, both (\ref{eq:matic}) and (\ref{eq:mata}) can be bounded by $O\left(a^4\left(\frac{npr}{k}\right)^5\right)$, when $npr/k>1$. Finally, when $npra/k>1$, we obtain the overall bound $O\left(a^4\left(\frac{npr}{k}\right)^5\right)$.
\end{proof}

\begin{proposition}\label{prop:cls63}
Under the same setting of Lemma \ref{lem:T-NG},
$$\mathbb{E}\left(\sum_{i<j<l}\tilde{A}_{ij}\tilde{A}_{il}\tilde{A}_{jl}(S_iS_jS_l-p^3)\right)^2=O\left(a^6\left(\frac{npr}{k}\right)^5\right).$$
\end{proposition}
\begin{proof}
We decompose $\tilde{A}_{ij}\tilde{A}_{il}\tilde{A}_{jl}$ as the sum of $\tilde{A}_{ij}\tilde{A}_{il}\tilde{A}_{jl}-\tilde{\theta}_{ij}\tilde{\theta}_{il}\tilde{\theta}_{jl}$ and $\tilde{\theta}_{ij}\tilde{\theta}_{il}\tilde{\theta}_{jl}$. We first bound $\mathbb{E}\left(\sum_{i<j<l}(\tilde{A}_{ij}\tilde{A}_{il}\tilde{A}_{jl}-\tilde{\theta}_{ij}\tilde{\theta}_{il}\tilde{\theta}_{jl})(S_iS_jS_l-p^3)\right)^2$. Use the decomposition
\begin{eqnarray*}
&& \tilde{A}_{ij}\tilde{A}_{il}\tilde{A}_{jl}-\tilde{\theta}_{ij}\tilde{\theta}_{il}\tilde{\theta}_{jl} \\
&=& \tilde{\theta}_{il}\tilde{\theta}_{jl}(\tilde{A}_{ij}-\tilde{\theta}_{ij}) + \tilde{\theta}_{ij}\tilde{\theta}_{jl}(\tilde{A}_{il}-\tilde{\theta}_{il}) + \tilde{\theta}_{ij}\tilde{\theta}_{il}(\tilde{A}_{jl}-\tilde{\theta}_{jl}) \\
&& + \tilde{\theta}_{jl}(\tilde{A}_{ij}-\tilde{\theta}_{ij})(\tilde{A}_{il}-\tilde{\theta}_{il}) + \tilde{\theta}_{il}(\tilde{A}_{ij}-\tilde{\theta}_{ij})(\tilde{A}_{jl}-\tilde{\theta}_{jl}) + \tilde{\theta}_{ij}(\tilde{A}_{il}-\tilde{\theta}_{il})(\tilde{A}_{jl}-\tilde{\theta}_{jl}) \\
&& + (\tilde{A}_{jl}-\tilde{\theta}_{jl})(\tilde{A}_{ij}-\tilde{\theta}_{ij})(\tilde{A}_{il}-\tilde{\theta}_{il}).
\end{eqnarray*}
Then, we have
\begin{eqnarray*}
&& \mathbb{E}\left(\sum_{i<j, l\neq i,j}\tilde{\theta}_{il}\tilde{\theta}_{jl}(\tilde{A}_{ij}-\tilde{\theta}_{ij})(S_iS_jS_l-p^3)\right)^2 \\
&=& \mathbb{E}\left(\sum_{i<j}\left(\sum_{\{l:l\neq i,j\}}\tilde{\theta}_{il}\tilde{\theta}_{jl}(S_iS_jS_l-p^3)\right)(\tilde{A}_{ij}-\tilde{\theta}_{ij})\right)^2 \\
&\leq& \sum_{i<j}\mathbb{E}\left(\sum_{\{l:l\neq i,j\}}\tilde{\theta}_{il}\tilde{\theta}_{jl}(S_iS_jS_l-p^3)\right)^2\tilde{\theta}_{ij} \\
&=& \sum_{i<j}\mathbb{E}\left(\sum_{\{l:l\neq i,j\}}\tilde{\theta}_{il}\tilde{\theta}_{jl}(S_iS_jS_l-S_iS_jp)\right)^2\tilde{\theta}_{ij} + \sum_{i<j}\mathbb{E}\left(\sum_{\{l:l\neq i,j\}}\tilde{\theta}_{il}\tilde{\theta}_{jl}(S_iS_jp-p^3)\right)^2\tilde{\theta}_{ij} \\
&\leq& \sum_{i<j}\sum_{\{l:l\neq i,j\}}\mathbb{E}\tilde{\theta}_{il}^2\tilde{\theta}_{jl}^2\tilde{\theta}_{ij}S_iS_jp + \sum_{i<j}n^2p^2\mathbb{E}\tilde{\theta}_{il}^2\tilde{\theta}_{jl}^2\tilde{\theta}_{ij}(S_iS_j-p^2)^2 \\
&=& O\left(a^5\left(\frac{npr}{k}\right)^3\right),
\end{eqnarray*}
\begin{eqnarray*}
&& \mathbb{E}\left(\sum_{i<j<l}\tilde{\theta}_{jl}(\tilde{A}_{ij}-\tilde{\theta}_{ij})(\tilde{A}_{il}-\tilde{\theta}_{il})(S_iS_jS_l-p^3)\right)^2 \\
&\leq& \sum_{i<j<l}\mathbb{E}\tilde{\theta}_{jl}^2(S_iS_jS_l-p^3)^2\tilde{\theta}_{ij}\tilde{\theta}_{il} \\
&=& O\left(a^4\left(\frac{npr}{k}\right)^3\right),
\end{eqnarray*}
and
\begin{eqnarray*}
&& \mathbb{E}\left(\sum_{i<j<l}(\tilde{A}_{jl}-\tilde{\theta}_{jl})(\tilde{A}_{ij}-\tilde{\theta}_{ij})(\tilde{A}_{il}-\tilde{\theta}_{il})(S_iS_jS_l-p^3)\right)^2 \\
&\leq& \sum_{i<j<l}\mathbb{E}(S_iS_jS_l-p^3)^2\tilde{\theta}_{ij}\tilde{\theta}_{il}\tilde{\theta}_{jl} \\
&=& O\left(\left(\frac{anpr}{k}\right)^3\right).
\end{eqnarray*}
Therefore,
$$\mathbb{E}\left(\sum_{i<j<l}(\tilde{A}_{ij}\tilde{A}_{il}\tilde{A}_{jl}-\tilde{\theta}_{ij}\tilde{\theta}_{il}\tilde{\theta}_{jl})(S_iS_jS_l-p^3)\right)^2=O\left(\left(\frac{anpr}{k}\right)^3\right).$$

Next, we study $\mathbb{E}\left(\sum_{i<j<l}\tilde{\theta}_{ij}\tilde{\theta}_{il}\tilde{\theta}_{jl}(S_iS_jS_l-p^3)\right)^2$. Note that $\tilde{\theta}_{ij}\tilde{\theta}_{il}\tilde{\theta}_{jl}=\eta_{ij}\eta_{jl}\eta_{il}W_iW_jW_lR_iR_jR_l$. With the same argument used in the proof of Proposition \ref{prop:glorymu}, it is sufficient to bound the following two terms
\begin{equation}
O(a^6)|H|\sum_{(h_1,h_2,h_3)\in H}\mathbb{E}\left(\sum_{i<j<l}\tilde{z}_{ih_1}\tilde{z}_{lh_3}\tilde{z}_{jh_2}W_iW_jW_l\left(S_iS_jS_l-p^3\right)\right)^2,\label{eq:marshall}
\end{equation}
\begin{equation}
O(a^6)\mathbb{E}\left(\sum_{i<j<l}R_iR_jR_lW_iW_jW_l\left(S_iS_jS_l-p^3\right)\right)^2.\label{eq:rashford}
\end{equation}
We use the notation $\tilde{z}_{ih}W_i=\bar{z}_{ih}$ and $R_iW_i=\bar{R}_i$,  (\ref{eq:marshall}) and (\ref{eq:rashford}) are of the same forms that we have already analyzed in the proof of Proposition \ref{prop:glorymu}. Here, we have $\bar{z}_{ih}, S_i, p, \bar{R}_i$ instead of $\tilde{z}_{ih}, W_i,\mathbb{E}W, R_i$ in the proof of Proposition \ref{prop:glorymu}. Using the same argument there,  both (\ref{eq:marshall}) and (\ref{eq:rashford}) can be bounded by $O\left(a^6\left(\frac{npr}{k}\right)^5\right)$, when $npr/k>1$. Finally, when $npra/k>1$, we obtain the overall bound $O\left(a^6\left(\frac{npr}{k}\right)^5\right)$.
\end{proof}

\begin{proposition}\label{prop:civic-si}
Under the same setting of Lemma \ref{lem:EV-NG},
$$\mathbb{E}\left(\sum_{1\leq i<j\leq n}(\tilde{\eta}_{ij}-\mathbb{E}\tilde{\eta}_{ij})\right)^2=O\left(a^2\left(\frac{nr}{k}\right)^3\right).$$
\end{proposition}
\begin{proof}
Since
$$\tilde{\eta}_{ij}=a\mathbb{I}\{Z_i=Z_j\}-b\mathbb{I}\{Z_i=Z_j\}+bR_iR_j,$$
we use a similar argument in the proof of Proposition \ref{prop:z-foundation}, and
it is sufficient to bound
$$O(a^2)r\sum_{h=1}^r\mathbb{E}\left(\sum_{i<j}\tilde{z}_{ih}\tilde{z}_{jh}-(\mathbb{E}\tilde{z}_{ih})(\mathbb{E}\tilde{z}_{jh})\right)^2,$$
and
$$O(a^2)\mathbb{E}\left(\sum_{i<j}R_iR_j-(r/k)^2\right)^2.$$
By the argument in the proof of Proposition \ref{prop:wij}, these two terms can be bounded by $O\left(a^2r^2\left(\frac{n}{k}\right)^3\right)$ and $O\left(a^2\left(\frac{nr}{k}\right)^3\right)$, respectively.
\end{proof}

\begin{proposition}\label{prop:civic-type-R}
Under the same setting of Lemma \ref{lem:EV-NG},
$$\Var\left(\sum_{i<j<l}\frac{\tilde{\eta}_{ij}\tilde{\eta}_{il}+\tilde{\eta}_{ij}\tilde{\eta}_{jl}+\tilde{\eta}_{il}\tilde{\eta}_{jl}}{3}\right)=O\left(a^4\left(\frac{nr}{k}\right)^5\right).$$
\end{proposition}
\begin{proof}
We use a similar argument in the proof of Proposition \ref{prop:z-turbo}, and it is sufficient to bound
$$O(a^4)|H|\sum_{(h_1,h_2,h_3)\in H}\mathbb{E}\left(\sum_{i<j<l}\tilde{z}_{ih_1}\tilde{z}_{jh_2}\tilde{z}_{lh_3}-\mathbb{E}(\tilde{z}_{ih_1}\tilde{z}_{jh_2}\tilde{z}_{lh_3})\right)^2,$$
and
$$O(a^4)\mathbb{E}\left(\sum_{i<j<l}R_iR_jR_l-(r/k)^3\right)^2.$$
With the same argument in the proof of Proposition \ref{prop:z-turbo}, these two terms can be bounded by $O\left(a^4r^4\left(\frac{n}{k}\right)^5\right)$ and $O\left(a^4\left(\frac{nr}{k}\right)^5\right)$.
\end{proof}

\begin{proposition}\label{prop:E63}
Under the same setting of Lemma \ref{lem:T-NG},
$$\Var\left(\sum_{i<j<l}\tilde{\eta}_{ij}\tilde{\eta}_{jl}\tilde{\eta}_{il}\right)=O\left(a^6\left(\frac{nr}{k}\right)^5\right).$$
\end{proposition}
\begin{proof}
We use a similar argument in the proof of Proposition \ref{prop:z-v8}, and it is sufficient to bound
$$O(a^6)|H|\sum_{(h_1,h_2,h_3)\in H}\mathbb{E}\left(\sum_{i<j<l}\tilde{z}_{ih_1}\tilde{z}_{jh_2}\tilde{z}_{lh_3}-\mathbb{E}(\tilde{z}_{ih_1}\tilde{z}_{jh_2}\tilde{z}_{lh_3})\right)^2,$$
and
$$O(a^6)\mathbb{E}\left(\sum_{i<j<l}R_iR_jR_l-(r/k)^3\right)^2.$$
With the same argument in the proof of Proposition \ref{prop:z-v8}, these two terms can be bounded by $O\left(a^6r^4\left(\frac{n}{k}\right)^5\right)$ and $O\left(a^6\left(\frac{nr}{k}\right)^5\right)$.
\end{proof}

\end{document}